\documentclass[10pt,a4paper]{article}

\usepackage{caption}
\usepackage[font=footnotesize,labelformat=simple]{subcaption}
\usepackage{amssymb,amsmath,graphicx,color,amsthm} 

\usepackage[authoryear]{natbib}
\usepackage{dcolumn} 

\usepackage{float} 
\usepackage{bbm}

\usepackage{bbding}
\usepackage{pifont}
\usepackage{wasysym}
\usepackage{amsfonts}

\usepackage{xr-hyper}
\usepackage[colorlinks=true,linkcolor=blue,citecolor=red]{hyperref}
\newcommand*{\addFileDependency}[1]{
\typeout{(#1)}
\@addtofilelist{#1}
\IfFileExists{#1}{}{\typeout{No file #1.}}
}\makeatother

\usepackage[ruled,vlined]{algorithm2e}

\usepackage{bm}

\usepackage{lscape}
\usepackage{pgf}
\usepackage{verbatim}
\usepackage{multirow}
\usepackage{enumerate}
\usepackage{comment}
\usepackage[utf8]{inputenc} 
\usepackage[T1]{fontenc}    
\usepackage{hyphenat}

\usepackage[a4paper]{geometry}
\newtheorem{teo}{Theorem}[section]

\newtheorem{cor}[teo]{Corollary}

\title{\textbf{Imputation of missing data using multivariate Gaussian Linear Cluster-Weighted Modeling}}
\author{Luis Alejandro Masmela-Caita\footnote{Universidad Distrital Francisco José de Caldas. Bogotá - Colombia. e-mail: lmasmela@udistrital.edu.co}\\
	Thais Paiva Galletti\footnote{Universidade Federal de Minas Gerais. Belo Horizonte - MG - Brasil. e-mail: thaispaiva@est.ufmg.br}\\
	Marcos Oliveira Prates\footnote{Universidade Federal de Minas Gerais. Belo Horizonte - MG - Brasil. e-mail: marcosop@est.ufmg.br}
}

\date{}

\begin{document}

\maketitle

\begin{abstract}

Missing data arises when certain values are not recorded or observed for variables of interest. However, most of the statistical theory assume complete data availability. To address incomplete databases, one approach is to fill the gaps corresponding to the missing information based on specific criteria, known as imputation. In this study, we propose a novel imputation methodology for databases with non-response units by leveraging additional information from fully observed auxiliary variables. We assume that the variables included in the database are continuous and that the auxiliary variables, which are fully observed, help to improve the imputation capacity of the model.
Within a fully Bayesian framework, our method utilizes a flexible mixture of multivariate normal distributions to jointly model the response and auxiliary variables. By employing the principles of Gaussian Cluster-Weighted modeling, we construct a predictive model to impute the missing values by leveraging information from the covariates. We present simulation studies and a real data illustration to demonstrate the imputation capacity of our method across various scenarios, comparing it to other methods in the literature.\\


\end{abstract}

\section{Introduction}
\label{sec:1}

In statistics, the desirable characteristics of estimators are often based on the assumption of complete data. However, in the presence of missing data, these desirable properties may not hold, and special techniques such as imputation may be required to ensure the use of regular statistical methods. \citep{rubin1976inference, rubin1987multiple, little2019statistical}.

Finite mixture models offer a flexible approach to statistical modeling of various random phenomena. \citet{di2007imputation} propose a finite mixture of multivariate Gaussian distributions for imputing missing data utilizing the Expectation-Maximization (EM) algorithm. Two imputation  alternatives are used to compare against their method, the conditional mean imputation and random selection. In particular, the authors investigate its performance in terms of preservation of the mean and the variance-covariance structure of the observed data. \citet{liao2007quadratically} present the Quadratically Gated Mixture of Experts (QGME) for classification of incomplete data, using the EM algorithm for joint likelihood maximization and adaptive imputation in step E. In the context of general regression problems with missing data, \citet{sovilj2016extreme} develop a novel methodology based on the Gaussian Mixture Model and Extreme Learning Machine to provide reliable estimates for the regression function. The Gaussian Mixture Model is used to model the data distribution, while the Extreme Learning Machine designs a multiple imputation strategy for the final estimate.

\citet{paiva2017stop} propose a methodology to impute continuous variables with missing data, where the missing data mechanism is non-ignorable for unit non-response (see \figurename~\ref{subfig:base_1}). Under a Bayesian approach, the procedure begins by fitting a mixture of multivariate normal distributions based on the observed data. Then, using samples from the posterior distribution, an analyst can obtain imputed data in a variety of scenarios. When additional information can be acquired for all individuals from alternative sources, one can opt to include these extra information in the database. This alters the pattern of missing data analyzed by \citet{paiva2017stop} in \figurename~\ref{subfig:base_1} to the pattern shown in \figurename~\ref{subfig:base_2}.

Therefore, a crucial question arises: how to include the information of the new variables in this model in such a way that it improves the imputation process under some established criteria?
\begin{figure}[htb]
\centering
\begin{subfigure}[]{0.40\textwidth}
\includegraphics[width=\textwidth]{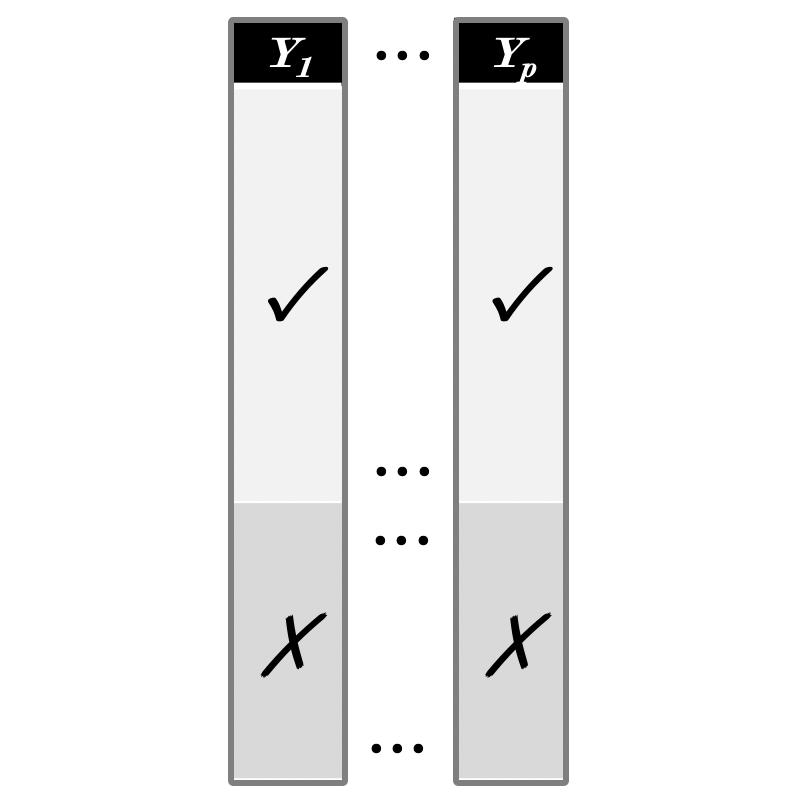}
\caption{No auxiliary information.}
\label{subfig:base_1}
\end{subfigure}
\hfill
\begin{subfigure}[]{0.40\textwidth}
\includegraphics[width=\textwidth]{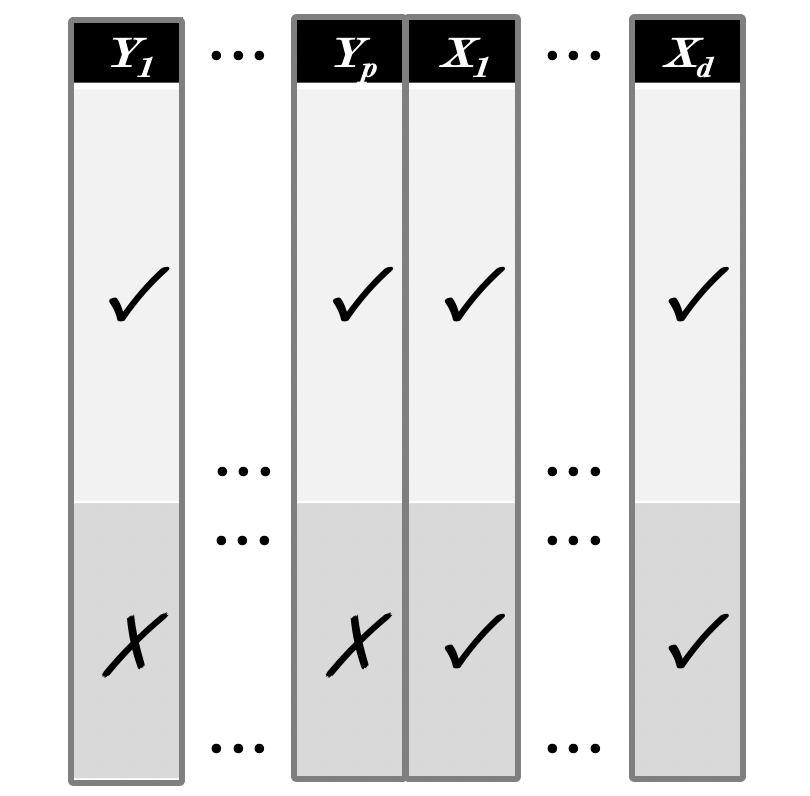}
\caption{With auxiliary information.}
\label{subfig:base_2}
\end{subfigure}
\caption[Missing data pattern from multivariate databases]{Missing data pattern from multivariate databases in the cases of not including and including auxiliary information. We illustrate the missing data patterns from the databases we want to deal with. In them, the check mark ({\footnotesize \Checkmark}) indicates that the data is observed, while the tag ({\footnotesize \XSolidBrush}) indicates that the data is missing.}
\label{fig:bases}
\end{figure}

The implementation of mixture regression models was initially studied in the context that fully observed variables assume the role of covariates. The covariates are treated as deterministic, so they do not carry information about which mixture component the subject is likely to belong to. Although this assumption may be reasonable in experiments where the explanatory variable is completely determined by the experimenter, \citet{hoshikawa2013mixture} states that with observational data the covariates may behave differently between groups. Therefore, the model must also account for the heterogeneity of the covariates, as it enables the estimation of the component to which the subject belongs based on this information. A model that has
these characteristics is the \textit{Cluster-Weighted Modeling} \citep[CWM,][]{gershenfeld1997nonlinear}. 
\citet{ingrassia2012local} proposed to use the CWM in a statistical environment and showed that it is a general and flexible family of mixture models.
In this work, we propose an imputation methodology for handling data with non-response units. Our method relies on additional information from auxiliary variables, which are completely observed and sourced externally. To this end, we employ the Multivariate Gaussian Mixture Model and adopt the Cluster-Weighted Modeling approach  considering that the auxiliary variables can be treated as observational and non-deterministic data, extending the work of \citet{masmela2021ufmg} that used the CWM to impute univariate responses. By assuming that these variables exhibit different behaviors between groups, the model is capable of selecting the appropriate component to carry out the imputation, which is our main interest. For this purpose, we evaluate the model performance using the Kullback-Liebler Divergence. In summary, this measure quantifies the information loss when approximating one distribution with another. Remarkably, we demonstrate that, although our model assume a Missing At Random (MAR) mechanism, under certain conditions, it yields promising imputation results even when data is Missing Not At Random (MNAR). This finding emphasizes the effectiveness and robustness of our approach.

The paper is structured as follows. In Section~\ref{Methodology}, the Gaussian Linear Cluster Weighted is presented in the framework of Gaussian mixture distributions. Some theoretical results are introduced through some theorems that relate various models. The model estimation process is presented from a Bayesian perspective as well as the implementation of the imputation mechanism. Next, simulation studies are presented in Section~\ref{Simulation}. They include the performance of the model when variables of various types are included, as well as comparison with other procedures in the statistical literature. Section~\ref{Ejemplo} is dedicated to presenting the findings obtained from applying the imputation model to actual data sourced from the Iris flower dataset or Fisher's Iris data set. 
Finally, concluding remarks about the method's benefits, limitations, and open problems are provided in Section~\ref{Conclusiones}.

\section{Methodology}
\label{Methodology}

\textit{Cluster-Weighted Modeling} (CWM) is a flexible mixed approach to model the joint probability of data from heterogeneous populations through a weighted sum of the products of marginal distributions and conditional distributions. The CWM was initially introduced by \citet{gershenfeld1997nonlinear} for modeling time series data related to musical instrument parameters. Under the Gaussian assumption, \citet{ingrassia2012local} showed that the Gaussian linear CWM includes \textit{Finite Mixture Models} (FMM) and \textit{Mixtures of Regressions Models} (MRM) as special cases in the case of a univariate response variable that takes values in $\mathbb{R}$. The Gaussian CWM is a mixture of regression models with random covariates that allows for flexible clustering of a random vector composed of response variables and covariates \citep{punzo2017robust}.

A multivariate extension of CWM that can explain the correlations between multivariate responses is presented by \citet{dang2017multivariate}. Suppose the pair $(\bm{X}^\top,\bm{Y}^\top)^\top$ constituted by the random vectors $\bm{X}\in \mathbb{R}^d$ and  $\bm{Y} \in \mathbb{R}^p$, thus the Multivariate Gaussian \textit{Linear Cluster-Weighted Modeling} (LCWM) decomposes the joint probability $p(\bm{x},\bm{y})$ as follows,
\begin{equation}
p(\bm{x},\bm{y})=\sum_{Z=1}^{G}\phi_p(\bm{y};\bm{B}^\top_Z\bm{x}+\bm{b}_{Z,0},\widetilde{\bm{\Sigma}}_Z) \, \phi_d(\bm{x};\bm{\mu}_Z,\bm{\Sigma}_Z) \, \alpha_Z,
\label{lcwmg}
\end{equation}
where $\alpha_Z=\Pr(Z)$ is the mixture weight of cluster $Z$, while $p(\bm{x}|Z)=\phi_d(\bm{x};\bm{\mu}_Z,\bm{\Sigma}_Z)$ and $p(\bm{y}|\bm{x},Z)=\phi_p(\bm{y};\bm{B}^\top_Z\bm{x}+\bm{b}_{Z,0},\widetilde{\bm{\Sigma}}_Z)$ correspond to the marginal and conditional distributions respectively. The notation $\phi_d$ (and $\phi_p$) represents the density of a $d$-variate (and $p$-variate) Gaussian random vector. Likewise, $\bm{\mu}_Z$ denotes the mean vector, $\bm{\Sigma}_Z$ and $\widetilde{\bm{\Sigma}}_Z$ are variance-covariance matrices, $\bm{B}_Z \in \mathbb{R}^{d \times p}$ is a matrix of coefficients, and $\bm{b}_{Z,0} \in \mathbb{R}^{p}$.

Two essential posterior distributions, utilized primarily for classification purposes, are defined through conditional probabilities, namely $p(Z|\bm{x},\bm{y})$ and $p(Z|\bm{x})$. The former represents the probability that the observation $(\bm{x},\bm{y})$ belongs to group $Z$ and is defined as,
\begin{equation}
p(Z|\bm{x},\bm{y})=\frac{\phi_p(\bm{y};B^\top_Z\bm{x}+\bm{b}_{Z,0},\widetilde{\bm{\Sigma}}_Z) \, \phi_d(\bm{x};\bm{\mu}_Z,\bm{\Sigma}_Z) \, \alpha_Z}{\sum_{Z=1}^{G}\phi_p(\bm{y};B^\top_Z\bm{x}+\bm{b}_{Z,0},\widetilde{\bm{\Sigma}}_Z) \, \phi_d(\bm{x};\bm{\mu}_Z,\bm{\Sigma}_Z) \, \alpha_Z}.
\label{pesos_LCWMG}
\end{equation}
The latter, used to determine the component responsible for imputing information related to the input vector $\bm{x}$, is defined as,
\begin{equation}
p(Z|\bm{x})=\frac{\phi_d(\bm{x};\bm{\mu}_Z,\bm{\Sigma}_Z)\alpha_Z}{\sum_{Z=1}^{G}\phi_d(\bm{x};\bm{\mu}_Z,\bm{\Sigma}_Z)\alpha_Z}.
\label{pesos_LCWMG_x}
\end{equation} 
Henceforth, we will refer to the distributions in (\ref{pesos_LCWMG}) and (\ref{pesos_LCWMG_x}) as  \textit{responsibilities}.

\subsection{Relationships between FMM, MRM, and LCWM under Gaussian assumptions}
\label{subsec:rel_models}

In this section, we introduce three crucial results regarding the probability functions and responsibilities that establish a connection between FMM, MRM, and LCWM in the multivariate Gaussian setting, and allow us to implement our imputation proposal. The first result is the most important to our method, while the two additional ones enrich our comprehension of the interrelationships among these models, for more details see \citet{masmela2021ufmg}.

Let $\bm{W}=(\bm{X}^\top,\bm{Y}^\top)^\top$ be a random vector that takes values in $\mathbb{R}^{d+p}$, where $\bm{X}\in \mathbb{R}^{d}$ and $\bm{Y}\in \mathbb{R}^{p}$. Assume that the density $p(\bm{w})$ corresponds to a Gaussian FMM, that is,
\begin{equation}
	p(\bm{w})=\sum_{Z=1}^{G} \phi_{d+p}\left(\bm{w};\bm{\mu}_Z^{(\bm{w})},\bm{\Sigma}_Z^{(\bm{w})}\right)\alpha_Z,
	\label{fmm_M}
	\end{equation}
where $\bm{\mu}^{(\bm{w})}_Z$ and $\bm{\Sigma}^{(\bm{w})}_Z$ are the mean vector and the variance-covariance matrix of $\bm{W}|Z$, respectively. In (\ref{fmm_M}), $\alpha_Z=\Pr(Z)$ is the mixture weight of the components marked with $Z \in \{1,\ldots, G\}$ and
\begin{equation*}
\bm{\mu}_Z^{(\bm{w})}=
\begin{pmatrix}
\bm{\mu}_Z^{(\bm{x})}\\
\bm{\mu}_Z^{(\bm{y})}
\end{pmatrix}
\quad \text{ and } \quad \bm{\Sigma}_Z^{(\bm{w})}=
\begin{pmatrix}
\bm{\Sigma}_Z^{(\bm{xx})} & \bm{\Sigma}_Z^{(\bm{xy})} \\
\bm{\Sigma}_Z^{(\bm{yx})} & \bm{\Sigma}_Z^{(\bm{yy})} 
\end{pmatrix}.
\end{equation*}

A first interesting result indicates that the CWM contains the FMM as a special case and, specifically, restricting the CWM to the the Gaussian context, FMM and LCWM are equivalent.

\begin{teo}
\label{PropM_1}
	Assume that $\bm{W}=(\bm{X}^\top,\bm{Y}^\top)^\top$ is a random vector that takes values in a subset of $\mathbb{R}^{d+p}$, and suppose that $\bm{W}|Z \sim \mathcal{N}_{d+p}\left(\bm{\mu}_Z^{(\bm{w})},\bm{\Sigma}_Z^{(\bm{w})}\right)$ with $Z \in \{1,\ldots , G\}$. In particular, the density $p(\bm{w})$ of $\bm{W}$ is the Gaussian FMM:
	\begin{equation*}
	p(\bm{w})=\sum_{Z=1}^{G} \phi_{d+p}\left(\bm{w};\bm{\mu}_Z^{(\bm{w})},\bm{\Sigma}_Z^{(\bm{w})}\right)\alpha_Z.
	\end{equation*}
	So, $p(\bm{w})$ can be written similarly to 
	\begin{equation*}
	p(\bm{x},\bm{y})=\sum_{Z=1}^{G}\phi_p(\bm{y};B^\top_Z\bm{x}+\bm{b}_{Z,0},\widetilde{\bm{\Sigma}}_Z)\phi_d(\bm{x};\bm{\mu}_Z,\bm{\Sigma}_Z)\alpha_Z,
	\end{equation*}		
	that is, a Gaussian LCWM.
\end{teo}
	
\begin{proof}[\textbf{\upshape Proof:}]
	Let us set $\bm{W}=(\bm{X}^\top,\bm{Y}^\top)^\top$, where $\bm{X}$ is a $d-$dimensional random vector and $\bm{Y}$ is a $p-$dimensional random vector. Using properties of the multivariate normal distribution, 
		\begin{equation*}
\begin{split}
p(\bm{w})=\sum_{Z=1}^{G}\phi_{d+p}\left((\bm{x},\bm{y});\bm{\mu}_Z^{(\bm{w})},\bm{\Sigma}_Z^{(\bm{w})}\right)\alpha_Z 
=\sum_{Z=1}^{G}\phi_p\left(\bm{y};\bm{\mu}_Z^{(\bm{y|x})},\bm{\Sigma}_Z^{(\bm{y|x})}\right)\phi_d\left(\bm{x};\bm{\mu}_Z^{(\bm{x})},\bm{\Sigma}_Z^{(\bm{x})}\right)\alpha_Z,
\end{split}
\end{equation*}
where,
\begin{equation}
	\begin{split}
	\bm{\mu}_Z^{(\bm{y|x})}&=\bm{\mu}_Z^{(\bm{y})}+\bm{\Sigma}_Z^{(\bm{yx})}{\bm{\Sigma}_Z^{(\bm{xx})}}^{-1}\left(\bm{x}-\bm{\mu}_Z^{(\bm{x})}\right)\\
	&=\left[ \bm{\Sigma}_Z^{(\bm{yx})}{\bm{\Sigma}_Z^{(\bm{xx})}}^{-1}\right] \bm{x}+\left[\bm{\mu}_Z^{(\bm{y})}-\bm{\Sigma}_Z^{(\bm{yx})}{\bm{\Sigma}_Z^{(\bm{xx})}}^{-1}\bm{\mu}_Z^{(\bm{x})} \right]\\
	&=B^\top_Z\bm{x}+\bm{b}_{Z,0},
	\end{split}
	\label{eq:mean_LCWM}
\end{equation}	
and,		
\begin{equation}
	\begin{split}
	\bm{\Sigma}_Z^{(\bm{y|x})}=\bm{\Sigma}_Z^{(\bm{yy})}-\bm{\Sigma}_Z^{(\bm{yx})}{\bm{\Sigma}_Z^{(\bm{xx})}}^{-1}\bm{\Sigma}_Z^{(\bm{xy})}
	=\widetilde{\bm{\Sigma}}_Z.
	\end{split}
	\label{eq:var_LCWM}
\end{equation}	
Then, Equation~\eqref{fmm_M} can be written as Equation~\eqref{lcwmg}.
\end{proof}

The proof of Theorem~\ref{PropM_1} enables us to conclude that once the parameters of the FMM are estimated, these estimates can be utilized to derive the parameters of the Gaussian LCWM. This obviates the need to establish a separate procedure for directly estimating the parameters of the Gaussian LCWM. Instead, we leverage the procedure established by \citet{kim2014multiple} and \citet{paiva2017stop} for estimating the parameters of the Gaussian FMM, along with the corresponding mappings to the expressions in (\ref{eq:mean_LCWM}) and (\ref{eq:var_LCWM}). The following result establishes the relationship between the LCWM and the MRM when the covariate $\bm{x}$ has the same distribution across all components.

\begin{teo}
\label{PropM_2}
	Consider the Gaussian LCWM given in (\ref{lcwmg}), with $\bm{X}|Z \sim \mathcal{N}_d(\bm{\mu}_Z,\bm{\Sigma}_Z)$ and $Z \in \{1,\ldots, G\}$. If the probability density of $\bm{X}|Z$ does not depend on the component, i.e., $\phi_d(\bm{x};\bm{\mu}_Z,\bm{\Sigma}_Z)=\phi_d(\bm{x};\bm{\mu},\Sigma)$ for all $Z \in \{1,\ldots, G\}$, then it follows that
	\begin{equation*}
	p(\bm{x},\bm{y})=\phi_d(\bm{x};\bm{\mu},\Sigma)p(\bm{y}|\bm{x}),
	\end{equation*}
	where $p(\bm{y}|\bm{x})$ is the Gaussian MRM given by the expression
	\begin{equation}
	p(\bm{y}|\bm{x})=\sum_{Z=1}^{G}\phi_p(\bm{y};\bm{B}^\top_Z\bm{x}+\bm{b}_{Z,0},\widetilde{\bm{\Sigma}}_Z)\alpha_Z.
	\label{GLMRM_M}
	\end{equation}
\end{teo}
\begin{proof}[\textbf{\upshape Proof:}]
	Assume that $\phi_d(\bm{x};\bm{\mu}_Z,\bm{\Sigma}_Z)=\phi_d(\bm{x};\bm{\mu},\Sigma)$ for all $Z \in \{1,. . . , G\}$, then from the expression in (\ref{lcwmg})
	\begin{equation*}
	\begin{split}
	p(\bm{x},\bm{y})&=\sum_{Z=1}^{G}\phi_p(\bm{y};\bm{B}^\top_Z\bm{x}+\bm{b}_{Z,0},\widetilde{\bm{\Sigma}}_Z)\phi_d(\bm{x};\bm{\mu}_Z,\bm{\Sigma}_Z)\alpha_Z\\
	&=\phi_d(\bm{x};\bm{\mu},\Sigma)\sum_{Z=1}^{G}\phi_p(\bm{y};\bm{B}^\top_Z\bm{x}+\bm{b}_{Z,0},\widetilde{\bm{\Sigma}}_Z)\alpha_Z\\
	&=\phi_d(\bm{x};\bm{\mu},\Sigma)p(\bm{y}|\bm{x})
	\end{split}
	\end{equation*}	 
	where $p(\bm{y}|\bm{x})$ is the Gaussian MRM given in Equation~\eqref{GLMRM_M}.
\end{proof}

A significant outcome derived from the same hypotheses assumed in Theorem \ref{PropM_2}, is presented in Corollary~\ref{CorM_1}. This result provides the expression for the responsibilities that are adopted by our imputation method.

\begin{cor}
\label{CorM_1}
If the probability density of $\bm{X}|Z$ does not depend on the component, that is, $\phi_d(\bm{x};\bm{\mu}_Z,\bm{\Sigma}_Z)=\phi_d(\bm{x};\bm{\mu},\Sigma)$ for all $Z \in \{1,\ldots, G\}$, then the responsibility given by 
	\begin{equation}
p(Z|\bm{x},\bm{y})=\frac{\phi_p(\bm{y};\bm{B}^\top_Z\bm{x}+\bm{b}_{Z,0},\widetilde{\bm{\Sigma}}_Z)\phi_d(\bm{x};\bm{\mu}_Z,\bm{\Sigma}_Z)\alpha_Z}{\sum_{Z=1}^{G}\phi_p(\bm{y};\bm{B}^\top_Z\bm{x}+\bm{b}_{Z,0},\widetilde{\bm{\Sigma}}_Z)\phi_d(\bm{x};\bm{\mu}_Z,\bm{\Sigma}_Z)\alpha_Z},
	\label{pesos_LCWMG_M}
	\end{equation}
coincides with
	\begin{equation}
	p(Z|\bm{x},\bm{y})=\frac{\phi_p(\bm{y};\bm{B}^\top_Z\bm{x}+\bm{b}_{Z,0},\widetilde{\bm{\Sigma}}_Z)\alpha_Z}{\sum_{Z=1}^{G}\phi_p(\bm{y};\bm{B}^\top_Z\bm{x}+\bm{b}_{Z,0},\widetilde{\bm{\Sigma}}_Z)\alpha_Z}.
	\label{pesos_LMRM_M}
	\end{equation} 
\end{cor}

\begin{proof}[\textbf{\upshape Proof:}]
Assume that $\phi_d(\bm{x};\bm{\mu}_Z,\bm{\Sigma}_Z)=\phi_d(\bm{x};\bm{\mu},\Sigma)$ for all $Z \in \{1,\ldots,G\}$, from the expression in (\ref{pesos_LCWMG_M}) we get
\begin{equation*}
\begin{split}
p(Z|\bm{x},\bm{y})&=\frac{\phi_p(\bm{y};\bm{B}^\top_Z\bm{x}+\bm{b}_{Z,0},\widetilde{\bm{\Sigma}}_Z)\phi_d(\bm{x};\bm{\mu}_Z,\bm{\Sigma}_Z)\alpha_Z}{\sum_{Z=1}^{G}\phi_p(\bm{y};\bm{B}^\top_Z\bm{x}+\bm{b}_{Z,0},\widetilde{\bm{\Sigma}}_Z)\phi_d(\bm{x};\bm{\mu}_Z,\bm{\Sigma}_Z)\alpha_Z}\\
&=\frac{\phi_d(\bm{x};\bm{\mu},\Sigma)\phi_p(\bm{y};\bm{B}^\top_Z\bm{x}+\bm{b}_{Z,0},\widetilde{\bm{\Sigma}}_Z)\alpha_Z}{\phi_d(\bm{x};\bm{\mu},\Sigma)\sum_{Z=1}^{G}\phi_p(\bm{y};\bm{B}^\top_Z\bm{x}+\bm{b}_{Z,0},\widetilde{\bm{\Sigma}}_Z)\alpha_Z}\\
&=\frac{\phi_p(\bm{y};\bm{B}^\top_Z\bm{x}+\bm{b}_{Z,0},\widetilde{\bm{\Sigma}}_Z)\alpha_Z}{\sum_{Z=1}^{G}\phi_p(\bm{y};\bm{B}^\top_Z\bm{x}+\bm{b}_{Z,0},\widetilde{\bm{\Sigma}}_Z)\alpha_Z},
\end{split}
\end{equation*}
for $Z \in \{1,\ldots,G\}$.
\end{proof}

Equation~\eqref{pesos_LMRM_M} allows us to infer that, under the assumption that the distribution of $\bm{X}|Z$ does not depend on cluster, the classification of each observation $(\bm{x},\bm{y})$ is determined by the conditional distribution and the mixing weight. These responsibilities coincide with those defined for the MRM \citep{desarbo1988maximum,mclachlan2004finite,fruhwirth2006finite}. Additionally, under the same conditions established by the Corollary \ref{CorM_1}, we have that the responsibility in Equation~\eqref{pesos_LCWMG_x} simplifies to $p(Z|\bm{x})=\alpha_Z$.

Based on the presented results, we can conclude that have two scenarios. The first scenario, when the distribution of the input vector $\bm{X}|Z$ varies between components, it allows us to utilize the responsibilities to determine the appropriate component for imputation. This indicates that incorporating the vector of variables $\bm{X}$ should improve the imputation process. Conversely, in the second scenario, if the distribution of the vector $\bm{X}|Z$ does not depend on the component, it suggests that the information provided by variable $\bm{X}$ is uninformative for the imputation process. In such cases, the LCWM or MRM would provide similar outcomes.

\subsection{Bayesian estimation of the model} 
\label{subsec:Bayesian_est}

A mapping between the parameter vectors of the FMM and the LCWM was presented in  Theorem~\ref{PropM_1}. Therefore, within the presented methodology we are initially interested in estimating the parameters of the Gaussian FMM given in expression \eqref{fmm_M}. For this purpose, we implement a Bayesian approach to estimate the multivariate Gaussian mixture model using a stick-breaking representation of a truncated Dirichlet process as the prior distribution of the mixture weights \citep{ferguson1973bayesian, sethuraman1994constructive}.

The full model is constructed as follow. Suppose that each individual in the data set belongs to one of $G$ mixture components, that is, $Z_i \in \{1,\ldots,G \}$ so that $\bm{Z}^\top=(Z_1,\ldots,Z_n)^\top$. The mixture weights are given by $\bm{\alpha}^\top=(\alpha_1,\ldots,\alpha_G)^\top$ with $\alpha_g = \Pr(Z_i=g)$ where $i \in \{1.\ldots,n \}$ and $g \in \{1,\ldots,G\}$. If $\bm \mu= (\bm \mu_1,\ldots,\bm \mu_G)$ and $\bm\Sigma=(\bm\Sigma_1,\ldots,\bm\Sigma_G)$ then,
\begin{equation*}
\begin{split}
\textbf{\textit{W}}_i|Z_i,\bm\mu,\bm\Sigma &\sim \mathcal{N}_{d+p}(\bm \mu_{Z_i},\Sigma_{Z_i}),\\
Z_i|\bm \alpha &\sim \text{Multinomial}(\bm \alpha). 
\end{split}
\end{equation*}

In other mixed model applications for multiple imputation of missing data  \citep{si2013nonparametric,manrique2017bayesian}, it has been shown that a stick-breaking representation of a truncated Dirichlet process was successful and allows for a fast convergence of the MCMC methods \citep{kim2014multiple}. Therefore, using this representation, we have that
\begin{equation*}
\begin{split}
\alpha_g &= \nu_g \prod_{k<g} (1-\nu_k) \text{ for } g \in \{1,\ldots,G\},\\
\nu_g &\sim \text{Beta}(1,\eta) \text{ for } g \in \{1,\ldots,G-1\}; \quad \nu_G=1,\\
\eta &\sim \text{Gamma}(a_{\eta},b_{\eta}).
\end{split}
\end{equation*}
Next, for $g \in \{1,\ldots,G\}$, we specify the prior distributions for the parameters $(\bm\mu,\bm\Sigma)$ 
\begin{equation*}
\begin{split}
\bm{\mu}_g|\Sigma_g &\sim \mathcal{N}_{d+p}(\bm \mu_0,h^{-1}\Sigma_g),\\
\Sigma_g &\sim \text{Inverse Wishart}(f,\bm{\Delta}),
\end{split}
\end{equation*}
where $f$ is the prior degrees of freedom and $\bm{\Delta}=\text{diag}(\delta_1,\ldots,\delta_p)$ with $\delta_j \sim \text{Gamma}(a_{\delta},b_{\delta})$ for $j \in \{1,\ldots,p\}$. 

In the next section, we present the full conditionals and the Gibbs sampler algorithm to estimate the posterior distribution and to implement the proposed imputation method. It is important to notice that, for the number of components $G$, \citet{kim2015simultaneous} recommends fixing a large value initially. At each iteration of the Gibbs sampler, the number of nonempty components is counted. If this count reaches the value assigned to $G$, it is prudent to increase $G$ and readjust the model with more components. When the count of nonempty components is less than $G$, then the selected value of $G$ is reasonable.

\subsection{Imputation procedure}
\label{subsec: imp_proced}

Our goal is to establish an imputation procedure that takes advantage of extra information available for all individuals (see \figurename~\ref{subfig:base_2}). If these new variables demonstrate a strong correlation with the variables containing missing data, this supplementary information can significantly influence the adequacy of the distribution mixture fit, thus affecting the imputation procedure. Specifically, the new information enhances the ability to determine the suitable component for imputation. By incorporating this information, we can make the assumption that the missing data mechanism is MAR \citep[see, e.g.,][]{little2019statistical, zhang2003multiple}.
Assuming that the input data are observational, we jointly use input and output variables to fit a Gaussian FMM (see Section~\ref{subsec:Bayesian_est}). Based on the obtained estimates and utilizing the expressions in \eqref{eq:mean_LCWM} and \eqref{eq:var_LCWM}, we derive parameter estimates for the conditional model in the Gaussian LCWM. These estimates play a crucial role in establishing the imputation steps. The values of the input variables for individuals with missing information will be used in combination with the responsibilities in an adaptive manner to determine the appropriate component for imputation.

To implement the Gibbs sampler algorithm, it is necessary to calculate the conditional posterior distributions for each parameter in our model. This involves considering the functional form of the priors distributions, the information provided by the observed data, and the current values of the model parameters. The updates are performed based on the initial values assigned to parameters, hyperparameters, and missing data, as follows. For $i \in \{1,\ldots,n \}$, the full conditional for $\bm{Z}_i$ is:
\begin{equation}
\bm{Z}_i|\bm{\alpha},\bm{\mu},\bm{\Sigma},\bm{W}_i \sim \text{Multinomial}(\alpha^*_{i,1},\ldots,\alpha^*_{i ,G})
\end{equation}
where,
\begin{equation*}
\alpha^*_{i,g}=\frac{ \phi_{d+p}(\bm{w}_i|\bm{\mu}_g,\Sigma_g) \alpha_g}{\sum_{k=1}^ {G} \phi_{d+p}(\bm{w}_i|\bm{\mu}_k,\Sigma_k) \alpha_k},
\end{equation*}
for $g \in \{1,\ldots,G\}$.
To estimate the mixing probabilities $\bm{\alpha}$, a construction based on a stick-breaking process is employed. The parameter $\eta$ is updated from its full conditional distribution as follows:
\begin{equation*}
\eta|\bm{\alpha} \sim \text{Gamma}\left(a_{\eta}+G-1,b_{\eta} - \ln \alpha_G\right).
\end{equation*}
Specifically, $\nu_g$ for $g \in \{1,\ldots,G-1\}$ is updated using the following conditional distribution:
\begin{equation*}
\nu_g|\bm{z},\eta \sim \text{Beta}\left(1+n_g(\bm{z}),\eta + \sum_{j>g}n_j(\bm{z}) \right),
\end{equation*}
additionally, it is set that $\nu_G=1$. By applying a transformation to the values of $\bm{\nu}$, we can obtain the corresponding values of $\bm{\alpha}$ (see Section~\ref{subsec:Bayesian_est}).

The full conditional for $\delta_j|\bm{\Sigma}$ can be expressed as follows:
\begin{equation*}
	\delta_j|\bm{\Sigma} \sim \text{Gamma}(a^j_g,b^j_g),
\end{equation*}
where,
\begin{equation*}
	a^j_g=a_{\delta}+\frac{G(p+1)}{2} \quad\text{and}\quad
	b^j_g=b_{\delta}+\frac{1}{2}\sum_{g=1}^{G}(\bm{\Sigma}_g)^{-1}_{j,j} \; \mbox{for } j \in \{1,\ldots,p\}.
\end{equation*}
For the variance-covariance matrix, we can express the conditional posterior distribution as follows:
\begin{equation*}
 \bm{\Sigma}_g|\bm{W},\bm{z} \sim \text{Inverse Wishart}(\bm{\Delta}_g,f_g),   
\end{equation*}
such that,
\begin{equation*}
 f_g=f+n_g \quad\text{and}\quad
	\bm{\Delta}_g=\bm{\Delta}+\bm{S}_g+\frac{(\bar{\bm{w}}_g-\bm{\mu}_0)(\bar{\bm{w}}_g-\bm{\mu}_0)^\top}{1/h+1/n_g},   
\end{equation*}
where, $\bar{\bm{w}}_g =\sum_{\{i:Z_i=g\}}\bm{w}_i/n_g$ and $\bm{S}_g=\sum_{\{i:Z_i=g\}}(\bm{w}_i-\bar{\bm{w}}_g)(\bm{w}_i-\bar{\bm{w}}_g)^\top$. The conditional posterior distribution for the mean vector, given the variance-covariance matrix, can be expressed as follows:
\begin{equation*}
 \bm{\mu}_g|\bm{\Sigma}_g,\bm{W},\bm{Z} \sim \mathcal{N}_{d+p}(\bar{\bm{\mu}}_g,\bar{\bm{\Sigma}}_g),   
\end{equation*}
where,
\begin{equation*}
 \bar{\bm{\mu}}_g=\frac{ h\bm{\mu}_0 + n_g \bar{\bm{w}}_g}{h+n_g} \quad\text{and}\quad \bar{\bm{\Sigma}}_g=\frac{1}{h+n_g}\bm{\Sigma}_g.   
\end{equation*}
For further details on the computation of marginal distributions, please refer to  \href{JMVA_Supplementar.pdf}{Supplementary Material, Section} \ref{A-Bayesian_Inference}.

Once the parameters of the Gaussian FMM have been estimated, Theorem~\ref{PropM_1} allows us to obtain the estimated values for the Gaussian LCWM. Specifically, we can derive estimates for the parameters of the marginal and conditional distributions involved in the expression \eqref{lcwmg}, as well as for the responsibilities. The first set of responsibilities in \eqref{pesos_LCWMG} allows us to generate posterior samples of $\bm{Z}$ for all observations, this process is the basis for estimating the parameters $\bm{\alpha}$, $\bm{\mu}$ and $\bm{\Sigma}$. Meanwhile, the responsibilities specified in \eqref{pesos_LCWMG_x} are used to obtain posterior samples of $\bm{Z}$ for individuals with missing data, specifically used in the imputation step included in the algorithm.

Algorithm~\ref{algo:cwm} summarizes the Gibbs sampler algorithm for the Gaussian FMM 
used to implement the imputation procedure. The output variable can be partitioned as $\bm{y}=(\bm{y}_\text{obs},\bm{y}_\text{mis})$, where $\bm{y}_\text{obs}$ represents the observed portion and $\bm{y}_\text{mis}$ represents the missing portion. The values of the input variables corresponding to the previously defined partition will be denoted as $\bm{X}_{\text{obs}}$ and $\bm{X}_{\text{mis}}$, respectively. For classification purposes, the notation $\bm{z}$ is employed to classify the observations in the process of estimating the model parameters, while $\bm{z}_\text{mis}$ is used to classify the observations with missing information during the imputation process. The algorithm incorporates two imputation steps to update $\bm{z}_\text{mis}$ and $\bm{y}_\text{mis}$.

\begin{algorithm}[htp]
\DontPrintSemicolon 
\KwIn{$\bm{y}_{\text{obs}}$, $\bm{X}_{\text{obs}}$, $\bm{X}_{\text{mis}}$}
\KwOut{${\bm{y}}_{\text{mis}}$, $\hat{\bm{\alpha}}$, $\hat{\bm{\mu}}$, $\hat{\bm{\Sigma}}$}
\SetAlgoLined
 initialization: $\bm{y_\text{mis}}^{(0)}$, $\bm{\alpha}^{(0)}$, $\bm{\mu}^{(0)}$, $\bm{\Sigma}^{(0)}$\;
 \For{$j\in \{1,\ldots,J\}$}{
  generate $\bm{z}^{(j)}$ from 
  $p(\bm{z}|\bm{y_\text{obs}},\bm{X}_{\text{obs}},\bm{X}_{\text{mis}},\bm{y_\text{mis}}^{(j-1)},\bm{\alpha}^{(j-1)},\bm{\mu}^{(j-1)},\bm{\Sigma}^{(j-1)})$\;
  generate $\bm{\nu}^{(j)}$ from 
  $p(\bm{\nu}|\bm{z}^{(j)})$\;
  calculate $\bm{\alpha}^{(j)} = f(\bm{\nu}^{(j)})$\;
  \For{$g \in \{1,\ldots,G\}$}{
  generate $\Sigma_g^{(j)}$ from 
  $p(\Sigma_g|\bm{y_\text{obs}},\bm{X}_{\text{obs}},\bm{X}_{\text{mis}},\bm{y_\text{mis}}^{(j-1)},\bm{z}^{(j)})$\;
  generate $\bm{\mu}_g^{(j)}$ from 
  $p(\bm{\mu}_g|\Sigma_g^{(j)},\bm{y_\text{obs}},\bm{X}_{\text{obs}},\bm{X}_{\text{mis}},\bm{y_\text{mis}}^{(j-1)},\bm{z}^{(j)})$\;
  }
  generate $\bm{z_\text{mis}}^{(j)}$ from 
  $p(\bm{z_\text{mis}}|\bm{X}_{\text{mis}},\bm{\alpha}^{(j)},\bm{\mu}^{(j)},\bm{\Sigma}^{(j)})$\;
  generate $\bm{y_\text{mis}}^{(j)}$ from 
  $p(\bm{y_\text{mis}}|\bm{X}_{\text{mis}},\bm{z_\text{mis}}^{(j)},\bm{\mu}^{(j)},\bm{\Sigma}^{(j)})$\;
  sort $\bm{\alpha}^{(j)}$ in decreasing order\; 
  reorder $\bm{z}^{(j)}$, $\bm{z_\text{mis}}^{(j)}$, $\bm{\mu}^{(j)}$, $\bm{\Sigma}^{(j)}$ based on 
  the order of $\bm{\alpha}^{(j)}$\;
 }
 \KwResult{Complete imputed database and parameters posterior estimation}
 \caption{\textbf{Gaussian Linear CWM}}
 \label{algo:cwm}
\end{algorithm}

\section{Simulation studies}
\label{Simulation}

A dataset was simulated from a mixture of Gaussian distributions in four dimensions with $C=2$ components. Two of the variables were considered output variables ($p=2$), while the other two were considered input variables ($d=2$). The database contains $n = 1000$ observations of the form $(x_1,x_2,y_1,y_2)$. The mixing probabilities are $\alpha_1=0.6$ and $\alpha_2=0.4$, the mean vectors $ \bm{\mu}_1=(1{.}0,3{.}0,4{.}0,2{.}0)$ and $\bm{\mu}_2=(1{.}0,9{.}0,7{.}0,6{.}0)$, and the covariance matrices are 
\begin{equation*}
\Sigma_1=\left(\begin{matrix} 
\phantom{-}1{.}00 & \phantom{-}0{.}50 & \phantom{-}0{.}50 & \phantom{-}0{.}50\\ 
\phantom{-}0{.}50 & \phantom{-}1{.}00 & \phantom{-}0{.}50 & \phantom{-}0{.}50 \\ 
\phantom{-}0{.}50 & \phantom{-}0{.}50 & \phantom{-}1{.}00 & \phantom{-}0{.}50 \\
\phantom{-}0{.}50 & \phantom{-}0{.}50 & \phantom{-}0{.}50 & \phantom{-}1{.}00
\end{matrix}\right) \text{,  } \quad \Sigma_2=\left(\begin{matrix} 
\phantom{-}1{.}00 & -0{.}50 & -0{.}50 & \phantom{-}0{.}50\\ 
-0{.}50 & \phantom{-}1{.}00 & \phantom{-}0{.}50 & -0{.}50 \\ 
-0{.}50 & \phantom{-}0{.}50 & \phantom{-}1{.}00 & -0{.}50 \\
\phantom{-}0{.}50 & -0{.}50 & -0{.}50 & \phantom{-}1{.}00
\end{matrix}\right) .
\end{equation*}
Missing data, generated under the mechanism of MAR, are incorporated into the dataset. In this context, the variables $Y_1$ and $Y_2$ contain missing information, while the variables $X_1$ and $X_2$ are fully observed. Approximately 10\% of the data in cluster 1 was randomly selected and labeled as missing, whereas around 50\% of the data in cluster 2 was selected for the same purpose. Table \ref{cuadro:datos_simulados_mult} displays the distribution of the simulated data. In each cell of the table, the corresponding absolute frequency is presented. To the right of the absolute frequency, in parentheses, the relative frequencies of the rows are indicated, which correspond to the proportion of observed and missing data for each cluster. At the bottom of each cell, in parentheses, the relative frequencies by columns are shown, which indicate the distribution of observed data by cluster for the first column, while for the second column it is shown for missing data. 
An imputation procedure that only considers information from observed data will impute approximately 70\% for the first cluster and 30\% for the second. However, upon entering new information, we want the model to impute missing data according to their distribution by cluster, which would be approximately 22\% for the first cluster and 78\% for the second.
\begin{table}[htb]
\centering

\begin{tabular}{c|cl|cl|cl}
\multicolumn{1}{l|}{}  & \multicolumn{2}{c|}{\textbf{observed}}& \multicolumn{2}{c|}{\textbf{missing}}& \multicolumn{2}{c}{\textbf{complete}} \\ \hline
\multirow{1}{*}{\textbf{cluster 1}} & 516& \textit{(89.7\%)} & \phantom{-}59     & \textit{(10.3\%)} & 575 & \textit{(100\%)} \\
& \multicolumn{1}{l}{\textit{(69.9\%)}} &  & \multicolumn{1}{l}{\textit{(22.5\%)}} &  & \multicolumn{1}{l}{\textit{(57.5\%)}} & \\ 
\hline
\multirow{1}{*}{\textbf{cluster 2}}  & 222  & \textit{(52.2\%)} & 203 & \textit{(47.8\%)} & 425  & \textit{(100\%)} \\
& \multicolumn{1}{l}{\textit{(30.1\%)}} &  & \multicolumn{1}{l}{\textit{(77.5\%)}} &  & \multicolumn{1}{l}{\textit{(42.5\%)}} &  \\ 
\hline
\multirow{1}{*}{\textbf{total}}  & 738 & \textit{(68.3\%)} & 262 & \textit{(31.7\%)} & 1000  & \textit{(100\%)} \\
& \multicolumn{1}{l}{\textit{(100\%)}} &  & \multicolumn{1}{l}{\textit{(100\%)}}  & & \multicolumn{1}{l}{\textit{(100\%)}}  &                 
\end{tabular}

\caption{Distribution of simulated data and pattern of missing data under a MAR mechanism for the multivariate case. }
\label{cuadro:datos_simulados_mult}
\end{table}

\figurename~\ref{fig:disp_simulados_multi} presents a pairwise plot of the data where the distribution of observed and missing data is illustrated. Two types of input variables with fully observed information are considered, the input variable $X_1$ that does not give information on which component to impute from, and the input variable $X_2$ distributed \textit{separately among the components}\footnote{We refer to a variable distributed \textit{separately among the components} as a variable whose distribution is multimodal (in the case of \figurename~\ref{fig:disp_simulados_multi} bimodal) and by knowing its value, it is possible to determine which of the components it belongs to.} and whose information allows us to conclude the component with which to impute (see the projections $X_1 \times Y_1$, $X_1 \times Y_2$, $X_2 \times Y_1$, and $X_2 \times Y_2$ in the planes in the lower left corner of \figurename~\ref{fig:disp_simulados_multi}).

\begin{figure}[htb]
    \centering
	\includegraphics[width=1.0\textwidth]{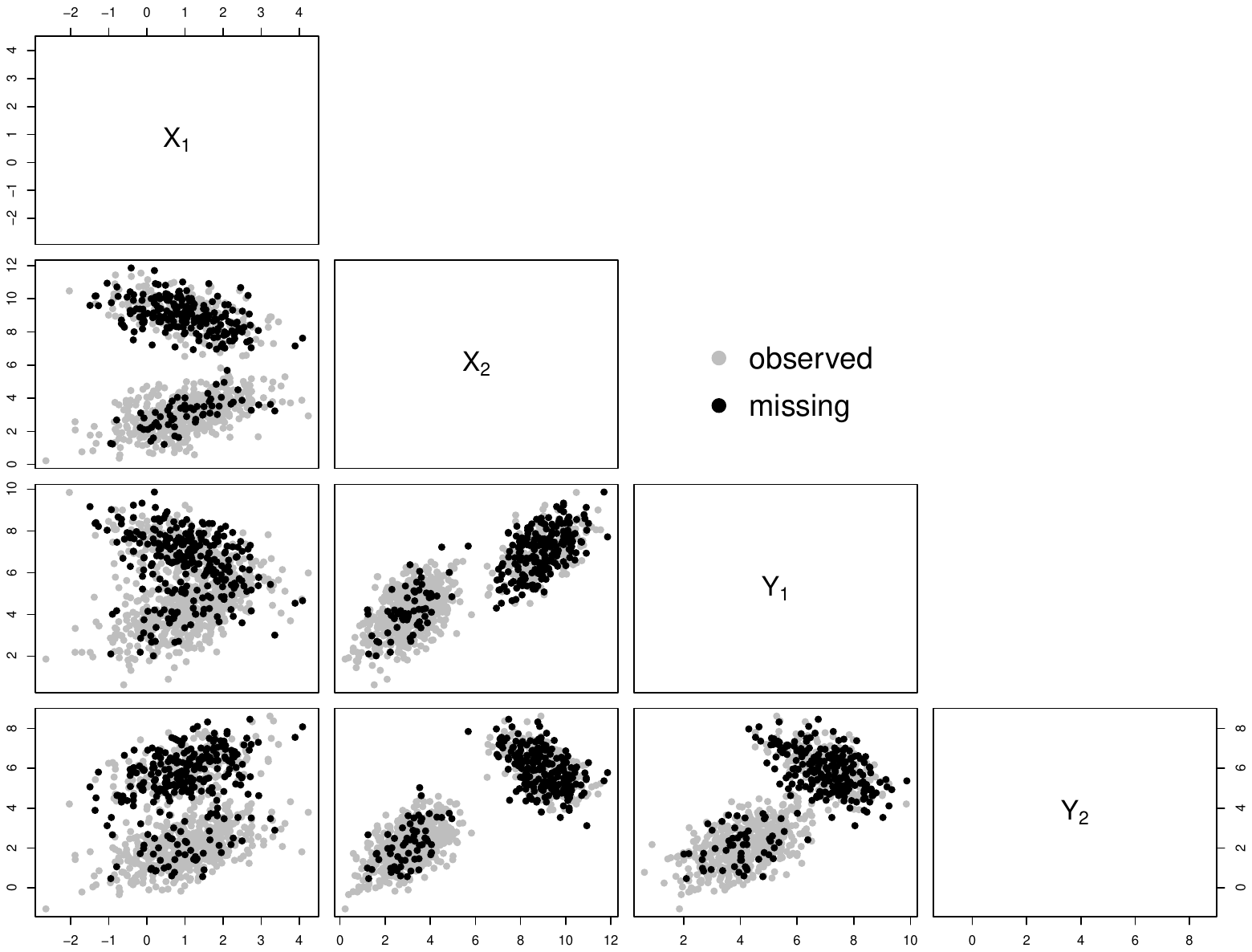}
	\caption[Observed and missing data generated under an MAR mechanism for the multivariate case.]{Pairwise plot of the variables in the simulated database. Observed and missing data generated under a MAR mechanism for the multivariate case.}
	\label{fig:disp_simulados_multi}
\end{figure}

\subsection{Model performance when new information is included}
\label{subsec:LCWM_X1_X2_X1X2_mult}

In this section, our objective is to analyze the type of information that can be incorporated into the model through the input variables. Specifically, we will focus on the variables $X_1$, $X_2$, as well as a vector consisting of both variables, in the case of our simulated data. In the first scenario, the input variable does not provide specific information to determine which of the two components to use for imputation. In the second scenario, the input variable is distributed separately among components, enabling the model to accurately determine the appropriate component for imputation. Finally, in the third scenario, we construct an input vector by combining the two variables mentioned in the previous scenarios. The inclusion of a variable that is distributed separately among components allows the vector to inherit this desirable characteristic.

\figurename~\ref{fig:imp_mult} illustrates the construction of the imputation models using information obtained from the three described scenarios. Specifically, \figurename~\ref{fig:imp_mult_X1} showcases the imputation model construction when auxiliary information is solely provided through the variable $X_1$. However, this variable does not provide information about which component you should use for the imputation process. Consequently, the decision is based on the estimates of the mixture probabilities, which strongly depend on the number of observations in each component. The left-hand panels display the $X_1 \times Y_1$ and $X_1 \times Y_2$ projection planes. Within these panels, 95\% quantile ellipses are shown for each component alongside their respective regression lines. The imputations generated by the model cluster closely around the regression lines, as depicted in the graphs. The right-hand side panels present projections involving the variable $X_2$, revealing that several imputed observations are located far from the regions of observed data.

In the scenario where we incorporate information from the variable $X_2$, \figurename~\ref{fig:imp_mult_X2} illustrates the construction of the model. Similar to the previous case, the figure presents the 95\% quantile ellipses and regression lines. By examining all four panels collectively, we observe that even when imputing solely based on information from variable $X_2$, the imputed data consistently originates from regions with observed information across all projection planes. This observation indicates the high quality of information provided to the model by the input variable $X_2$. As in the previous case, the imputations are generated around the regression lines.
\figurename~\ref{fig:imp_mult_X1X2} presents the implementation of the Gaussian LCWM procedure, which utilizes the input variables $X_1$ and $X_2$ jointly. The observed behavior is similar to the case where auxiliary information is derived solely from the input variable $X_2$. Imputed values are generated from regions with observed information and around the regression lines. It should be noted that, in this case, there is a regression plane and a quantile ellipsoid for the imputation process. The figures illustrate their projections on the relevant planes of interest.

\begin{figure}[htb]
\centering
\begin{subfigure}[]{.48\textwidth}
\includegraphics[width=1.0\textwidth]{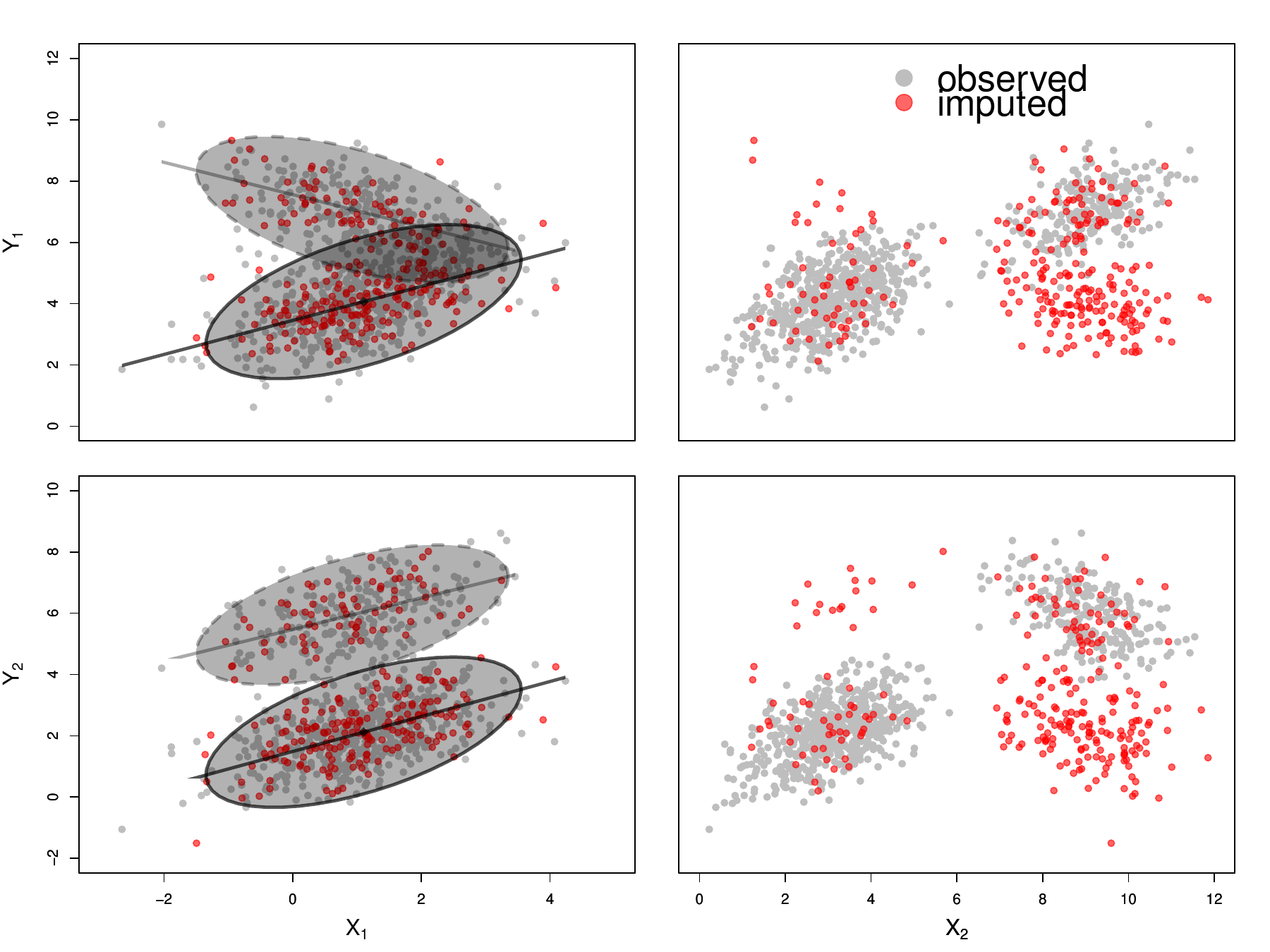}
\caption{With auxiliary information from variable $X_1$.}
\label{fig:imp_mult_X1}
\end{subfigure}
\hfill
\begin{subfigure}[]{.48\textwidth}
\includegraphics[width=1.0\textwidth]{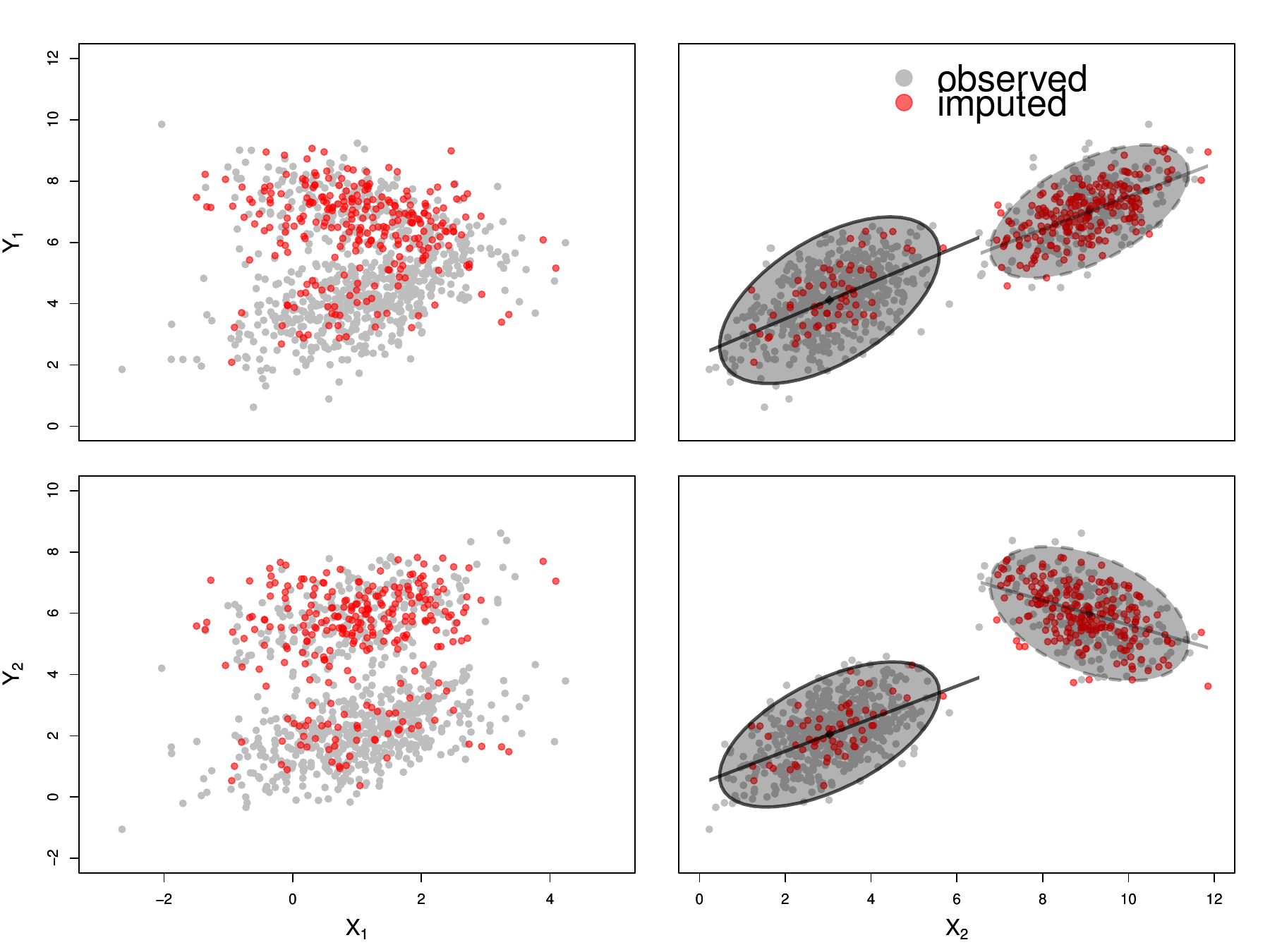}
\caption{With auxiliary information from variable $X_2$.}
\label{fig:imp_mult_X2}
\end{subfigure}
\hfill
\begin{subfigure}[]{.48\textwidth}
\includegraphics[width=1.0\textwidth]{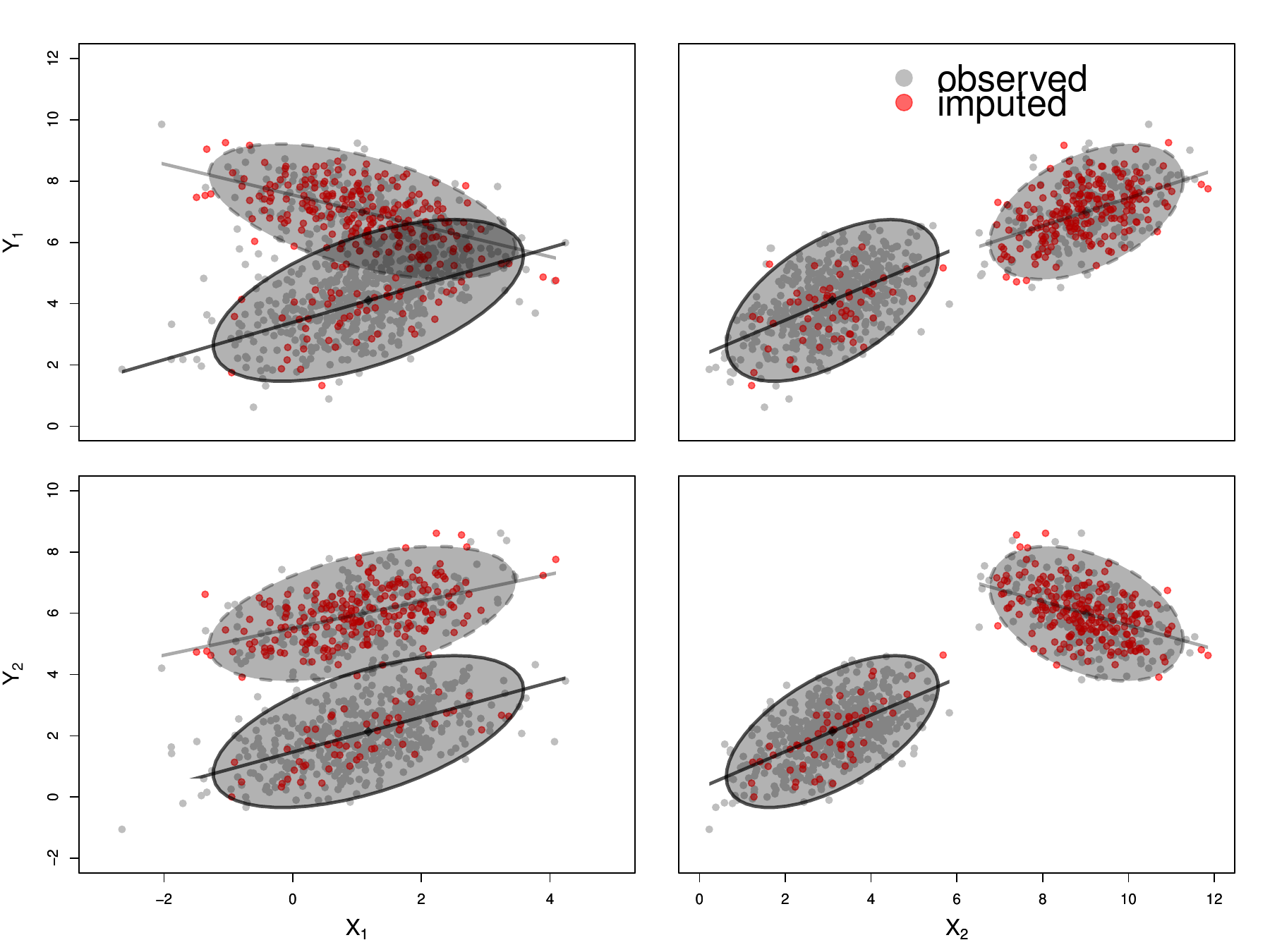}
\caption{With auxiliary information from vector $(X_1,X_2)$.}
\label{fig:imp_mult_X1X2}
\end{subfigure}
\caption[Construction of the imputation process through the Gaussian LCWM.]{Construction of the imputation process through the Gaussian LCWM using information from different types of variables. Observed and imputed data in the different projection planes. 95\% quantile ellipses and regression lines.}
\label{fig:imp_mult}
\end{figure}

In \figurename~\ref{fig:pairsplot_mult}, the pair plots illustrate the three imputation processes being discussed. Each scatter plot in the matrix displays the observed, missing, and imputed data. We first focus on the imputation process utilizing the input variable $X_1$. \figurename~\ref{fig:pairsplot_cwmx1_multi} provides insights into the behavior of the imputations (depicted by red dots) in relation to the missing data pattern (represented by black dots). Notably, the procedure incorrectly imputes data based on the proportion within each cluster. Cluster 1, despite having the smallest proportion of missing data, is imputed with the largest proportion of data. As mentioned earlier, this variable does not offer any information regarding which component to use for imputation. The imputation procedure relies on estimates of mixture probabilities, which are strongly influenced by the proportion of observed data in each component. By examining the projection planes $X_2 \times Y_1$ and $X_2 \times Y_2$, we can identify regions with inaccurately imputed observations that do not align with either observed or missing data. These projections effectively highlight the errors made during the imputation process, particularly noticeable in the $Y_1 \times Y_2$ plane due to the discrepancy in proportions between imputed and missing data within each component.

In the subsequent two scenarios, we witness a significant improvement in our model's performance. By incorporating an auxiliary variable or vector that is distributed separately among components, the Gaussian LCWM method enables precise determination of the component for imputation. \figurename~\ref{fig:pairsplot_cwmx2_multi} illustrates the imputation process using information from the variable $X_2$. As this variable is distributed separately among components, we can observe that the imputations accurately cover the regions corresponding to the missing data in proportion to their occurrence.

Moving on to \figurename~\ref{fig:pairsplot_cwmx1x2_multi}, which represents our model utilizing the input vector $(X_1, X_2)$, we notice a similar pattern to that of \figurename~\ref{fig:pairsplot_cwmx2_multi} in terms of the imputation results. The scatter plot in the $X_1 \times X_2$ projection panel reveals two distinct groups. This indicates that when incorporating information from both variables simultaneously, the distribution of the input vector $(X_1, X_2)$ becomes separated among the components.

\begin{figure}[htb]
\centering
\begin{subfigure}[]{0.48\textwidth}
\includegraphics[width=1.0\textwidth]{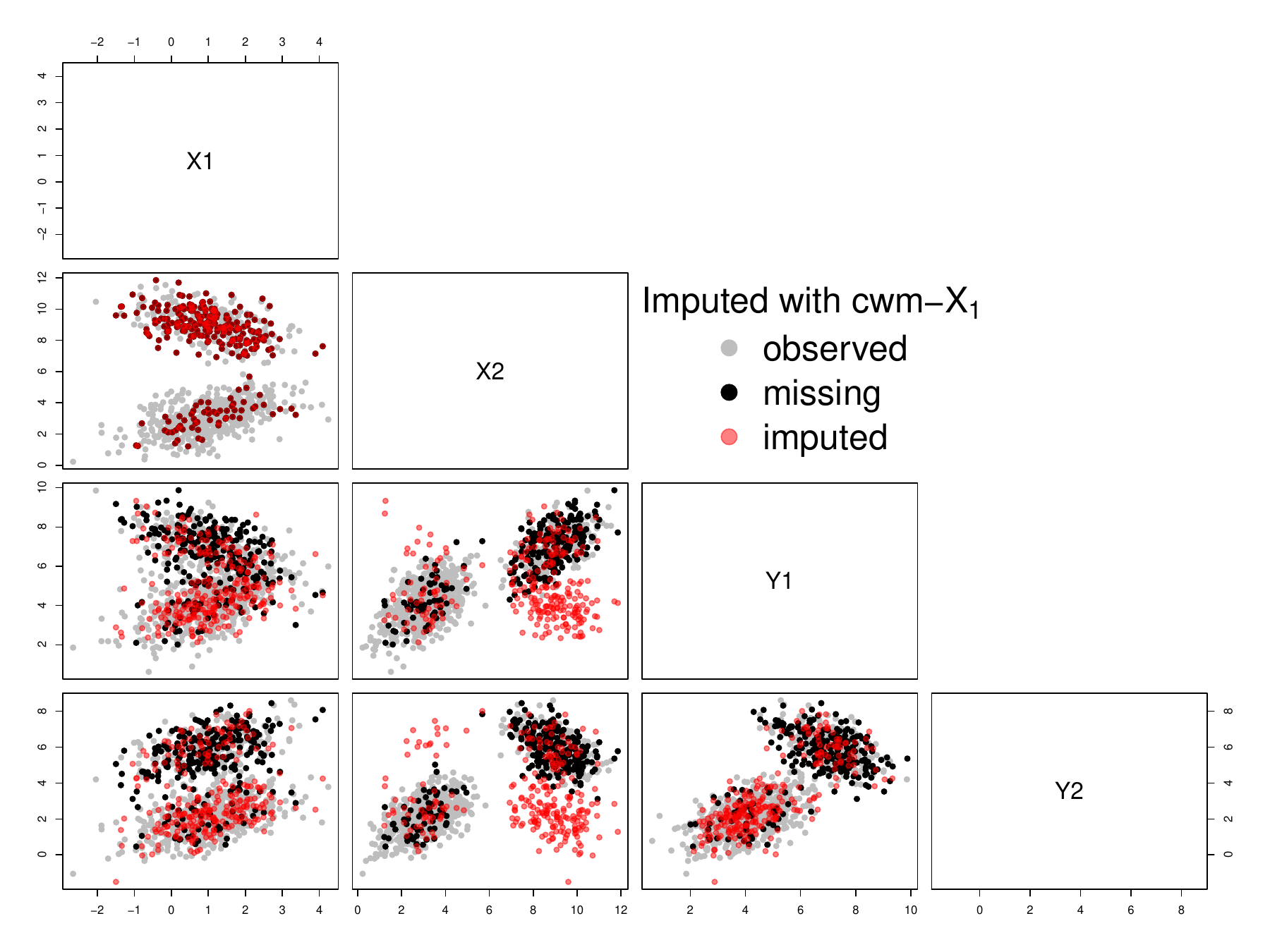}
\caption{Imputed with information from variable $X_1$}
\label{fig:pairsplot_cwmx1_multi}
\end{subfigure}
\hfill
\begin{subfigure}[]{0.48\textwidth}
\includegraphics[width=1.0\textwidth]{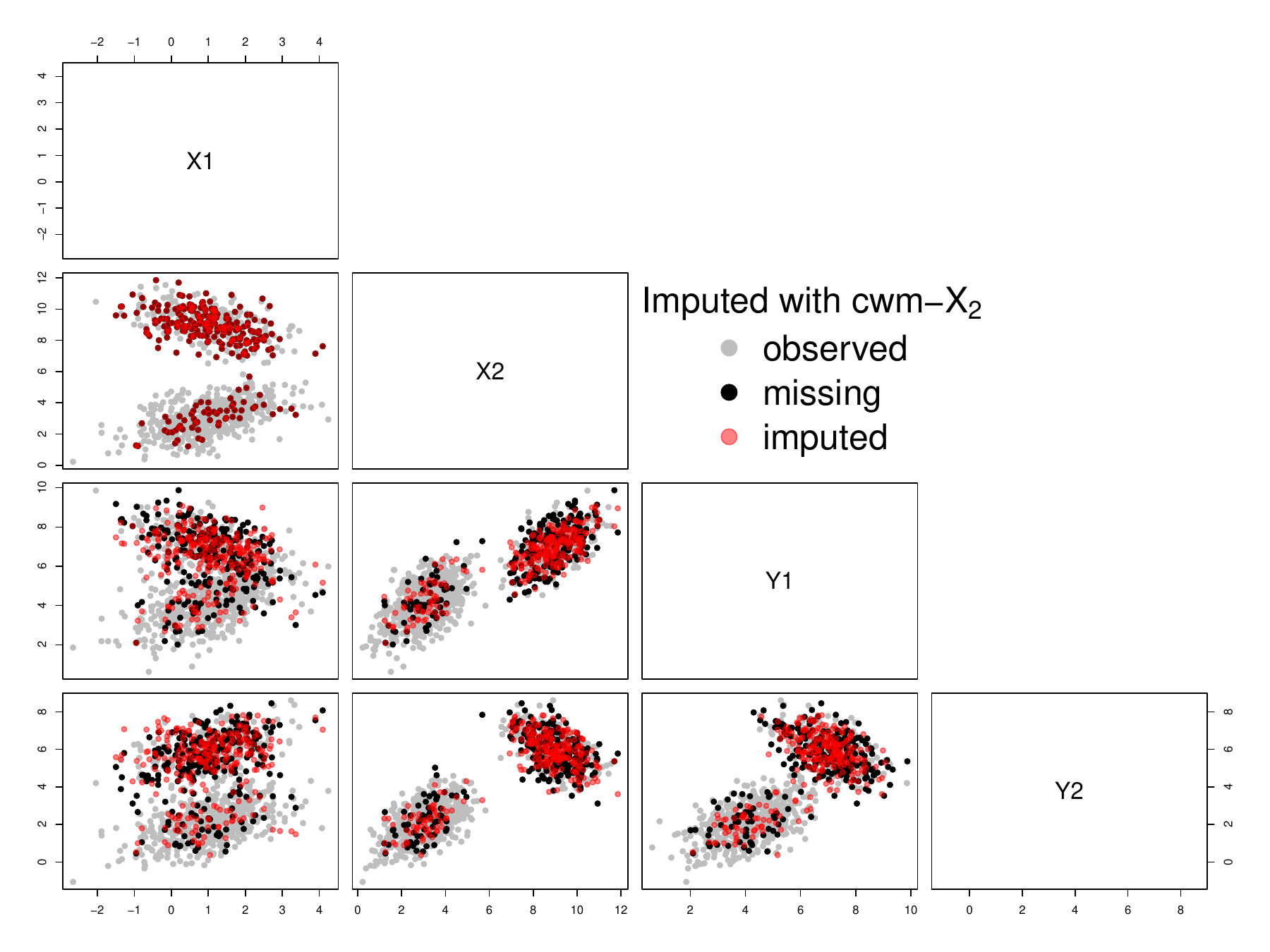}
\caption{Imputed with information from variable $X_2$}
\label{fig:pairsplot_cwmx2_multi}
\end{subfigure}
\hfill
\begin{subfigure}[]{0.48\textwidth}
\includegraphics[width=1.0\textwidth]{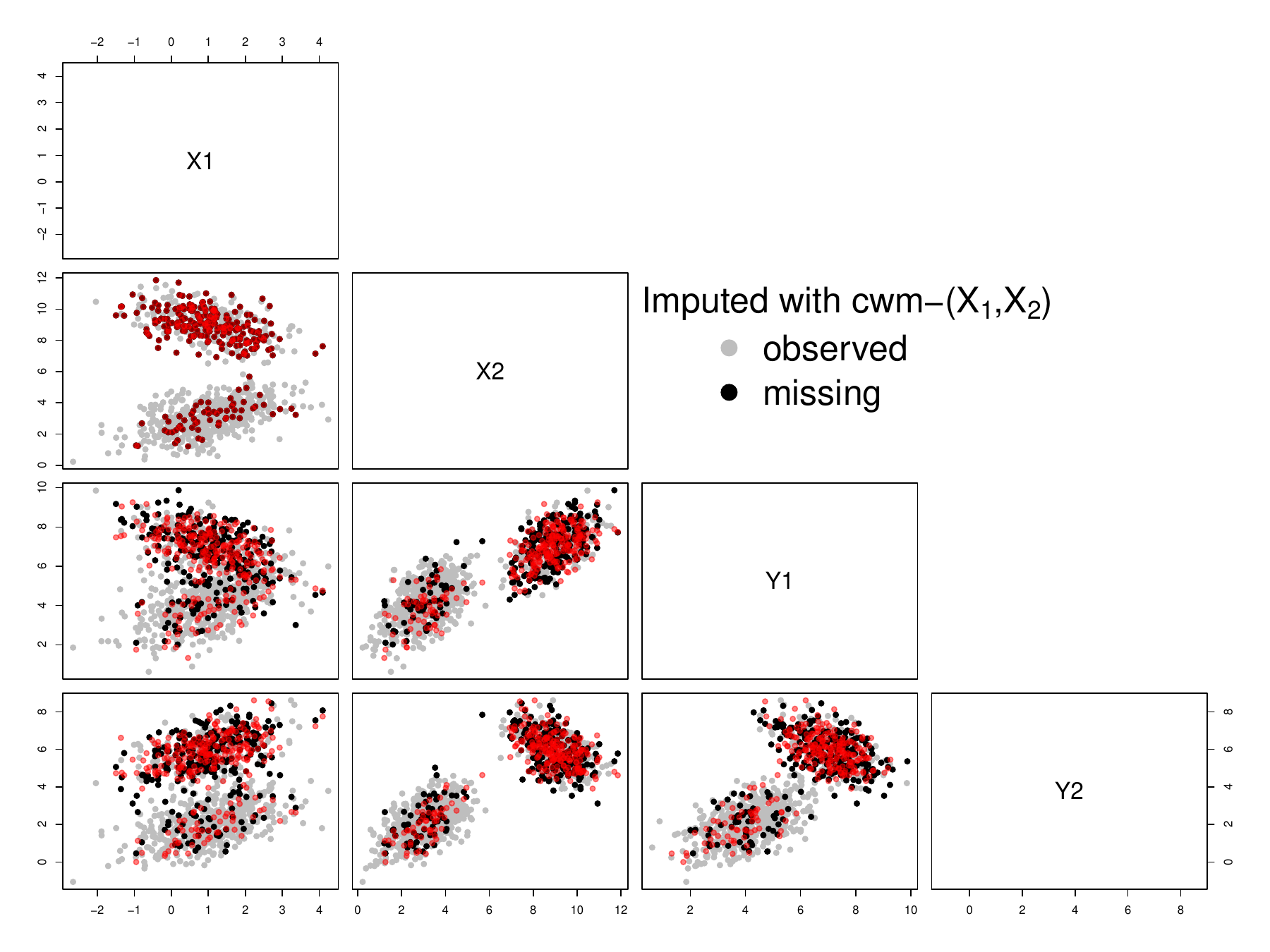}
\caption{Imputed with information from vector $(X_1,X_2)$.}
\label{fig:pairsplot_cwmx1x2_multi}
\end{subfigure}
\caption[Pairwise plots of the variables in the simulated multivariate database.]{Pairwise plots of the simulated multivariate database imputed by Gaussian LCWM with information on different types of variables.}
\label{fig:pairsplot_mult}
\end{figure}

Further analysis to support the findings described above are presented in the \href{JMVA_Supplementar.pdf}{Supplementary Material, Section} \ref{A-sec1:performance}. There, we present heat maps to assess both the missing data set and the complete data set. The results corroborate and strengthen our conclusions.

\subsection{Gaussian LCWM performance relative to other imputation methods}
\label{subsec:LCWM_meanpmmnorm_mult}

Next, we compared our imputation model, the Gaussian LCWM ({\ttfamily cwm}), with other procedures. We first consider the mean method ({\ttfamily mean}) \citep{paiva2017stop}, which does not incorporate any additional information to perform the imputation. We also consider two other Bayesian methods that create imputations using regression models, the \textit{predictive mean matching} ({\ttfamily pmm}) \citep{van2018flexible} and \textit{Bayesian multiple imputation} ({\ttfamily norm}) \citep{little2019statistical}. For these last two procedures, we use as auxiliary information the same input variables as in the previous section, that is, variables $X_1$, $X_2$ and the vector $(X_1,X_2)$. \figurename~\ref{fig:pairsplotnorm_mult} and \figurename~\ref{fig:pairsplotpmm_mult} depict the results of the imputations conducted with the provided auxiliary information when applying the {\ttfamily norm} and {\ttfamily pmm} imputation methods, respectively.

\begin{figure}[htb]
\centering
\begin{subfigure}[]{0.48\textwidth}
\includegraphics[width=1.0\textwidth]{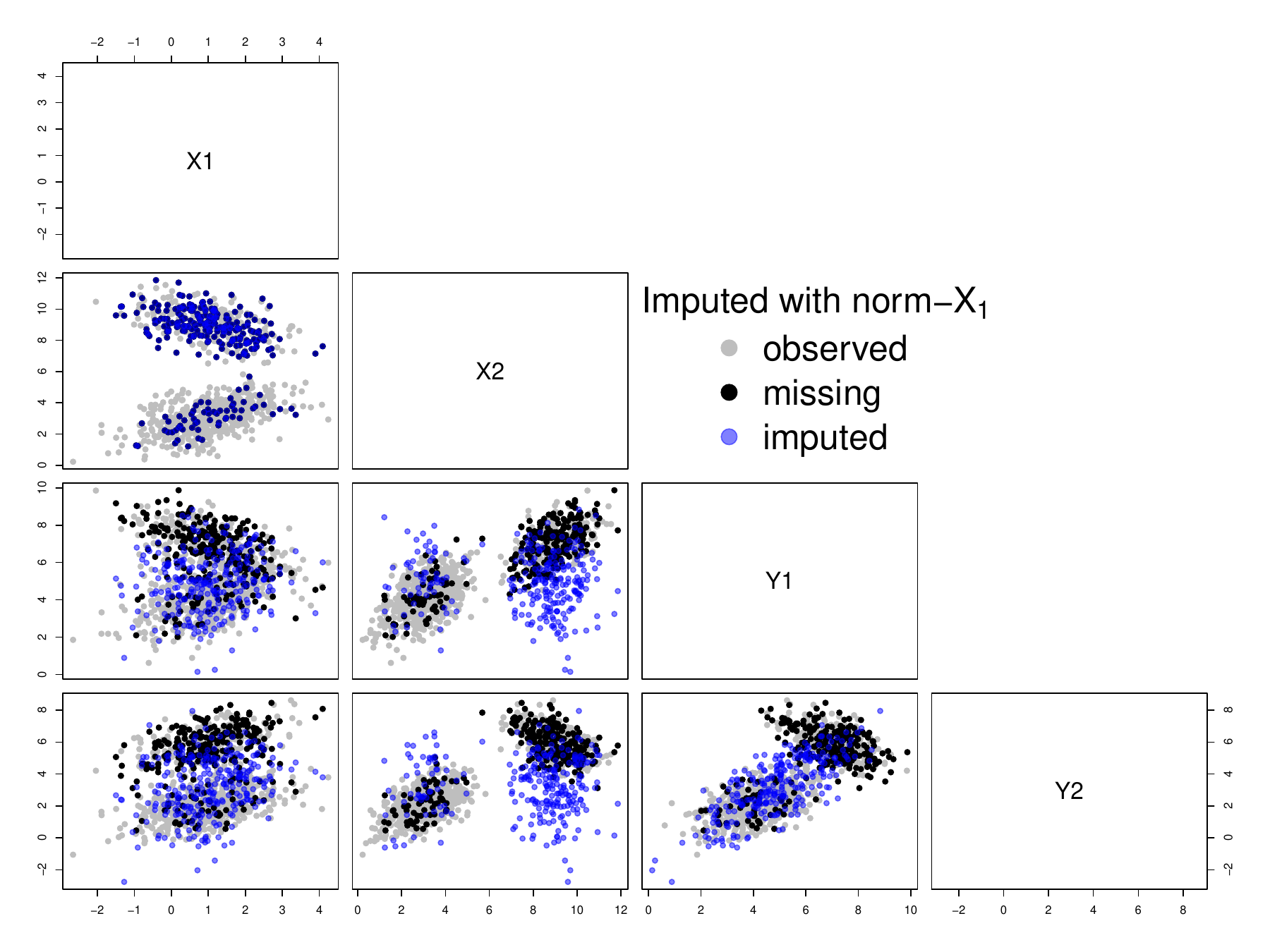}
\caption{Imputed with information from variable $X_1$}
\label{fig:pairsplot_normx1_multi}
\end{subfigure}
\hfill
\begin{subfigure}[]{0.48\textwidth}
\includegraphics[width=1.0\textwidth]{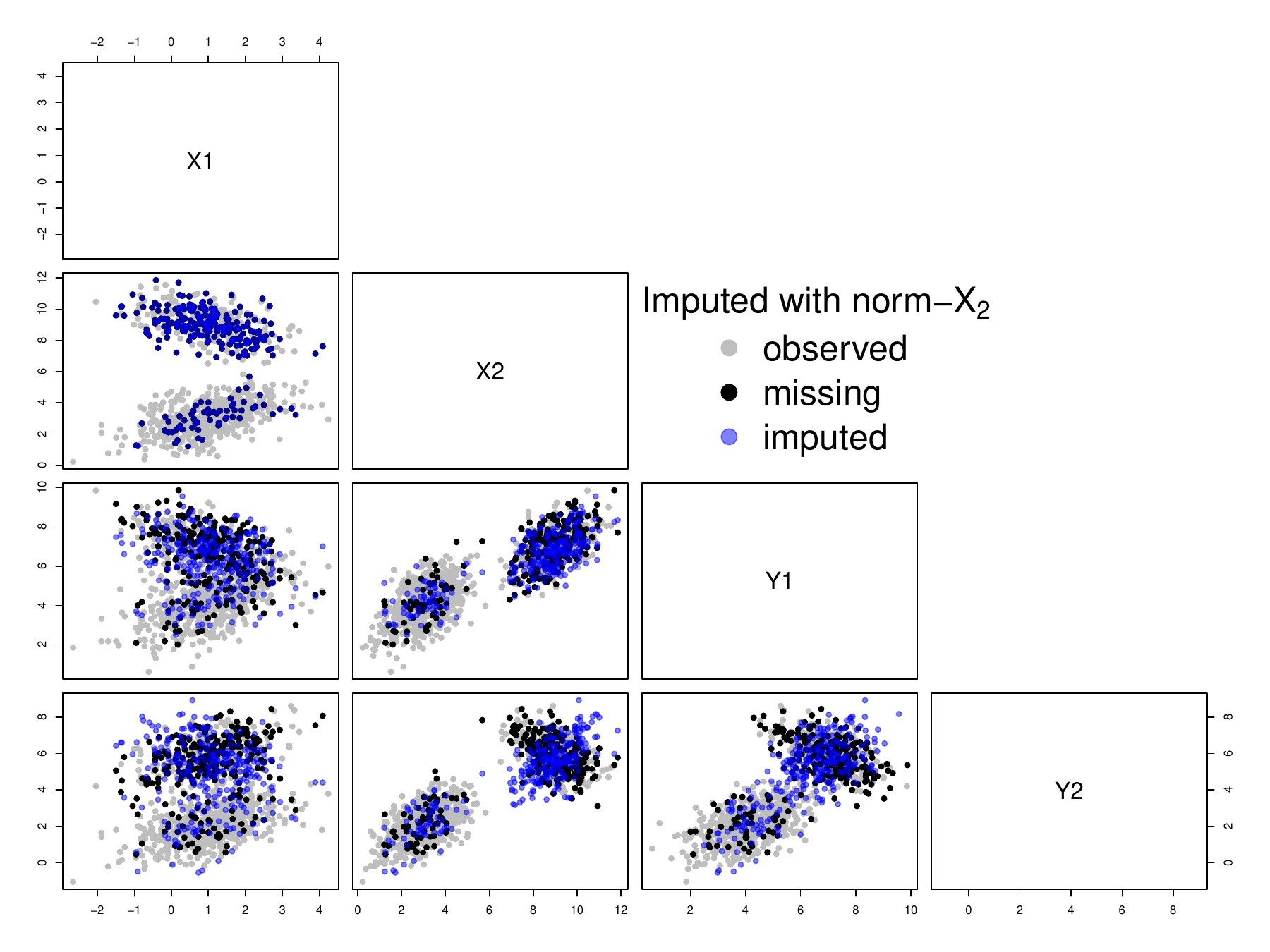}
\caption{Imputed with information from variable $X_2$}
\label{fig:pairsplot_normx2_multi}
\end{subfigure}
\hfill
\begin{subfigure}[]{0.48\textwidth}
\includegraphics[width=1.0\textwidth]{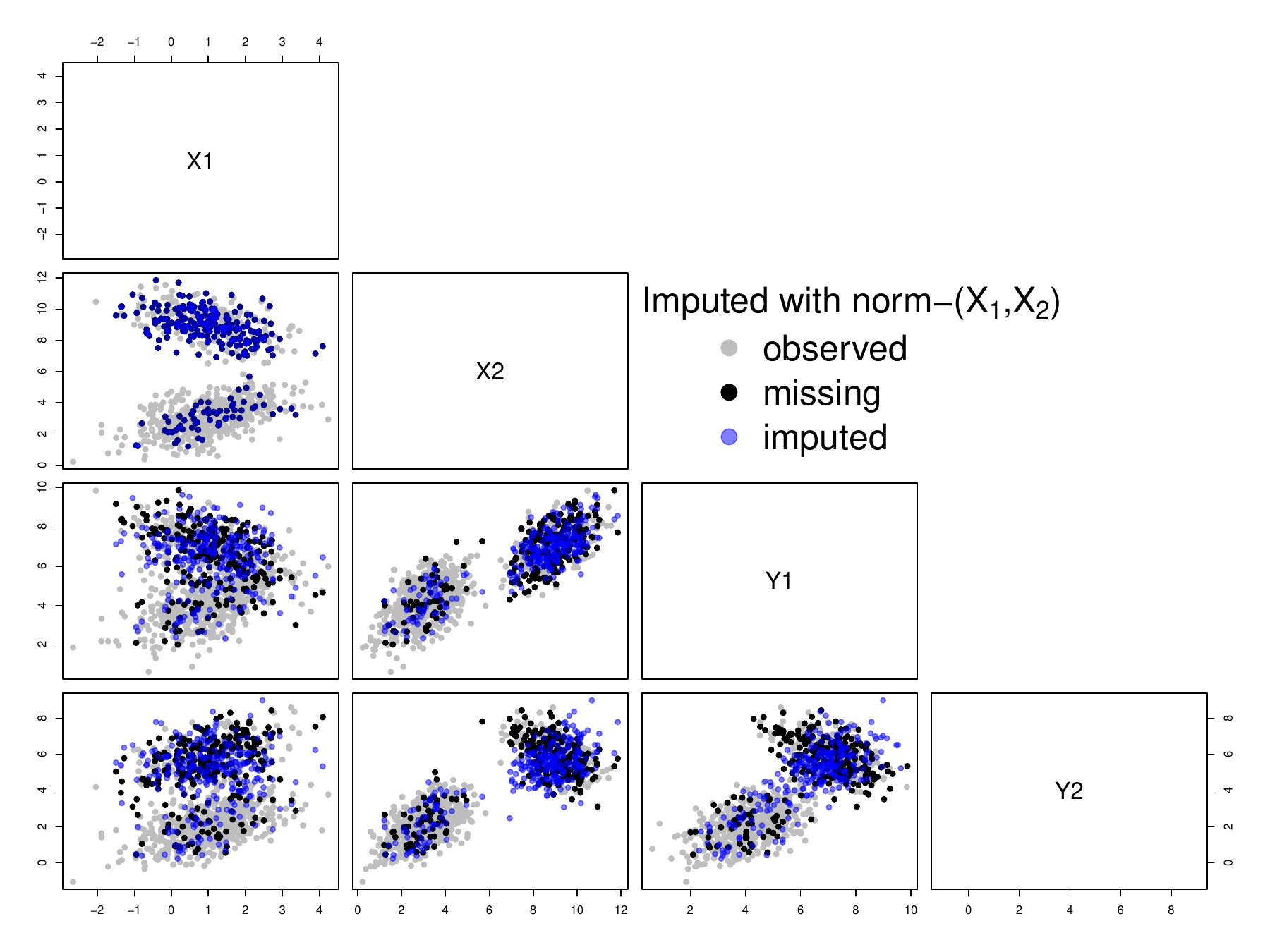}
\caption{Imputed with information from vector $(X_1,X_2)$.}
\label{fig:pairsplot_normx1x2_multi}
\end{subfigure}
\caption[Pairwise plots of the variables in the simulated multivariate database.]{Pairwise plots of the simulated multivariate database imputed by {\ttfamily norm} methodology with information on different types of variables.}
\label{fig:pairsplotnorm_mult}
\end{figure}
The behavior of the input variables in relation to the different imputation methods demonstrates a consistent pattern. When the variable $X_1$ is used, it provides poor information to the imputation procedure. However, there is a significant improvement in the imputations when utilizing the variable $X_2$ and the vector $(X_1,X_2)$. This is supported by \figurename~\ref{fig:pairsplotnorm_mult} which illustrate the results of the imputation process using the {\ttfamily norm} method. Upon comparing the imputations performed by the {\ttfamily cwm} method in 
\figurename~\ref{fig:pairsplot_mult} with those carried out by the {\ttfamily norm} method in \figurename~\ref{fig:pairsplotnorm_mult}, specifically in the panels depicting the $Y_1 \times Y_2$ projection planes, it is evident that the {\ttfamily cwm} method accurately imputes data within the two corresponding regions, whereas the {\ttfamily norm} methodology incorrectly assigns some imputations between the two regions. This visual analysis suggests superior performance of our model. To further substantiate this claim, a quantitative evaluation of the imputation processes will be conducted.

\begin{figure}[htb]
\centering
\begin{subfigure}[]{0.48\textwidth}
\includegraphics[width=1.0\textwidth]{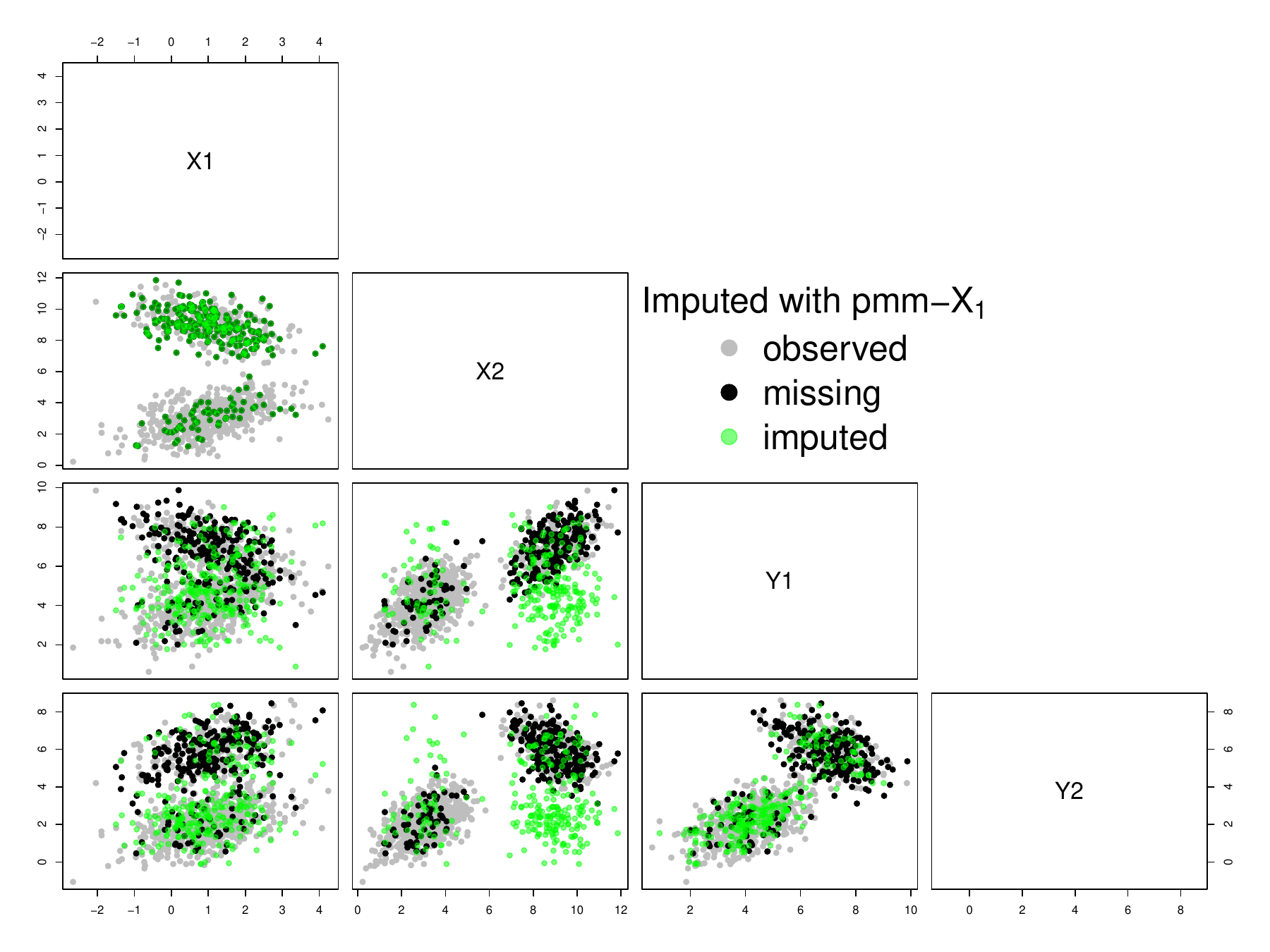}
\caption{Imputed with information from variable $X_1$}
\label{fig:pairsplot_pmmx1_multi}
\end{subfigure}
\hfill
\begin{subfigure}[]{0.48\textwidth}
\includegraphics[width=1.0\textwidth]{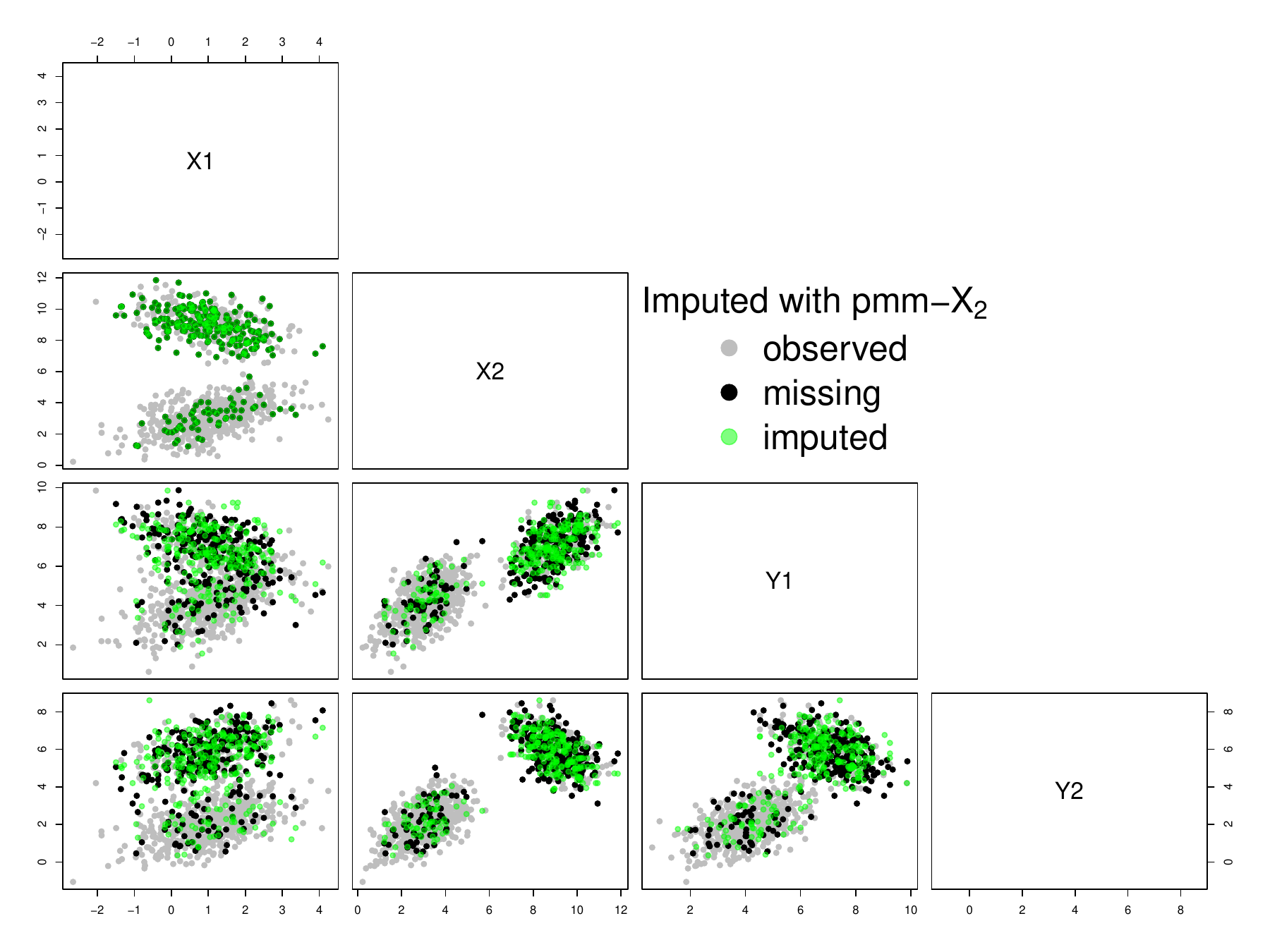}
\caption{Imputed with information from variable $X_2$}
\label{fig:pairsplot_pmmx2_multi}
\end{subfigure}
\hfill
\begin{subfigure}[]{0.48\textwidth}
\includegraphics[width=1.0\textwidth]{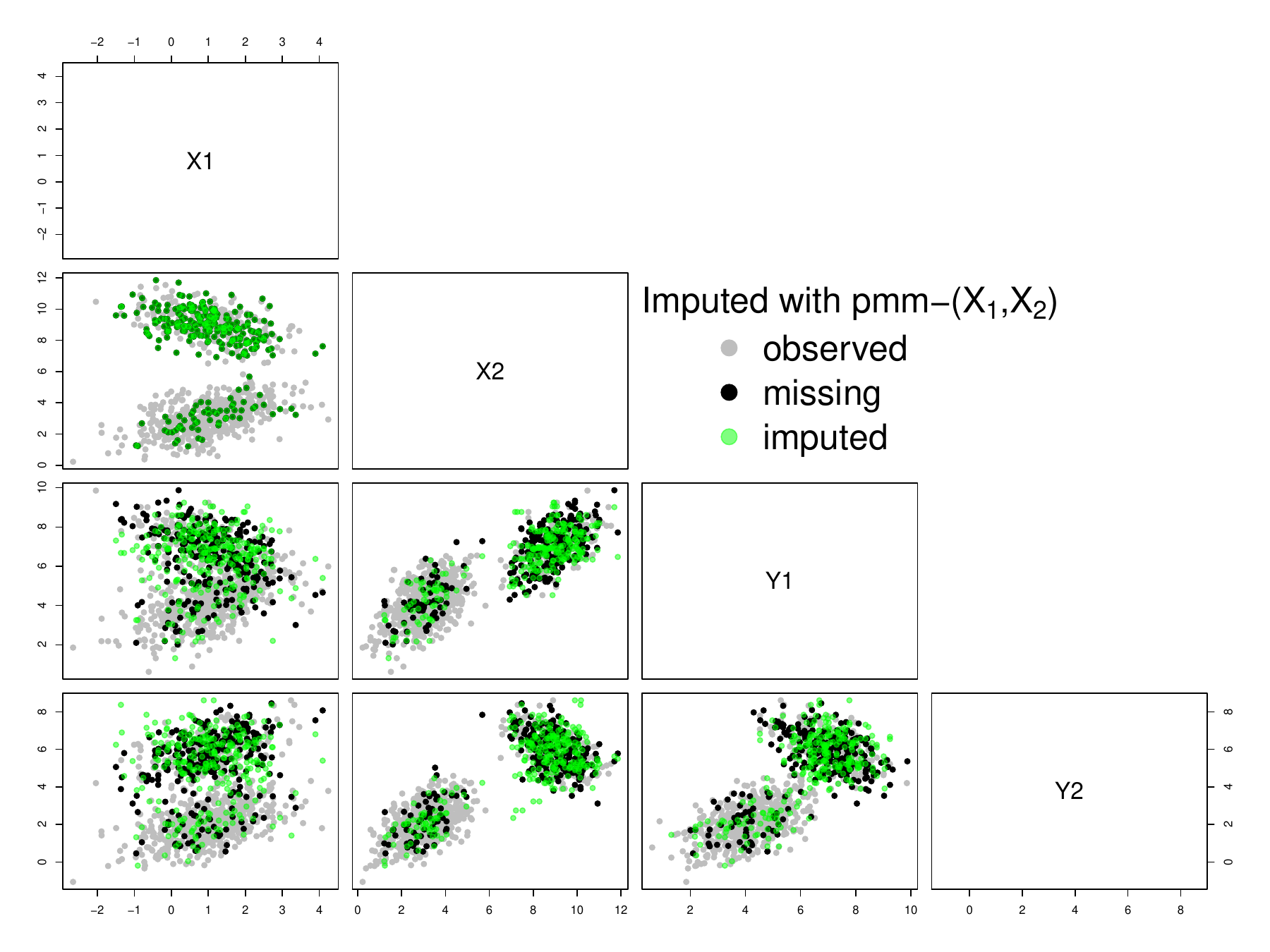}
\caption{Imputed with information from vector $(X_1,X_2)$.}
\label{fig:pairsplot_pmmx1x2_multi}
\end{subfigure}
\caption[Pairwise plots of the variables in the simulated multivariate database.]{Pairwise plots of the simulated multivariate database imputed by {\ttfamily pmm} methodology with information on different types of variables.}
\label{fig:pairsplotpmm_mult}
\end{figure}

The methodology implemented using the {\ttfamily pmm} imputation procedure closely followed the \texttt{cwm} approach. A visual analysis, comparing the graphs in \figurename~\ref{fig:pairsplot_mult} and \figurename~\ref{fig:pairsplotpmm_mult}, allows us to conclude that the {\ttfamily pmm} methodology, in addition to responding in a similar manner to the information provided through different input variables, demonstrates adequate performance when supplied with high-quality information such as variable $X_2$ and the vector $(X_1, X_2)$. It accurately imputes the data within the regions where there is observed information and does so in the appropriate proportions for each component. However, a few imputations that deviate from this characterization can be observed in the projection on the plane $X_2 \times Y_2$ when we impute using information from the vector $(X_1, X_2)$. One advantage of our methodology is that the imputed values are generated within the regions established by the observed data and are very close to the missing values. 

Regarding the {\ttfamily mean} methodology, which does not include additional information, this procedure imputes observations in regions where no missing data occurred, and the proportions of imputed data in each component appear to be influenced solely by the amount of observed data in each component. For plots similar to 
\figurename~\ref{fig:pairsplot_mult} to \figurename~\ref{fig:pairsplotpmm_mult}, illustrating the results of the {\ttfamily mean} procedure, please refer to the \href{JMVA_Supplementar.pdf}{Supplementary Material, Section} \ref{A-simulation_other-methods_multivariate}. In the next section, a quantitative procedure to determine the capacity to approximate the fully observed data will be established based on the Kullback-Liebler divergence \citep{kullback1951information}. The use of the divergence will enable us to evaluate and compare the various methodologies implemented for the imputation process, as well as to diagnose the type of incoming information.

\subsection{Quantitative diagnosis of imputation processes}
\label{KL-simulation}

To assess the accuracy of the imputation process, we employed the Kullback-Liebler divergence \citep{kullback1951information}. This measure quantifies the difference between two probability functions and is non-symmetric in nature. For multivariate continuous data, the computation of the Kullback-Liebler divergence involves integrating over the domain of the variables, as shown below:

\begin{equation}
\text{KL}(f,g):=\int_{\mathbb{R}^{d+p}}f(\bm{w})\ln\frac{f(\bm{w})}{g(\bm{w})}d\bm{w}.
\label{ecuac:div_KL}
\end{equation}

In the context provided, the density functions $f(\cdot)$ and $g(\cdot)$ represent two probability distributions, defined on $\mathbb{R}^{d+p}$. The Kullback-Liebler divergence, denoted as $\text{KL}(f,g)$, quantifies the loss of information that occurs when we approximate the distribution $f$ using the distribution $g$. Nevertheless, when dealing with Gaussian FMM distributions for $f(\cdot)$ and $g(\cdot)$, the expression in Equation (\ref{ecuac:div_KL}) becomes intractable. To address this challenge, researchers such as \citet{hershey2007approximating} and \citet{durrieu2012lower} have proposed various approaches that employ approximations and bounds for the divergence. The purpose of utilizing the Kullback-Leibler divergence is to evaluate the performance of the imputation methods by comparing the estimated distributions derived from the imputed databases with the original distribution. This enables us to evaluate the quality of the imputations performed using the different methods.

The process of calculating this measure begins by imputing the dataset using various methods. Subsequently, we estimate the parameters that describe the data distribution, specifically focusing on the output variables. Since there is no closed-form expression available for computing the KL divergence in the case of Gaussian FMM, we employ an approximation approach utilizing Monte Carlo methods \citep{hershey2007approximating}. To denote this approximation, we will use the notation $\text{KL}_{\text{mc}}$. The results obtained with the different imputation methods are presented in Table \ref{Cuadro:KL_bases_imputadas_multivariada}. The table's top row presents a 95\% quantile interval for the KL divergence values. These values are calculated based on $N = 10000$ sampled datasets, each with a size of $n = 1000$, representing the true distribution of the output variables $(Y_1, Y_2)$. The parameters of these datasets are estimated using the \texttt{mixsmsn} package \citep{prates2013mixsmsn}, and the KL divergence is computed for each case. If the KL divergence value falls within the quantile interval, it is considered as a successful recovery of the original distribution, denoted as \texttt{WI} (Within Interval) in the relative distance column. The values in this column are determined by measuring the relative distance between the KL divergence value and the upper limit of the quantile interval.

\begin{table}[ht]
\resizebox{15.0cm}{!} {
\begin{tabular}{llcclcclcc}
\cline{3-4} \cline{6-7} \cline{9-10}
&  & \multicolumn{2}{c}{\ttfamily cwm} &  & \multicolumn{2}{c}{\ttfamily pmm}  &  & \multicolumn{2}{c}{\ttfamily norm} \\ 
\cline{3-4} \cline{6-7} \cline{9-10} 
&  & $\text{KL}_{\text{mc}}$ & Relative distance &  & $\text{KL}_{\text{mc}}$ & Relative distance &  & $\text{KL}_{\text{mc}}$ & Relative distance \\ 
\cline{1-1} \cline{3-4} \cline{6-7} \cline{9-10} 
Qu.int. 95.0\% &  & (0,0.0107) & -  &  & (0,0.0107) & -  &  & (0,0.0107) & - \\ 
\cline{1-1} \cline{3-4} \cline{6-7} \cline{9-10} 
$\phantom{-}\bm{g}_{(Y_1,Y_2)_{\text{com}}}$ &  & 0.0081 & \phantom{-}{\ttfamily WI} &  & 0.0081  & \phantom{-}{\ttfamily WI} &  & 0.0081 & \phantom{-}{\ttfamily WI}  \\
$\phantom{-}\bm{g}_{(Y_1,Y_2)_{\text{obs}}}$ &  & 0.0358 & \phantom{-}3.36 &  & 0.0358 & \phantom{-}3.36 &  & 0.0358 & \phantom{-}3.36 \\
$\phantom{-}\bm{g}_{(Y_1,Y_2)_{\text{mean}}}$ &  & 0.3140 & 29.46 &  & - & - &  & - & - \\
$\phantom{-}\bm{g}_{(Y_1,Y_2)_{X_1}}$ &  & 0.0342 & \phantom{-}3.22 &  & 0.3092 & 29.00 &  & 0.3114 & 29.21 \\
$\phantom{-}\bm{g}_{(Y_1,Y_2)_{X_2}}$ &  & 0.0149 & \phantom{-}1.40 &  & 0.0181 & \phantom{-}1.75 &  & 0.0620 & \phantom{-}5.81 \\
$\phantom{-}\bm{g}_{(Y_1,Y_2)_{(X_1, X_2)}}$ &  & 0.0112 & \phantom{-}1.05 &  & 0.0247 & \phantom{-}2.31 &  & 0.0505 & \phantom{-}4.74 \\ 
\cline{1-1} \cline{3-4} \cline{6-7} \cline{9-10} 
\end{tabular}
}
\caption[KL divergence and relative distance for imputed databases using different methodologies.]{The KL divergence and relative distance are presented for the complete dataset, observed data, and imputed databases using different types of input variables. This analysis is performed for each of the methodologies being compared. Additionally, information for the {\ttfamily mean} procedure is included.}
\label{Cuadro:KL_bases_imputadas_multivariada}
\end{table}

The results in the {\ttfamily cwm} column of Table \ref{Cuadro:KL_bases_imputadas_multivariada} provide quantitative confirmation of the analyses conducted in Section \ref{subsec:LCWM_X1_X2_X1X2_mult}. The distributions that better recover the original distribution are those that were imputed using the information from the variable $X_2$ and the vector $(X_1,X_2)$. These distributions utilize information from variables or vectors that are distributed separately among components. On the other hand, the imputation process using only the variable $X_1$ performs poorly. However, constructing an input vector that includes at least one variable that are distributed separately among components enables the generation of a vector with a distinct distribution separately among components, resulting in a more precise imputation. This was observed when constructing an input vector with the variables $X_1$ and $X_2$ using the various imputation processes.

Section \ref{subsec:LCWM_meanpmmnorm_mult} aimed to compare various imputation methods with the Gaussian LCWM. An initial comparison was conducted using the KL divergence when comparing our model with the imputation method that does not utilize auxiliary information, referred to as the {\ttfamily mean} imputation method. We expected our model to perform better due to the use of additional information. The columns labeled {\ttfamily pmm} and {\ttfamily norm} in Table \ref{Cuadro:KL_bases_imputadas_multivariada} present the KL divergence values for similar cases implemented in our model, involving imputation of the output variables $(Y_1, Y_2)$ with information from $X_1$, $X_2$, and $(X_1, X_2)$. Similar results to those observed for the {\ttfamily cwm} method were obtained within each methodology. Comparing the three procedures, our model exhibited superior performance in all three specific situations examined. Even when using auxiliary information from the variable $X_1$, the three procedures displayed inferior performance, but our imputation model exhibited the least information loss. Furthermore, in scenarios involving $X_2$, we observed a relative improvement of at least 20\% and 76\% from {\ttfamily cwm} to the {\ttfamily ppm} and {\ttfamily norm}, respectively. 

\section{Illustrative example: Fisher’s Iris Data}
\label{Ejemplo}

The \textsc{Fisher’s Iris Data} is a multivariate dataset collected by \citet{anderson1935irises} and first analyzed by \citet{fisher1936use}. The dataset comprises 50 samples from each of three species of Iris: Iris setosa, Iris virginica, and Iris versicolor. For each sample, four features were measured: the sepal length, sepal width, petal length, and petal width, all measured in centimeters. In the context of finite mixture distributions, the IRIS database is extensively utilized as it offers an outstanding platform for experimentation and method development in this field. Authors such as \citet{everitt1981finite}, \citet{mclachlan2019finite}, \citet{fruhwirth2006finite}, among others, utilize this database to illustrate various procedures presented in their respective investigations.

The database variables were divided into two groups. The first group consisted of the output variables, namely {\ttfamily Sepal.Length} and {\ttfamily Petal.Width}, for which the missing data pattern was simulated. The second group included the input variables, {\ttfamily Sepal.Width} and {\ttfamily Petal.Length}, which were assumed to be fully observed. In the dataset, each of the three species was treated as a separate cluster. To evaluate the imputation processes, we made the assumption that the true number of components in the complete dataset, denoted as $C$, was equal to three, which also coincided with the number of species.

The distribution of variables in the \textsc{Fisher’s Iris Data} is shown in \figurename~\ref{fig:pairwise_iris_species}. The behavior of the input variables resembles the scenarios depicted in the simulated data. The variable {\ttfamily Petal.Length} displays a distribution separated into two groups. There is a clearly defined group corresponding to the setosa species, while the versicolor and virginica species form a single cluster, although they occupy well-defined regions. On the other hand, the variable {\ttfamily Sepal.Width} does not provide clear information on which of the components to impute from. Since the missing values will be imputed using information from both variables, the input vector inherits the feature of having a separate distribution between components from the {\ttfamily Petal.Length} variable. Two scenarios were simulated for missing data: one utilizing a MAR mechanism, and the other employing a MNAR mechanism. The databases were imputed using our model as well as the three methodologies discussed throughout the document.

\begin{figure}[htb]
    \centering
	\includegraphics[width=0.9\textwidth]{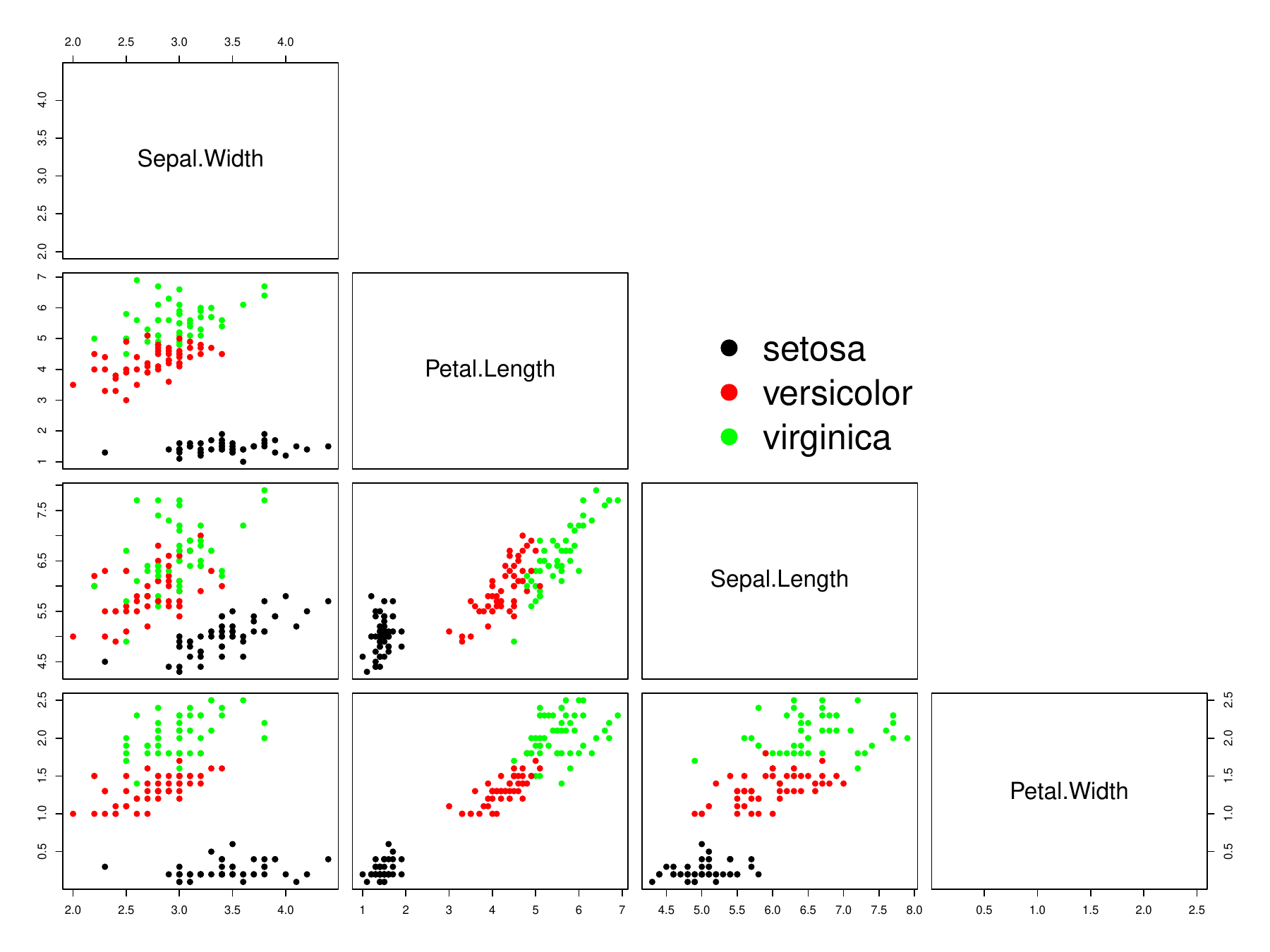}
	\caption[Pairwise plots of the variables in the {\ttfamily iris} database discriminated by species]{Pairwise plot of the variables in the {\ttfamily iris} database discriminated by species. At the top left and bottom right, scatter plots of the input variables and output variables are shown respectively. In the lower left corner, the crosses of these variables are displayed.}
	\label{fig:pairwise_iris_species}
\end{figure}

\subsection{Simulation of missing data under MAR mechanism}
\label{MAR_IRIS}

To generate a dataset with missing values under the MAR mechanism, distinct probabilities of missingness were assigned to each species. Specifically, 30\% of missing data was simulated for the setosa species, 20\% for the versicolor species, and 10\% for the virginica species. By establishing these proportions, an overall missing data rate of 20\% is expected in the complete database. The distribution of observed, missing, and complete data categorized by species is presented in Table \ref{cuadro:MAR_iris_faltantes}.

\begin{table}[htb]
\centering
\begin{tabular}{c|cl|cl|cl}
\multicolumn{1}{l|}{} & \multicolumn{2}{c|}{\textbf{observed}}  & \multicolumn{2}{c|}{\textbf{missing}}  & \multicolumn{2}{c}{\textbf{complete}}  \\ 
\hline
\multirow{1}{*}{\textbf{setosa}} & 36 & \textit{(72.0\%)} & 14 & \textit{(28.0\%)} & \phantom{-}50 & \textit{(100\%)} \\
& \multicolumn{1}{l}{\textit{(30.5\%)}} &   & \multicolumn{1}{l}{\textit{(43.8\%)}} &                   & \multicolumn{1}{l}{\textit{(33.3\%)}} &\\ 
\hline
\multirow{1}{*}{\textbf{versicolor}}  & 37 & \textit{(74.0\%)} & 13 & \textit{(26.0\%)} & \phantom{-}50 & \textit{(100\%)} \\
& \multicolumn{1}{l}{\textit{(31.4\%)}} &  & \multicolumn{1}{l}{\textit{(40.6\%)}} &  & \multicolumn{1}{l}{\textit{(33.3\%)}} & \\ 
\hline
\multirow{1}{*}{\textbf{virginica}}  & 45 & \textit{(90.0\%)} & \phantom{-}5 & \textit{(10.0\%)} & \phantom{-}50 & \textit{(100\%)} \\
& \multicolumn{1}{l}{\textit{(38.1\%)}} &  & \multicolumn{1}{l}{\textit{(15.6\%)}} &  & \multicolumn{1}{l}{\textit{(33.3\%)}} & \\ 
\hline
\multirow{1}{*}{\textbf{total}}   & 118 & \textit{(78.7\%)} & 32 & \textit{(21.3\%)} & 150                                  & \textit{(100\%)} \\
& \multicolumn{1}{l}{\textit{(100\%)}}  &                   & \multicolumn{1}{l}{\textit{(100\%)}}  &                   & \multicolumn{1}{l}{\textit{(100\%)}}  &                 
		\end{tabular}
\caption{Distribution of simulated missing data for the {\ttfamily iris} database under the MAR mechanism.}
\label{cuadro:MAR_iris_faltantes}
\end{table}

The database underwent imputation using the selected methods. Specifically, the variables {\ttfamily Sepal.Length} and {\ttfamily Petal.Width} were imputed by incorporating the available information from these variables, as well as the fully observed variables {\ttfamily Sepal.Width} and {\ttfamily Petal.Length}. The code developed by us to implement the Gaussian LCWM was also used for the {\ttfamily mean} and {\ttfamily norm} imputation procedures. In this context, the {\ttfamily mean} method is obtained by not considering any input variables, meaning there is no additional information. Consequently, the conditional model only considers the intercept as a parameter for each cluster. Regarding the {\ttfamily norm} procedure, we can implement it by assuming a single cluster, thus setting the value of $G=1$ in our code. The {\ttfamily pmm} procedure was implemented using the \texttt{MICE} package \citep{mice2011}. For our {\ttfamily cwm} model, we set the number of clusters to ten ($G=10$). 

We present the pairwise plot for the variables in the imputed database using the Gaussian LCWM, as shown in \figurename~\ref{fig:pairwise_irisMAR_impute_cwm}. Each panel displays the observed, missing, and imputed values. The imputed data effectively capture the regions where missing data occur. The upper left panel shows a scatter plot for the input variables, where the red points are superimposed on the black points since this information is considered known in the imputation procedure. In the lower right panel, the output variables with missing information are plotted, revealing the imputations made by the model. The four panels in the lower left corner present the cross between the input and output variables, illustrating the imputation process.

\begin{figure}[htb]
    \centering
	\includegraphics[width=0.9\textwidth]{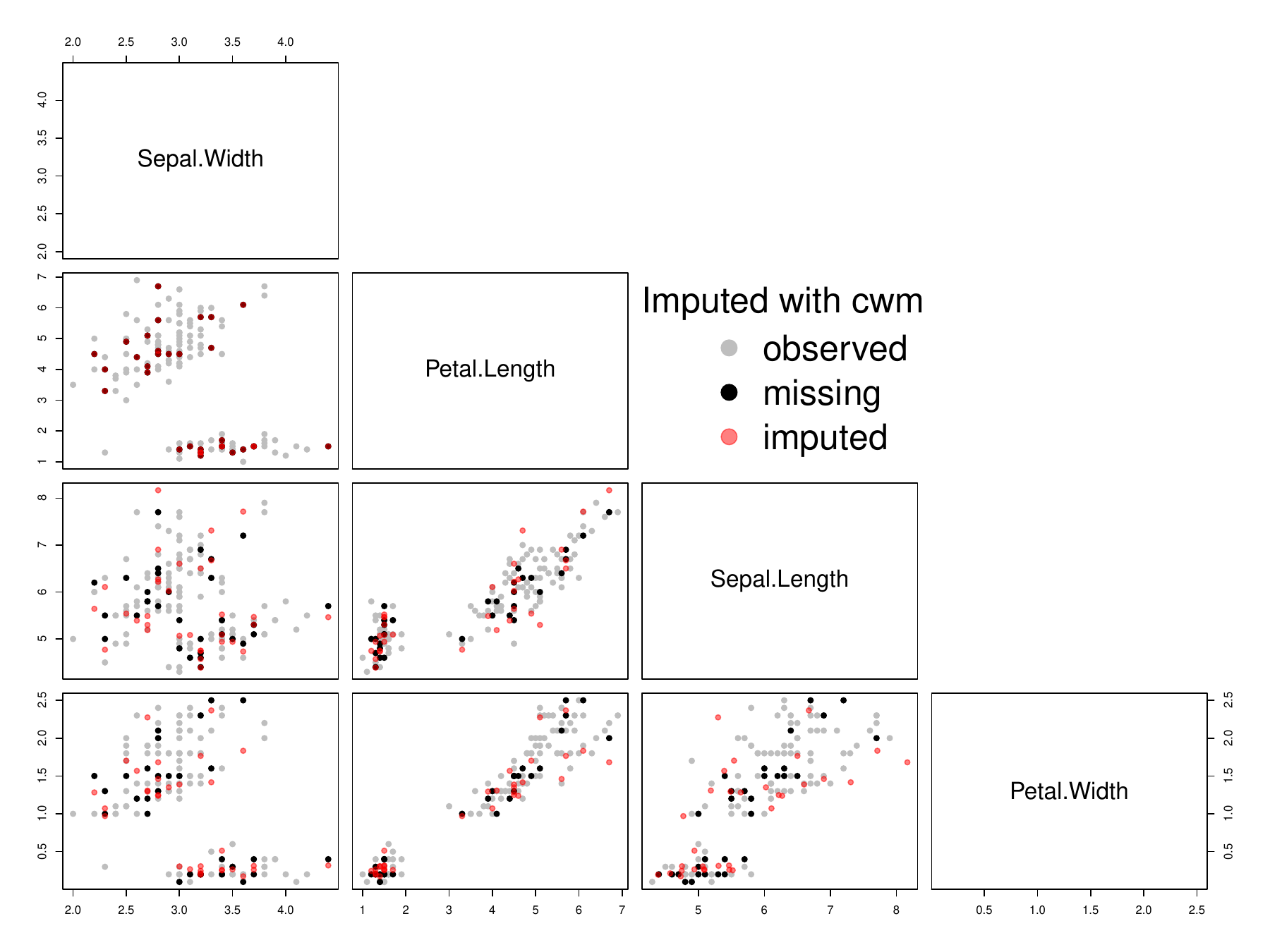}
	\caption[Pairwise plot of the imputed {\ttfamily iris} database using Gaussian LCWM. MAR mechanism.]{Pairwise plot of the imputed {\ttfamily iris} database using the Gaussian LCWM. Each panel presents the crossing of two of the variables specifying observed values, missing values and imputed values. Missing data generated using a MAR mechanism.}
	\label{fig:pairwise_irisMAR_impute_cwm}
\end{figure}

Pairwise plots of the same type for each imputation method are displayed in the \href{Supplementar.pdf}{Supplementary Material, Subsection} \ref{A-apend_c_sec_irisMAR}. These graphs reveal that the {\ttfamily pmm} procedure yields results that are highly similar to our methodology. The imputed data effectively cover the regions where missing data were generated, exhibiting comparable proportions. However, while the {\ttfamily norm} method is the next best-performing procedure, it does exhibit imputations outside the region of observed data. Similarly, the {\ttfamily mean} method, to a greater extent, demonstrates poor performance, especially considering that this method does not utilize input variables.

A quantitative analysis is conducted using the KL divergence values presented in Table \ref{cuadro:KL_irisMAR}. As a baseline, we employed a Gaussian FMM with a fixed number of components, $C = 3$, using the complete database. In contrast, for the imputation procedures, we utilized databases generated by each of the imputed methods of interest. The estimation procedure was executed with a {\ttfamily burn-in = 10000} and an {\ttfamily effectiveSize = 1000}. With values of the estimates of the mixture probabilities, mean vectors, and variance-covariance matrices, and employing a procedure based on Monte Carlo methods, we obtained the KL divergence values.

\begin{table}[htb]
\centering
\begin{tabular}{llcc}
\cline{3-4}
\multicolumn{1}{l}{}     &  & \multicolumn{2}{c}{\textbf{Approach method}}\\ 
\cline{3-4} 
\multicolumn{1}{l}{}     &  & \multicolumn{1}{l}{$\text{KL}_{\text{mc}}$} & \multicolumn{1}{l}{Relative distance}\\ 
\cline{1-1} \cline{3-4}  
$\bm{\text{{\ttfamily (PW,SL)}}}_{\text{com}}$   &  & -      & -      \\
$\bm{\text{{\ttfamily (PW,SL)}}}_{\text{obs}}$   &  & 0.0431 & 1.00   \\
$\bm{\text{{\ttfamily (PW,SL)}}}_{\text{{\ttfamily mean}}}$  &  & 0.1280 & 2.97   \\	
$\bm{\text{{\ttfamily (PW,SL)}}}_{\text{{\ttfamily cwm}}}$   &  & 0.0252 & 0.58   \\
$\bm{\text{{\ttfamily (PW,SL)}}}_{\text{{\ttfamily pmm}}}$   &  & 0.0396 & 0.92   \\
$\bm{\text{{\ttfamily (PW,SL)}}}_{\text{{\ttfamily norm}}}$  &  & 0.0429 & 1.00   \\
\cline{1-1} \cline{3-4}  
\end{tabular}
\caption[KL divergences for the imputed {\ttfamily iris} database. Missing data generated from a MAR mechanism.]{KL divergences for the imputed {\ttfamily iris} database. Relative distances taken with reference to the estimated distribution of observed data. Missing data generated from a MAR mechanism.}
\label{cuadro:KL_irisMAR}
\end{table}

In this case, the KL divergence between the observed data and the complete data serves as our reference. We consider one unit as the relative distance, meaning that imputed datasets with relative distances less than one indicate a close approximation to the complete database. According to this criterion, the results presented in Table \ref{cuadro:KL_irisMAR} demonstrate that the imputation process without auxiliary information, using the {\ttfamily mean} method, performed the worst, as anticipated. Among the methods that utilized auxiliary information, the {\ttfamily cwm} method showed the best performance, followed by the {\ttfamily pmm} method. The {\ttfamily norm} procedure maintained a similar level of information loss compared to the distribution consisting solely of observed data.

\subsection{Simulation of missing data under MNAR mechanism}
\label{MNAR_IRIS}

A second scenario simulated for the missing pattern was based on a MNAR mechanism. The mechanism is generated in such a way that larger values of the output variables \texttt {Sepal.Length} and \texttt{Petal.Width} are more likely to be missing.
An indicator variable $R_i \sim \text{Bern}(\theta_i)$ was simulated with missingness probability $\theta_i$ for the $i$-th observation. This probability is obtained from the expression $\theta_i=\text{logit}^{-1}(\beta_0+\beta_1 y_i)$ for $i \in \{1.\ldots,n \}$ with $\beta_0=-20.4$ and $\beta_1=3.0$. The distribution of missing data within the database is summarized in Table \ref{cuadro:MNAR_iris_faltantes} and establishes the amount of observed, missing, and complete data for each species. Due to the fact that the highest values of the output variables are typically observed in the Virginica species, followed by the Versicolor species, and the lowest values in the Setosa species, the missing data is primarily concentrated in the Virginica category, with 22 instances of missing data. The Versicolor group exhibits the second highest number of missing data cases (7), while the Setosa species has the fewest missing data instances (2).

\begin{table}[htb]
\centering
\begin{tabular}{c|cl|cl|cl}
\multicolumn{1}{l|}{} & \multicolumn{2}{c|}{\textbf{observed}}  & \multicolumn{2}{c|}{\textbf{missing}}  & \multicolumn{2}{c}{\textbf{complete}}  \\ 
\hline
\multirow{1}{*}{\textbf{setosa}} & 48 & \textit{(96.0\%)} & \phantom{--}2 & \textit{\phantom{--}(4.0\%)} & \phantom{-}50 & \textit{(100\%)} \\
& \multicolumn{1}{l}{\textit{(40.3\%)}} &   & \multicolumn{1}{l}{\textit{\phantom{--}(6.4\%)}} & & \multicolumn{1}{l}{\textit{(33.3\%)}} &\\ 
\hline
\multirow{1}{*}{\textbf{versicolor}}  & 43 & \textit{(86.0\%)} & \phantom{--}7 & \textit{(14.0\%)} & \phantom{-}50 & \textit{(100\%)} \\
& \multicolumn{1}{l}{\textit{(36.1\%)}} &  & \multicolumn{1}{l}{\textit{(22.6\%)}} &  & \multicolumn{1}{l}{\textit{(33.3\%)}} & \\ 
\hline
\multirow{1}{*}{\textbf{virginica}}  & 28 & \textit{(56.0\%)} & 22 & \textit{(44.0\%)} & \phantom{-}50 & \textit{(100\%)} \\
& \multicolumn{1}{l}{\textit{(23.5\%)}} &  & \multicolumn{1}{l}{\textit{(71.0\%)}} &  & \multicolumn{1}{l}{\textit{(33.3\%)}} & \\ 
\hline
\multirow{1}{*}{\textbf{total}}   & 119 & \textit{(80.0\%)} & 31 & \textit{(20.0\%)} & 150 & \textit{(100\%)} \\
& \multicolumn{1}{l}{\textit{(100\%)}}  &   & \multicolumn{1}{l}{\textit{(100\%)}}  &  & \multicolumn{1}{l}{\textit{(100\%)}}  &                 
\end{tabular}
\caption{Distribution of simulated missing data for the {\ttfamily iris} database under the MNAR mechanism.}
\label{cuadro:MNAR_iris_faltantes}
\end{table}

As in the case of missing data under MAR in the previous section, the output variables {\ttfamily Sepal.Length} and {\ttfamily Petal.Width} are imputed with information from the fully observed variables {\ttfamily Sepal.Width} and {\ttfamily Petal.Length}. 
\figurename~\ref{fig:pairwise_irisMNAR_impute_cwm} showcases the pairwise plots of the imputed database utilizing our {\ttfamily cwm} methodology. It is evident that our model imputes data in various regions with proportions that align with the patterns of missing data generation. Specifically, regions with higher variable values exhibit more imputations, corresponding to the number of missing observations. In the region characterized by the lowest values of the variable, a pair of missing values was generated, and our model successfully imputed the same number of values. Our imputation process effectively covers all regions where missing data was generated, which does not appear to be the case with the {\ttfamily norm} and {\ttfamily pmm} methods. In a specific portion of the region occupied by the Virginica species, the data completely vanished, leading to the absence of imputed data in that area for both procedures. According to \citet{van2018flexible}, although the {\ttfamily pmm} procedure is generally robust, this particular situation can pose a problem as it may result in the {\ttfamily pmm} method imputing the same observations repeatedly. The pairwise plots for the remaining imputation procedures can be found in the \href{Supplementar.pdf}{Supplementary Material, Subsection} \ref{A-apend_c_sec_irisMNAR}.

\begin{figure}[htb]
    \centering
	\includegraphics[width=0.9\textwidth]{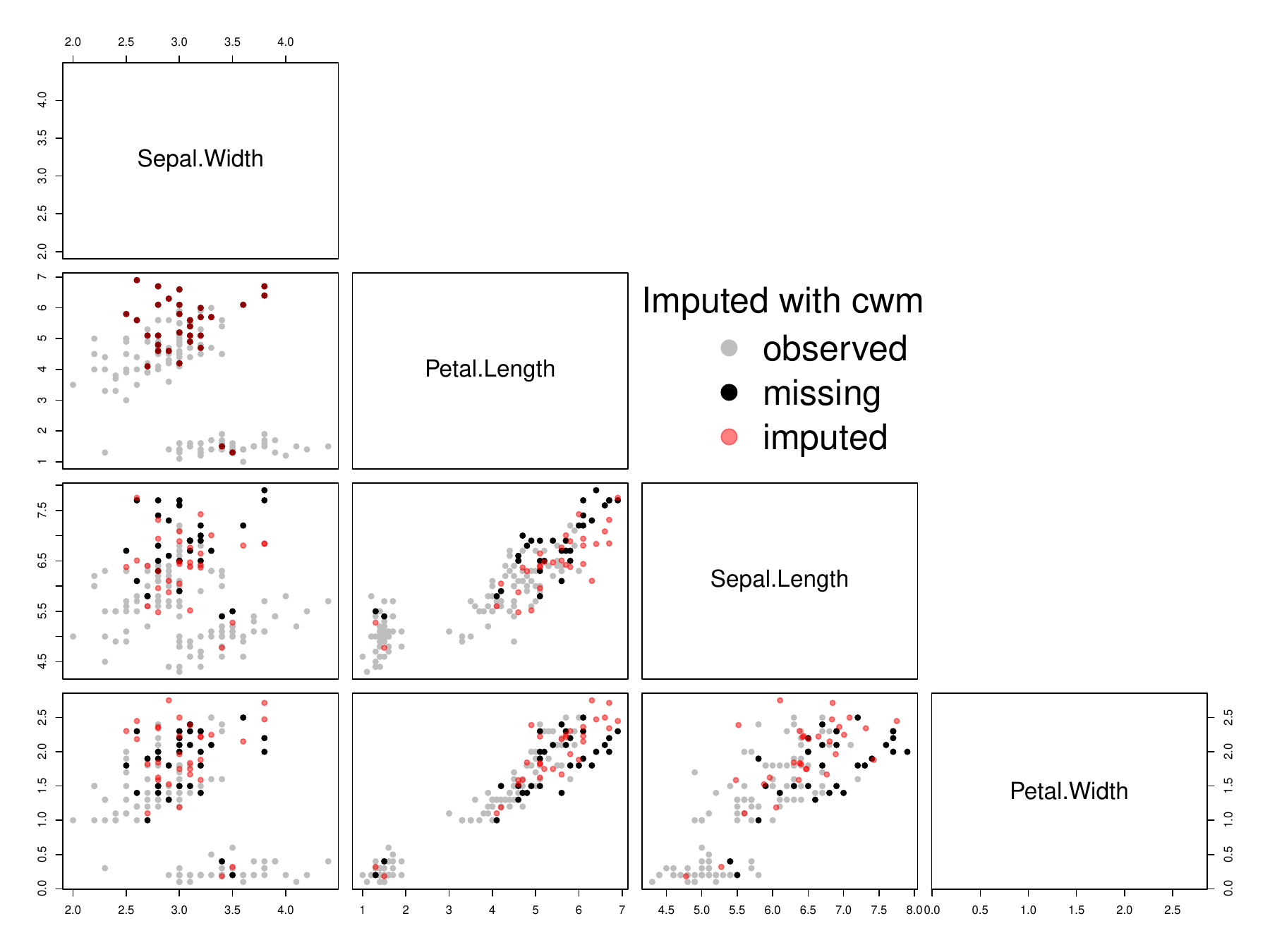}
	\caption[Pairwise plot of the imputed {\ttfamily iris} database using Gaussian LCWM. MNAR mechanism.]{Pairwise plot of the imputed {\ttfamily iris} database using the Gaussian LCWM. Each panel presents the crossing of two of the variables specifying observed values, missing values and imputed values. Missing data generated using a MNAR mechanism.}
	\label{fig:pairwise_irisMNAR_impute_cwm}
\end{figure}

Once again, the analyses conducted using the graphs can be further supported by the KL divergence values presented in Table \ref{cuadro:KL_irisMNAR}. It is noteworthy that among all the procedures examined, only the {\ttfamily cwm} approach demonstrated a satisfactory performance. Specifically, the KL divergence value for the imputation methodology {\ttfamily pmm} closely resembled that of the observed data. On the other hand, the {\ttfamily mean} and {\ttfamily norm} procedures exhibited the poorest performance. \\

\begin{table}[htb]
\centering
\begin{tabular}{llcc}
\cline{3-4}
\multicolumn{1}{l}{}     &  & \multicolumn{2}{c}{\textbf{Approach method}}\\ 
\cline{3-4} 
\multicolumn{1}{l}{}     &  & \multicolumn{1}{l}{$\text{KL}_{\text{mc}}$} & \multicolumn{1}{l}{Relative distance}\\ 
\cline{1-1} \cline{3-4}  
$\bm{\text{{\ttfamily (PW,SL)}}}_{\text{com}}$   &  & -      & -      \\
$\bm{\text{{\ttfamily (PW,SL)}}}_{\text{obs}}$   &  & 0.1405 & 1.00   \\
$\bm{\text{{\ttfamily (PW,SL)}}}_{\text{{\ttfamily mean}}}$  &  & 0.2061 & 1.46   \\	
$\bm{\text{{\ttfamily (PW,SL)}}}_{\text{{\ttfamily cwm}}}$   &  & 0.0762 & 0.54   \\
$\bm{\text{{\ttfamily (PW,SL)}}}_{\text{{\ttfamily pmm}}}$   &  & 0.1410 & 1.00   \\
$\bm{\text{{\ttfamily (PW,SL)}}}_{\text{{\ttfamily norm}}}$  &  & 0.2186 & 1.56   \\
\cline{1-1} \cline{3-4}  
\end{tabular}
\caption[KL divergences for the imputed {\ttfamily iris} database. Missing data generated from a MNAR mechanism.]{KL divergences for the imputed {\ttfamily iris} database. Relative distances taken with reference to the estimated distribution of observed data. Missing data generated from a MNAR mechanism.}
\label{cuadro:KL_irisMNAR}
\end{table}

In conclusion, the \texttt{cwm} was evaluated using two simulated missing data patterns in the iris database. The first pattern was generated under a MAR mechanism, while the second followed a MNAR mechanism. In both scenarios, our imputation model exhibited superior performance compared to other imputation methods.
Initially, it outperformed a methodology that did not utilize auxiliary information. Subsequently, when compared to other relevant methods that also incorporated auxiliary information, our model performed comparably. It is worth noting that our model does not explicitly account for non-ignorable response mechanisms, despite the missing data being simulated from MAR and MNAR mechanisms. Its good performance in the face of missing data under a MNAR mechanism characterizes it as a robust method.

\section{Conclusions}
\label{Conclusiones}

We introduced a novel methodology known as Multivariate Gaussian Linear Cluster-Weighted Modeling for imputing continuous multivariate data. This methodology is particularly suitable in scenarios where the dataset exhibits a group structure, where the objective is to uncover such structure in the data, or when the distribution shapes of the data are unknown.
It builds upon the results of finite mixture model theory \citep{mclachlan2019finite, fruhwirth2006finite} and Weighted-Cluster modeling \citep{gershenfeld1997nonlinear,ingrassia2012local} with a specific focus on Gaussian distributions. We use a fully Bayesian approach that jointly models and imputes missing values using a flexible combination of Dirichlet processes of multivariate normal distributions. The imputation model is designed based on a non-response unit pattern where variables with missing information are called output variables. For all individuals, we assume that it is possible to find auxiliary information from other sources and these fully observed variables are called input variables.
By assuming the input information as observational, we develop a joint modeling approach for both input and output variables. This approach enables the model to dynamically utilize the input information to characterize the components and determine the appropriate component for imputation.

Based on the original model that did not consider additional information, we developed a model to include this information in order to improve the imputation procedure.
This improvement is achieved by selecting input variables that exhibit a strong correlation with the output variables and are characterized by being \textit{distributed separately between components}.
We compare the performance of the proposed model with other Bayesian procedures in the literature that use additional information to obtain the required imputations. The procedures were the \textit{predictive mean matching} and the \textit{Bayesian multiple imputation}. The simulated scenarios showed a better performance of our model in the case of databases with group structure.
In the different scenarios simulated by us, the \textit{predictive mean matching} was the method that showed results closer to our proposal. However, when we simulated an extreme scenario, where the missing data pattern considered censored data and with a group structure, we could observe the poor performance of this method compared to ours. In this regard, \citet{van2018flexible} states that the danger when using {\ttfamily pmm} is the duplication of the same donor value many times. Also, this problem is more likely to occur if the sample is small or when a considerable amount of data is missing relative to the observed data within a specific region of the predicted value.

Theoretical results for the multivariate model were established, generalizing the results presented by \citet{gershenfeld1997nonlinear} and \citet{ingrassia2012local} for the univariate case. 
The performance of different types of input variables for the imputation model was analyzed on a set of simulated data. 
These variables move between two extreme scenarios: variables that do not provide information on which component to impute from, and variables that are characterized by being distributed separately among components. Due to its performance, this last type of variable is considered desirable to be included as auxiliary information in the proposed imputation model. On the same simulated scenarios, the Gaussian LCWM was compared with the aforementioned imputation procedures, obtaining favorable results for our model. Similarly, our methodology was implemented on the Iris database. On this data set, two input variables and two output variables were considered, and two missing data patterns based on MAR and MNAR mechanisms were simulated. The three imputation procedures were implemented for comparison. Once again, our model showed a better performance compared to the other procedures in both scenarios. Our imputation procedure was implemented in \textsf{R} software, and the scripts can be downloaded from \url{https://github.com/lamasmelac/lcwm_impute}.

A potential future direction of this study is to extend the methodology to the case of finite mixture models with distributions other than the normal distribution. For example, Skew Normal or $t-$Mixture Models, among others. It would also be of interest to develop selection criteria for meaningful variables that can be used as auxiliaries, as that allows to make the best possible imputation. Another possible extension is to create a {\ttfamily cwm-pmm} hybrid method, that takes advantage of both methodologies and check if that would improve the imputation performance. Finally, in some situations, according to the researcher's knowledge, it might be of interest to allow imputation to be carried out in a particular region of the data set, determining an explicit MNAR procedure.

\section*{Acknowledgments}
\label{Agradecimientos}

This study was financed in part by the Coordenação de Aperfeiçoamento de Pessoal de Nível
Superior - Brasil (CAPES). Prates acknowledges partial funding support from Conselho Nacional
de Desenvolvimento Científico e Tecnológico (CNPq) grants 436948/2018-4 and 307547/2018-4 and Fundação de Amparo à Pesquisa do Estado de Minas Gerais (FAPEMIG) grant APQ-01837-22. Masmela-Caita acknowledges partial funding support from Universidad Distrital Francisco José de
Caldas.

\bibliography{bibliografia}

\end{document}


\maketitle

\section{Performance of the model when incorporating new information}
\label{sec1:performance}

In \href{output-subdoc/JMVA_Main.pdf}{Main Material, Section}~\ref{B-subsec:LCWM_X1_X2_X1X2_mult}, we discuss the performance of the model when incorporating new information. Heatmaps showcasing the imputed output variables of the dataset using the Gaussian LCWM under different types of input variables are presented in \figurename~\ref{fig:mapas_calor_multi}. Specifically, we focus on heatmaps for imputed values when including the variables $X_1$, $X_2$, and the vector $(X_1,X_2)$ as auxiliary information in our model. These heatmaps are compared with a map illustrating the behavior of the missing data as simulated, as shown in \figurename~\ref{fig:mapas_calor_obs_falt_multi}. 

Upon analysis, we observe similarities between the maps representing the original missing data and the imputed data using information from variable $X_2$ and the vector $(X_1, X_2)$. However, the map obtained from the imputed database with variable $X_1$ differs from the map corresponding to the missing data, indicating the model's poorer performance in capturing the characteristics of this variable. The heatmaps presented in \figurename~\ref{fig:mapas_calor_bases_multi} illustrate the same scenarios, comparing the complete database with the imputed datasets. A higher level of similarity is observed in the heatmaps for complete and imputed data when considering variable $X_2$ and the vector $(X_1, X_2)$. However, a slight difference is noted in the heat map for the data imputed with variable $X_1$, specifically in the values corresponding to cluster 2.

\begin{figure*}[htb]
	\centering
	\begin{subfigure}[]{0.48\textwidth}
		\includegraphics[width=1.0\textwidth]{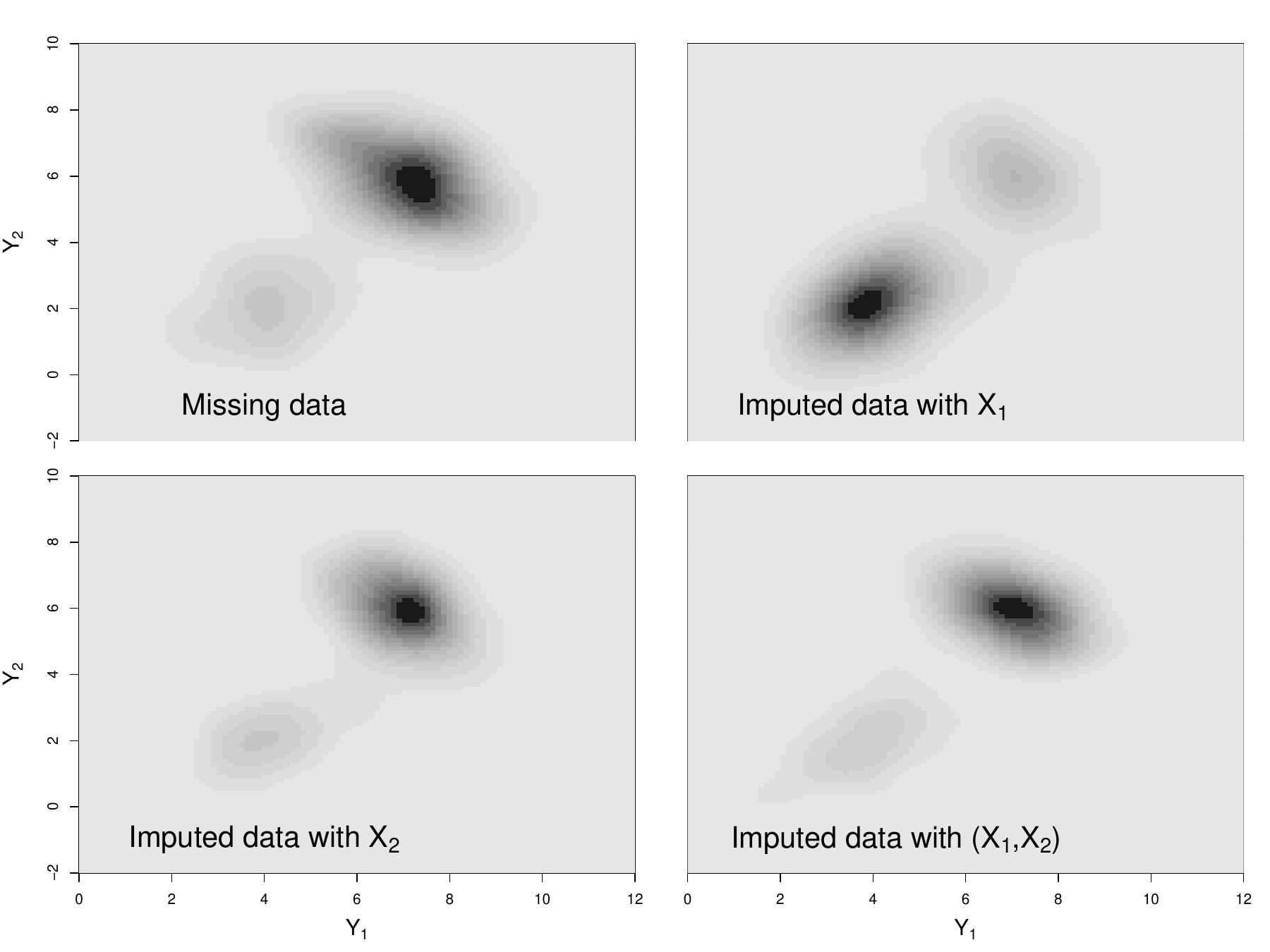}
		\caption{Missing data}
		\label{fig:mapas_calor_obs_falt_multi}
	\end{subfigure}
	\hfill
	\begin{subfigure}[]{0.48\textwidth}
		\includegraphics[width=1.0\textwidth]{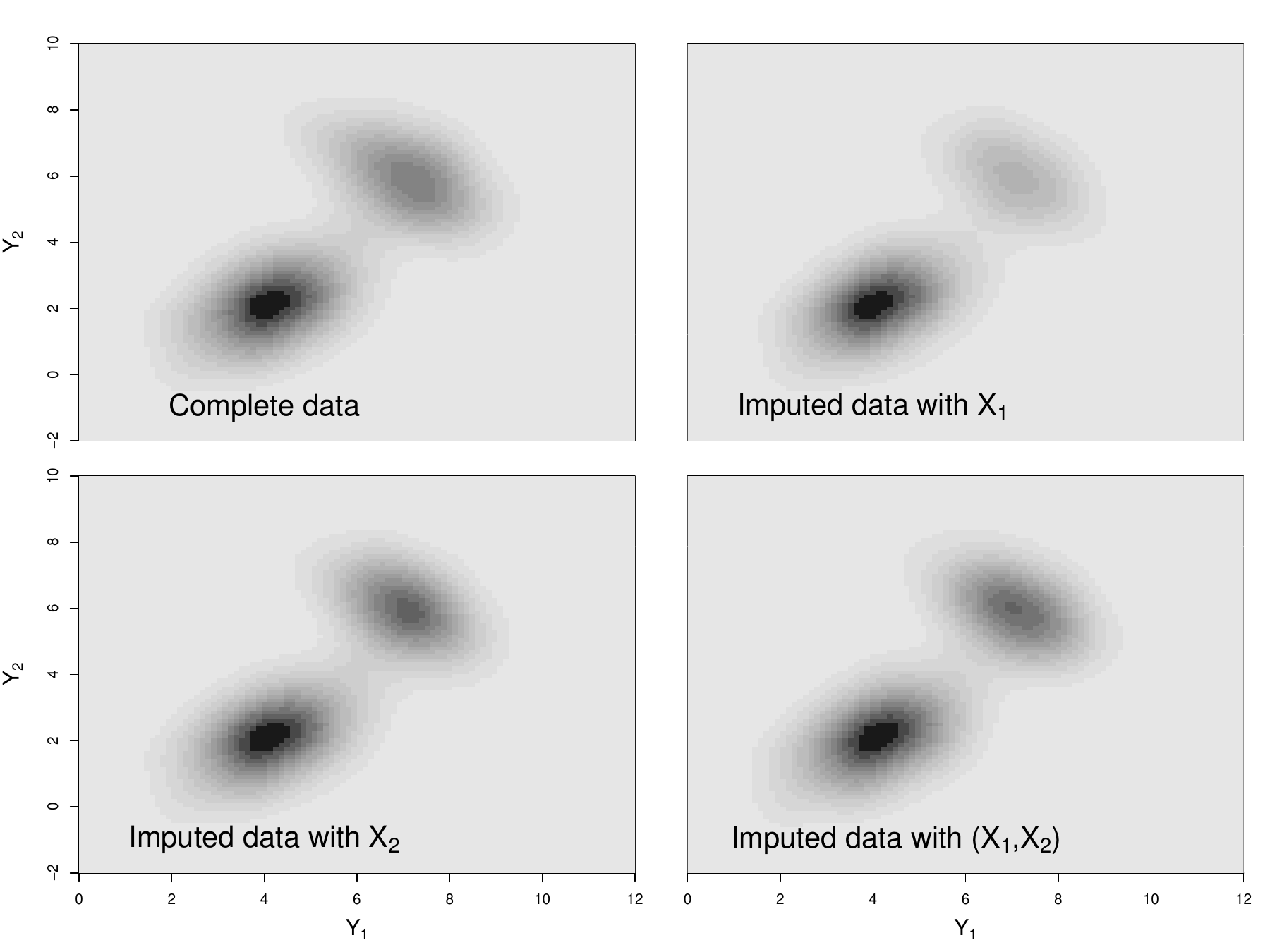}
		\caption{Complete data}
		\label{fig:mapas_calor_bases_multi}
	\end{subfigure}
	\caption[Bivariate distribution heatmaps of database imputed using Gaussian LCWM.]{Bivariate distribution heatmaps of database imputed using Gaussian LCWM with information on the variables $X_1$, $X_2$ and the vector $(X_1, X_2)$.}
	\label{fig:mapas_calor_multi}
\end{figure*}

\section{Imputation with {\ttfamily mean}, {\ttfamily norm} and {\ttfamily pmm} for the multivariate case}
\label{simulation_other-methods_multivariate}

In \href{output-subdoc/JMVA_Main.pdf}{Main Material, Section}~\ref{B-subsec:LCWM_meanpmmnorm_mult}, we conduct a comparative analysis of our imputation model, the Gaussian LCWM ({\ttfamily cwm}), with other techniques commonly found in statistical literature. Initially, we consider two Bayesian methods that use regression models for imputation and also include additional information from fully observed variables: the \textit{matching of predictive means} ({\ttfamily pmm}) \citep{van2018flexible} and \textit{Bayesian multiple imputation} ({\ttfamily norm}) \citep{little2019statistical}. We present pairwise plots and KL divergence values to assess and compare our methodology with these two approaches. In this section we apply a similar analysis comparing our methodology with the mean method ({\ttfamily mean}) \citep{paiva2017stop}, which does not incorporate any additional information for the imputation process and serves as a baseline for our proposed methodology.
\figurename~\ref{fig:mutiv_mean} illustrates the procedure followed by the {\ttfamily mean} methodology. The results obtained from our imputation model using the information from variable $X_1$ were comparable to those of the {\ttfamily mean} methodology. In general, the {\ttfamily mean} methodology tended to impute observations in regions where no missing data occurred, and the proportions of imputed data in each component appeared to be influenced solely by the amount of observed data in each component.

\begin{figure}[H]
	\centering
	\includegraphics[width=0.9\textwidth]{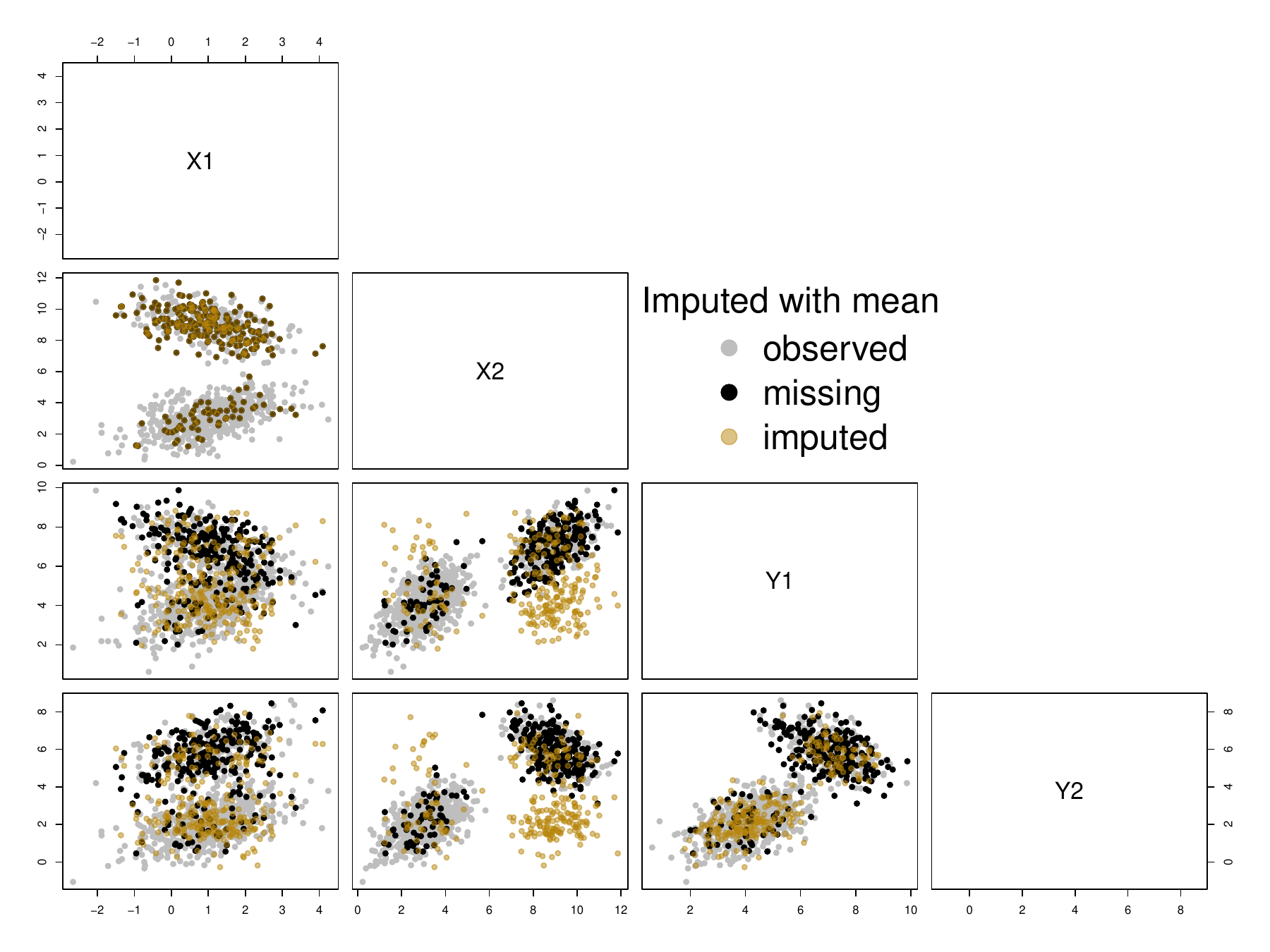}
	\caption{{Pair plots of the simulated multivariate database imputed by {\ttfamily mean} methodology.}}
	\label{fig:mutiv_mean}
\end{figure}

Table~\ref{B-Cuadro:KL_bases_imputadas_multivariada} in \href{output-subdoc/JMVA_Main.pdf}{Main Material, Subsection}~\ref{B-KL-simulation} quantitatively evaluates the various imputation methodologies using KL divergence. The KL divergence value for the {\ttfamily mean} method, along with the relative distance, demonstrates that the imputation process can be enhanced by employing the {\ttfamily cwm} methodology with supplementary information. By considering the KL divergence value for the mean method alongside the relative distance, we can observe that incorporating additional information through the cwm methodology improves the imputation process.


\section{Illustrative example: Fisher’s Iris Data} 
\label{xxxx} 

\href{output-subdoc/JMVA_Main.pdf}{Main Material, Section}~\ref{B-Ejemplo} showcases the application of the Gaussian LCWM in a multivariate setting using the well-known \textsc{Fisher's Iris Data}, a multivariate dataset collected by \citet{anderson1935irises} and first analyzed by \citet{fisher1936use}. In this study, we simulated two types of missing data mechanisms: Missing At Random (MAR) and Missing Not At Random (MNAR). Our methodology was implemented alongside the imputation procedures {\ttfamily mean}, {\ttfamily norm}, and {\ttfamily pmm}, allowing for a comprehensive comparison of these different approaches.

To evaluate the performance of the imputation procedures, we present pairwise plots that visually depict the results obtained by the {\ttfamily mean}, {\ttfamily norm}, and {\ttfamily pmm} methods. Additionally, we provide quantitative assessments using the KL divergence for each scenario. In the main paper, we include a discussion of our methodology, {\ttfamily cwm}, based on scatterplots and KL divergence values presented in \href{output-subdoc/JMVA_Main.pdf}{Main Material, Tables}~\ref{B-cuadro:KL_irisMAR} and \ref{B-cuadro:KL_irisMNAR}.

\subsection{Iris database imputation with missing data MAR mechanism}
\label{apend_c_sec_irisMAR}

\begin{figure}[htb]
	\centering
	\begin{subfigure}[]{0.48\textwidth}
		\includegraphics[width=1.0\textwidth]{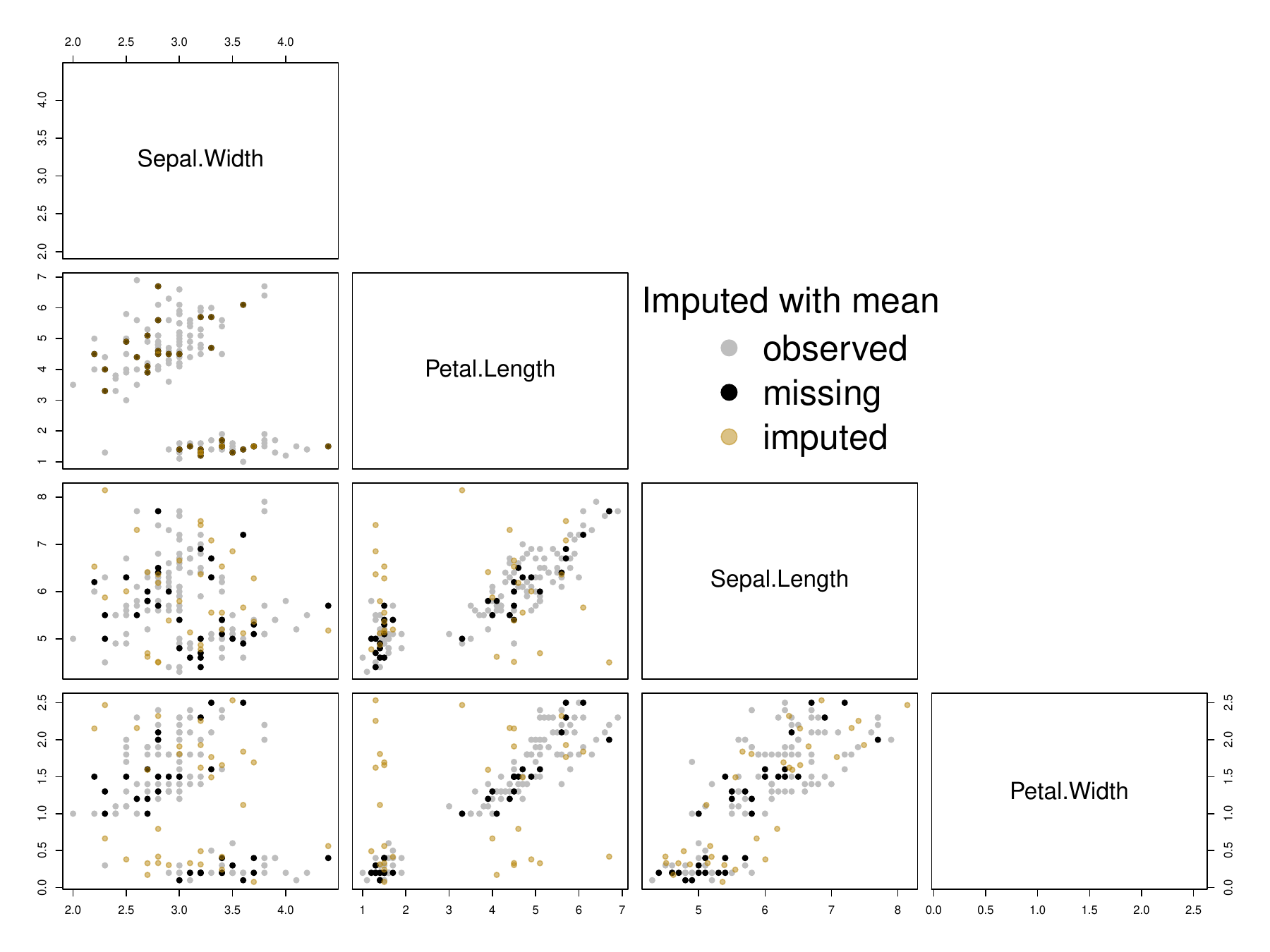}
		\caption{Imputed using the {\ttfamily mean} methodology.}
		\label{fig:pairwise_irisMAR_impute_mean}
	\end{subfigure}
	\hfill
	\begin{subfigure}[]{0.48\textwidth}
		\includegraphics[width=1.0\textwidth]{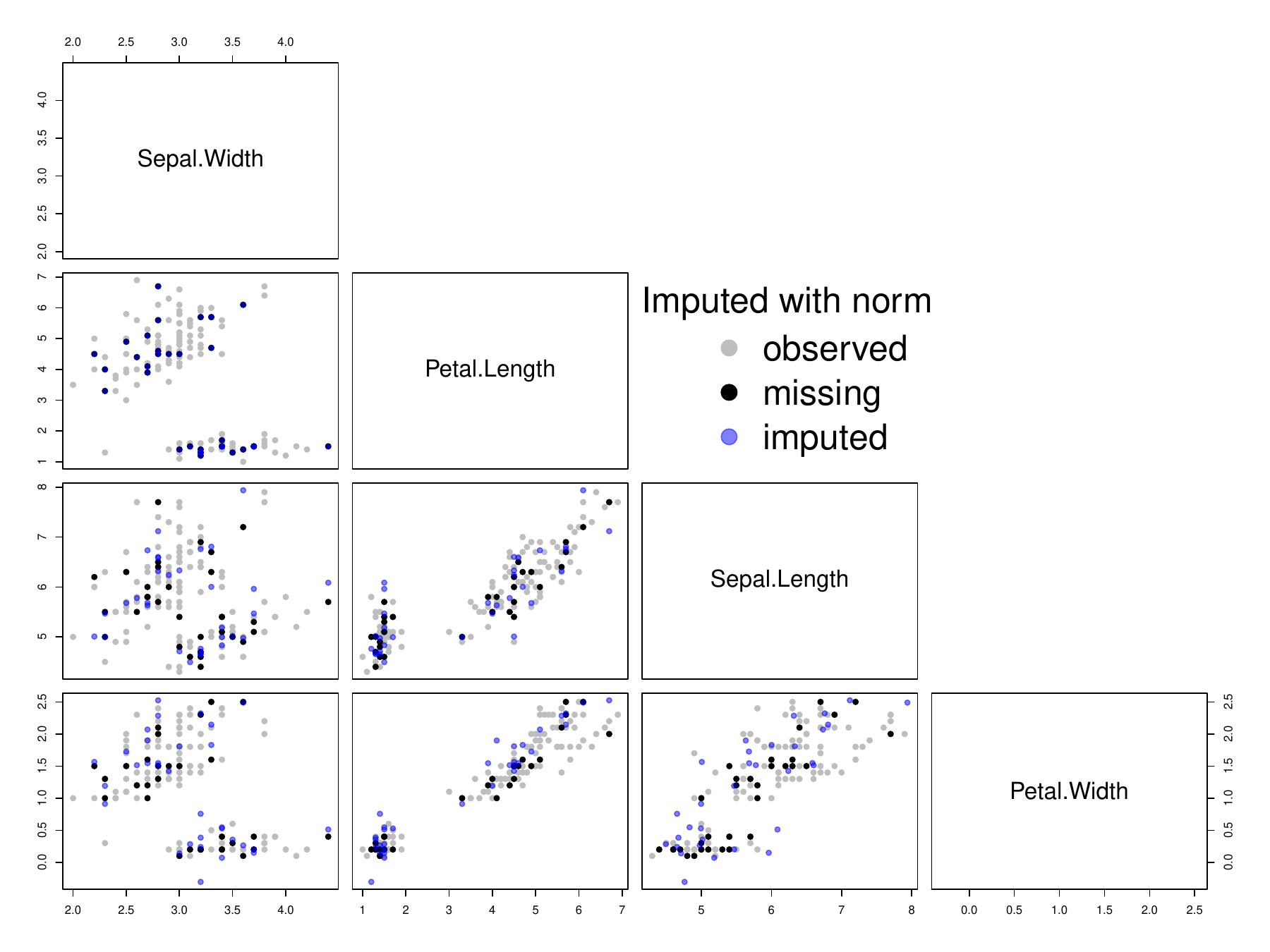}
		\caption{Imputed using the {\ttfamily norm} methodology.}
		\label{fig:pairwise_irisMAR_impute_norm}
	\end{subfigure}
	\hfill
	\begin{subfigure}[]{0.48\textwidth}
		\includegraphics[width=1.0\textwidth]{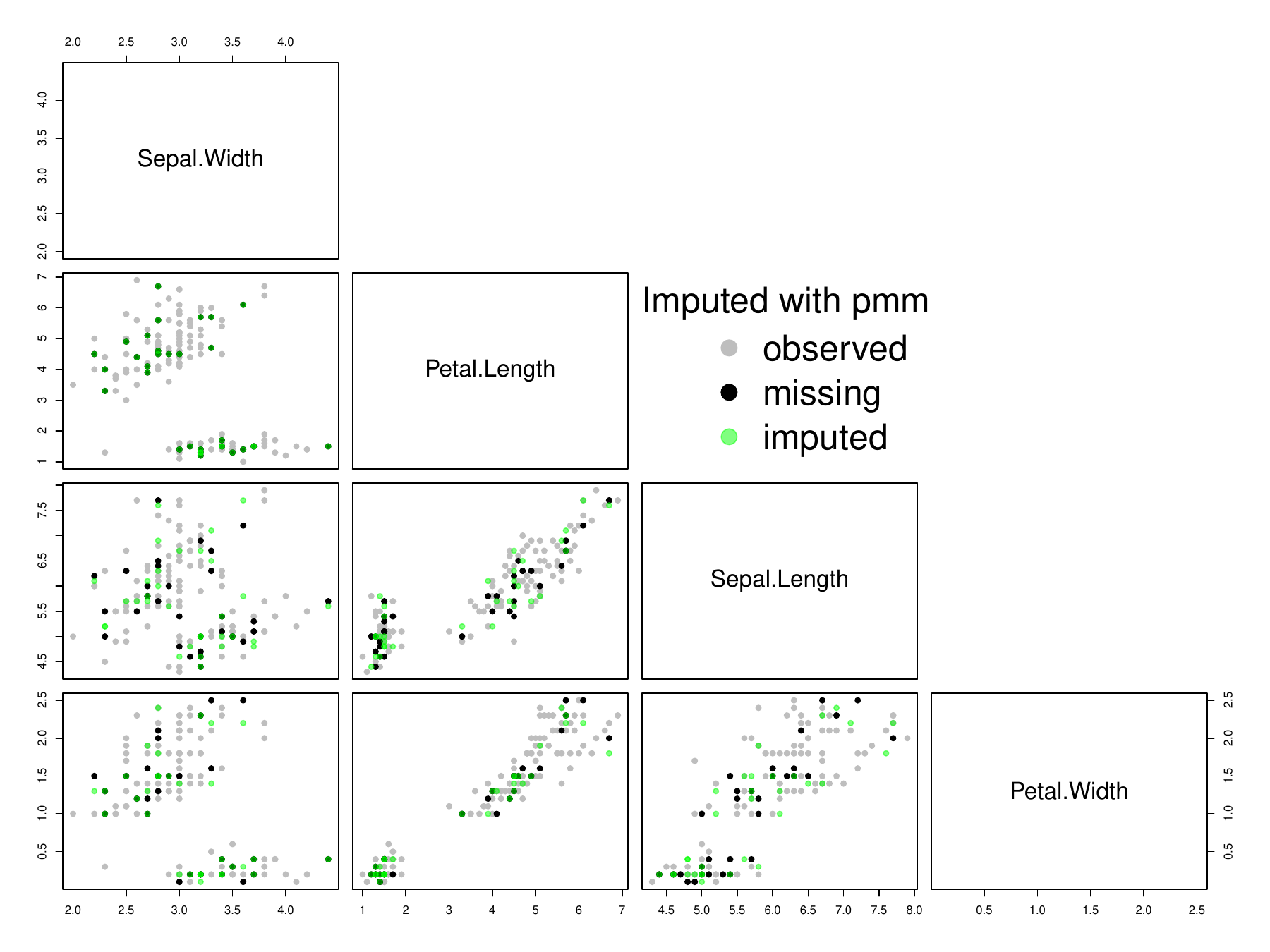}
		\caption{Imputed using the {\ttfamily pmm} methodology.}
		\label{fig:pairwise_irisMAR_impute_pmm}
	\end{subfigure}
	\caption[Pairwise plots of the imputed multivariate \textsc{Fisher's Iris Data}, missing data mechanism MAR.]{Pairwise plots of the imputed multivariate \textsc{Fisher's Iris Data} under three methodologies, namely {\ttfamily mean}, {\ttfamily norm}, and {\ttfamily pmm}. The missing data mechanism simulated in this case follows a Missing At Random (MAR) type.}
	\label{fig:pairwise_irisMAR_impute}
\end{figure}
We present pairwise plots for the different imputation methods: {\ttfamily mean}, {\ttfamily norm}, and {\ttfamily pmm} methods (see Figure~\ref{fig:pairwise_irisMAR_impute}). Each panel illustrates the observed, missing, and imputed values. In the upper left panel, the input variables are overlaid, as they are considered known information in the imputation procedure. The points corresponding to each methodology (a different color for each methods) appear superimposed on the black points. In the lower right panel, the output variables with missing information are overlaid, showing the imputations made by the respective models. The four panels in the lower left corner depict the intersection of the input and output variables, providing insights into the imputation process based on the respective projections.

From these graphs, we can observe that the {\ttfamily norm} and {\ttfamily pmm} procedures (see Figures~\ref{fig:pairwise_irisMAR_impute_norm} and \ref{fig:pairwise_irisMAR_impute_pmm}) exhibit visually similar results. The imputed data cover the regions where missing data was generated in comparable proportions. However, for the component corresponding to the setosa species, the {\ttfamily mean} procedure (see Figure \ref{fig:pairwise_irisMAR_impute_mean}) displays slightly greater dispersion in the imputed data.

The {\ttfamily mean} methodology imputes the components in different proportions compared to the simulated missing data, whereas the rest of the procedures exhibit the opposite behavior. The KL divergence values corresponding to the different procedures, presented in in \href{output-subdoc/JMVA_Main.pdf}{Main Material, Table}~\ref{B-cuadro:KL_irisMAR}, demonstrate that the {\ttfamily cwm} and {\ttfamily pmm} methods perform the best, followed by the {\ttfamily norm} method. Lastly, the {\ttfamily mean} procedure yields the poorest results.


\subsection{Iris database imputation with missing data MNAR mechanism}
\label{apend_c_sec_irisMNAR}

\begin{figure}[htb]
	\centering
	\begin{subfigure}[]{0.48\textwidth}
		\includegraphics[width=1.0\textwidth]{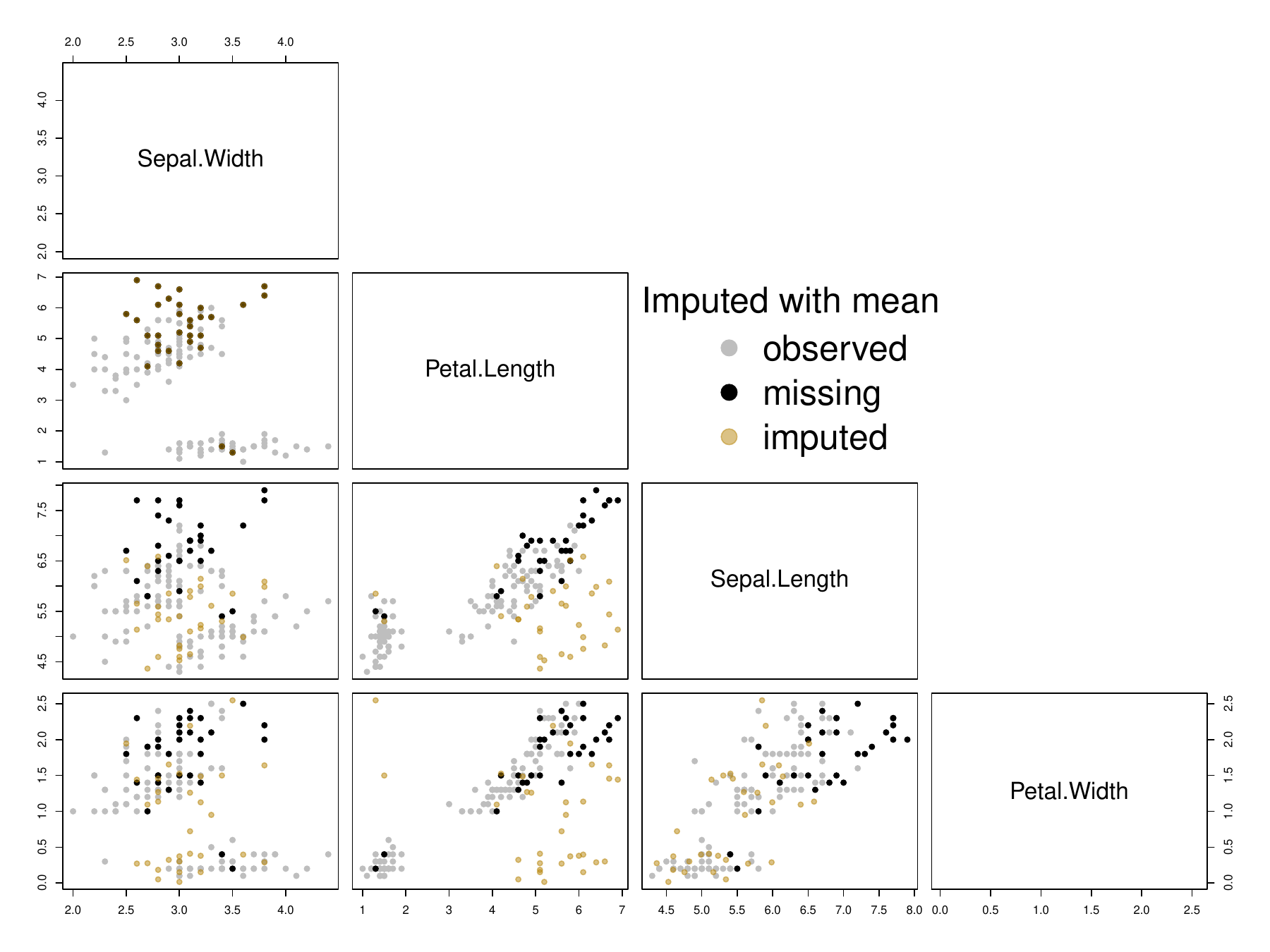}
		\caption{Imputed using the {\ttfamily mean} methodology.}
		\label{fig:pairwise_irisMNAR_impute_mean}
	\end{subfigure}
	\hfill
	\begin{subfigure}[]{0.48\textwidth}
		\includegraphics[width=1.0\textwidth]{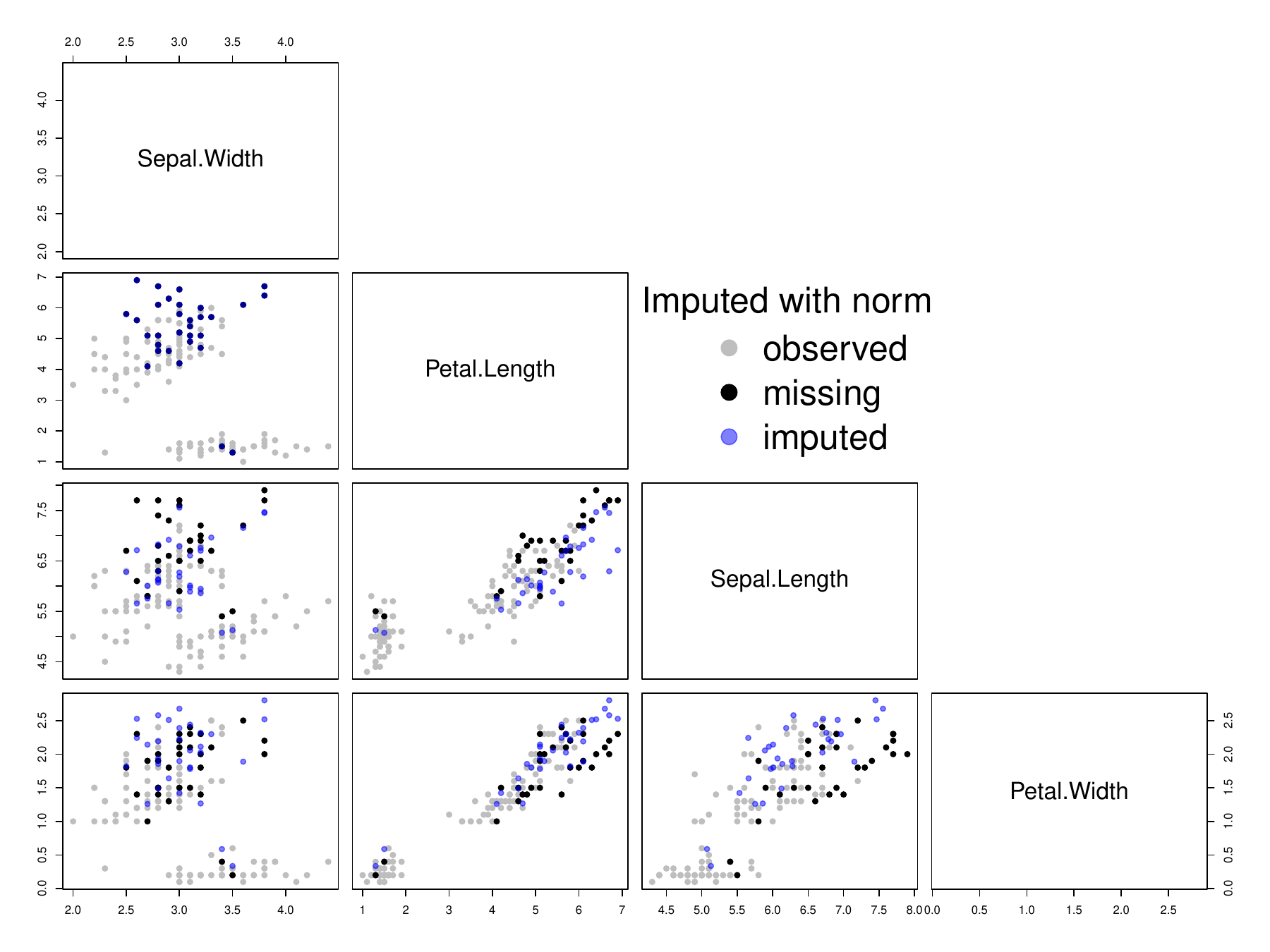}
		\caption{Imputed using the {\ttfamily norm} methodology.}
		\label{fig:pairwise_irisMNAR_impute_norm}
	\end{subfigure}
	\hfill
	\begin{subfigure}[]{0.48\textwidth}
		\includegraphics[width=1.0\textwidth]{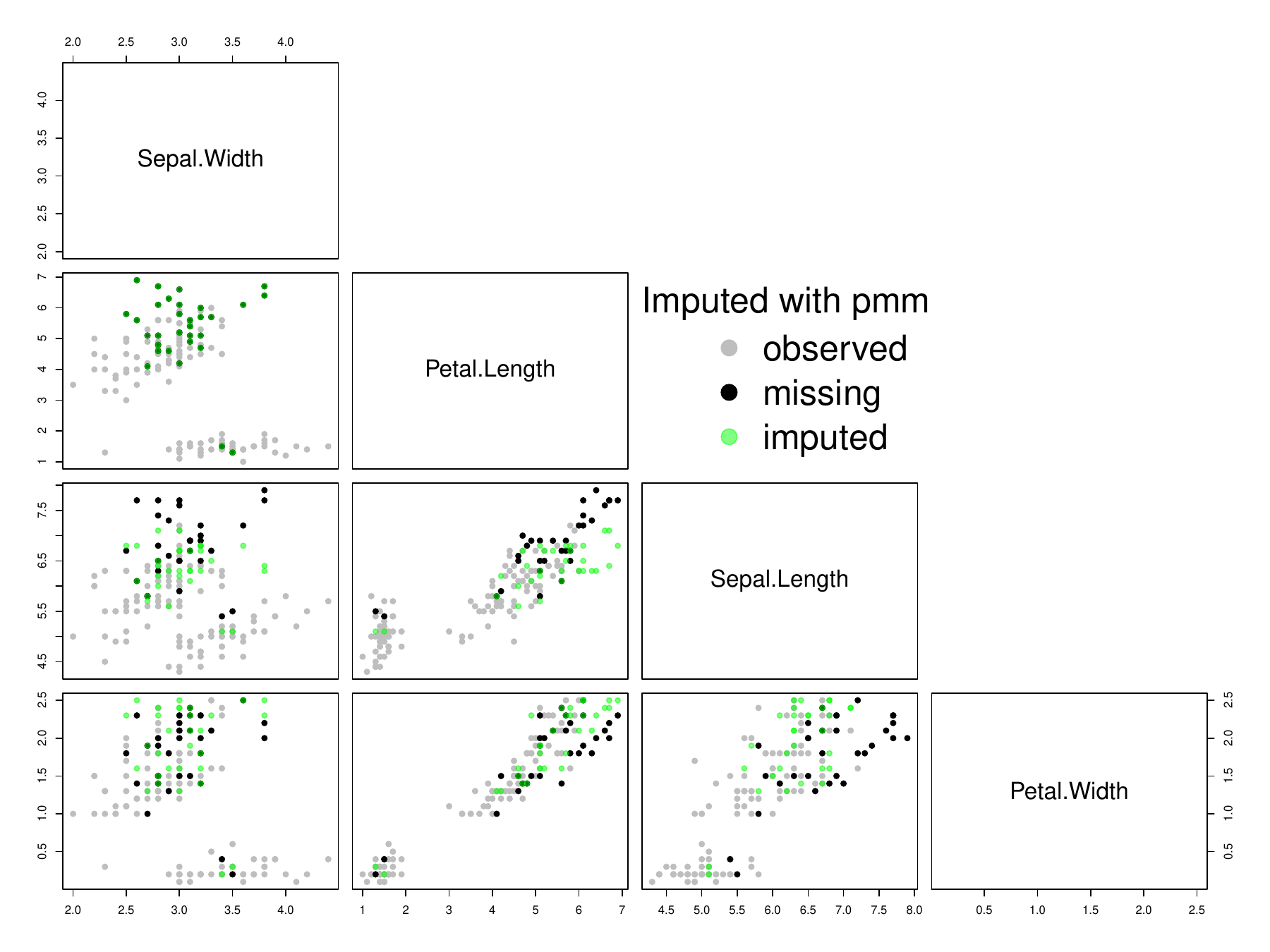}
		\caption{Imputed using the {\ttfamily pmm} methodology.}
		\label{fig:pairwise_irisMNAR_impute_pmm}
	\end{subfigure}
	\caption[Pairwise plots of the imputed multivariate \textsc{Fisher's Iris Data}, missing data mechanism MNAR.]{Pairwise plots of the imputed multivariate \textsc{Fisher's Iris Data} under three methodologies, namely {\ttfamily mean}, {\ttfamily norm}, and {\ttfamily pmm}. The missing data mechanism simulated in this case follows a Missing Not At Random (MNAR) type.}
	\label{fig:pairwise_irisMNAR_impute}
\end{figure}

Figure~\ref{fig:pairwise_irisMNAR_impute} displays pairwise plots depicting the imputed database using the {\ttfamily mean}, {\ttfamily norm}, and {\ttfamily pmm} methodologies. In this case, the missing data was simulated from an MNAR mechanism. The results were similar to those presented in the previous section for the three procedures. The {\ttfamily cwm} imputation process successfully covered all regions where missing data was generated, which does not appear to be the case with the  {\ttfamily norm}, and {\ttfamily pmm} methods. Particularly, in a specific part of the region occupied by the Virginica species, the observed data completely disappeared, leading to the failure of the two procedures to impute in this region. The performance of the different methodologies, as in the case of MAR, was evaluated based on the values of the KL divergence presented in \href{output-subdoc/JMVA_Main.pdf}{Main Material, Table}~\ref{B-cuadro:KL_irisMNAR}, to emphasize the advantage exhibited by our methodology, which outperformed the robust pmm procedure.


\section{Bayesian Inference for a Finite Mixture Model with Known Number of Components} 
\label{Bayesian_Inference} 

Consider a random sample $\bm{y}=(\bm{y_1},...,\bm{y_n})$ drawn from a finite mixture of distributions with density $p(\bm{y}|\bm{\theta})$, where the parameter $\bm{\theta} \in \Theta$ indexes the mixture model. The resulting finite mixture model can be expressed as:
\begin{equation*}
	p(\bm{y_i}|\bm{\vartheta})=\sum_{g=1}^{G}\alpha_g p(\bm{y_i}|\bm{\theta}_g),
\end{equation*}
Here, $\bm{\vartheta}=(\bm{\theta}_1,...,\bm{\theta}_G,\bm{\alpha})$ represents the parameter vector, where $\bm{\theta}_g$ denotes the parameters of the $g$-th component distribution, and $\alpha_g$ represents the mixing probability associated with the $g$-th component.
Below are the assumptions made about the model:
\begin{itemize}
	\item The number of components, $G$, is known.
	\item The parameters of the components, $\bm{\theta}_1,...,\bm{\theta}_G$, and the weight distribution, $\bm{\alpha} = (\alpha_1,...,\alpha_G)$, are unknown and need to be estimated.
	\item Another challenge is the allocation, which involves assigning each observation $\bm{y}_i$ to a specific component and inferring the hidden discrete indicators $\bm{Z} = (Z_1, ..., Z_n)$.
\end{itemize}
For the specific case where $p(\bm{y}_i|\bm{\theta}_g)$ follows a $p$-variate normal distribution, we have $\bm{\theta}_g=(\bm{\mu}_g,\Sigma_g)$, and the density function can be expressed as:

\begin{equation*}
	p(\bm{y}_i|\bm{\vartheta})=\sum_{g=1}^{G}  \phi_p(\bm{y}_i;\bm{\mu} _g,\Sigma_g)\alpha_g
\end{equation*}
where $\phi_p(\bm{y}_i;\bm{\mu}_g,\Sigma_g)$ denotes the $p$-variate normal probability density function given by:

\begin{equation*}
	\phi_p(\bm{y}_i;\bm{\mu}_g,\Sigma_g)= \frac{1}{(2\pi)^{p/2} |\Sigma_g|^{1/2}} \exp \left[-\frac{1}{2}(\bm{y}_i-\bm{\mu}_g)^\top \Sigma_g^{-1}(\bm{y}_i- \bm{\mu}_g) \right].
\end{equation*}
Here, $\bm{\mu}=(\bm{\mu}_1,...,\bm{\mu}_G)$, $\bm{\Sigma}=(\Sigma_1,...,\Sigma_G)$, and $\bm{\alpha}=(\alpha_1,...,\alpha_G)$. Therefore, the parameter is $\bm{\vartheta}=(\bm{\mu},\bm{\Sigma},\bm{\alpha})$.


Assume that the number of components $G$ and the parametric distribution family from which the component densities are derived are known, while the parameters of the components $\bm{\mu}=(\bm{\mu}_1,...,\bm{\mu}_G)$, $\bm{\Sigma}=(\Sigma_1,...,\Sigma_G)$, the mixing weights $\bm{\alpha}=(\alpha_1,...,\alpha_G)$, and the allocations $\bm{z}=(z_1,...,z_n)$ are unknown. According to Fruhwirth-Schnatter (2006), estimating the parameters of the mixture distribution is not a straightforward task. Even for standard distribution mixtures, such as the univariate normal distribution, there is no explicit estimation method, and in general some numerical method is needed for practical estimation.

The likelihood function of the mixture $p(\bm{y}|\bm{\vartheta})$, assuming no information about the allocation of $\bm{y_i}$ to any component, can be expressed as 
\begin{equation*}
	p(\bm{y}|\bm{\vartheta})=\prod_{i=1}^n p(\bm{y_i}|\bm{\vartheta})=\prod_{i=1}^n \sum_{g=1}^{G} \alpha_g p(\bm{y_i}|\bm{\theta_g}),
\end{equation*}
which is extended by a sum of $n^G$ individual terms to consider all possible combinations.
If a prior distribution for $\bm{\vartheta}$, denoted as $p(\bm{\vartheta})$, is available, the posterior distribution can be obtained as follows:
\begin{equation*}
		p(\bm{\vartheta}|\bm{y}) \propto p(\bm{y}|\bm{\vartheta})p(\bm{\vartheta}).
\end{equation*}

Unfortunately, there is no available natural conjugate prior for this probability function. This means that regardless of the chosen prior distribution $p(\bm{\vartheta})$, the resulting density does not belong to any tractable distribution family. Consequently, the Bayesian approach to finite mixture model was extremely challenging prior to the development of Markov Chain Monte Carlo (MCMC) methods.

According to \citet{fruhwirth2006finite}:

\begin{itemize}
	\item Bayesian estimation using MCMC methods is based on the work of \citet{dempster1977maximum}, who realized that a finite mixture model can always be expressed in terms of a problem of incompleteness when entering allocations as data absent.
	\item It is surprisingly simple to sample the posterior distribution using MCMC techniques presented in \citet{tanner1987calculation}, using the Data Augmentation technique.
\end{itemize}

Now, the focus shifts to the estimation of the parameter $\bm{\vartheta}$ when the allocations are known. In this case, for each data sample $\bm{y}=(\bm{y}_1,...,\bm{y}_n)$, we assume known allocations $\bm{Z}=(Z_1,..,Z_n)$, which provide the labels indicating the components of the mixture from which the observations were generated.

Then,
\begin{equation*} p(\bm{y},\bm{Z}|\bm{\vartheta})=p(\bm{y}|\bm{Z},\bm{\vartheta})p (\bm{Z}|\bm{\vartheta})=\prod_{i=1}^n p(\bm{y}_i|Z_i,\bm{\vartheta})p(Z_i|\bm{\vartheta})
\end{equation*}
where $p(\bm{y}_i|Z_i=g,\bm{\vartheta})=p(\bm{y}_i|\bm{\theta}_g)$ and $\Pr(Z_i=g|\bm{\vartheta})=\alpha_g$. 

Simplifying further, we have:
\begin{equation}
	\begin{split}
		p(\bm{y},\bm{Z}|\bm{\vartheta})&=\prod_{i=1}^n \prod_{g=1}^G [p(\bm{y}_i|\bm{\theta}_g)\alpha_g]^{\mathbbm{1}_g(Z_i)}\\
		&= \prod_{g=1}^G \left[\prod_{i:Z_i=g}p(\bm{y}_i|\bm{\theta}_g) \right] \left[\prod_{g=1}^G \alpha_g^{n_g(\bm{Z})}\right].
	\end{split}
	\label{exp:complete_likelihood}
\end{equation}

In equation (\ref{exp:complete_likelihood}), $n_g(\bm{Z})=\#\{Z_i=g\}$ counts the number of observations allocated to component $g$. This expression is known as the \textit{complete-data likelihood function}. The complete-data likelihood function is reduced to the product of $G+1$ factors, with the first $G$ factors corresponding to the parameters of the components $\bm{\theta}_g$, while the last factor depends only on the distribution of weights $\bm{\alpha}$.
In the Bayesian approach, the complete-data likelihood function $p(\bm{y}, \bm{Z}|\bm{\vartheta})$ is combined with the prior distribution $p(\bm{\vartheta})$ on the parameter $\bm{\vartheta}$ to obtain the posterior distribution of the complete-data,  given by:
	\begin{equation*}
	p(\bm{\vartheta}|\bm{y},\bm{Z}) \propto p(\bm{y},\bm{Z}|\bm{\vartheta})p(\bm{\vartheta}).
\end{equation*}

A structure similar to the complete-data likelihood function, with respect to the $G+1$ factors, can be assumed for the prior density $p(\bm{\vartheta})$, given by:
	\begin{equation*}
	p(\bm{\vartheta})= p(\bm{\alpha})\prod_{g=1}^Gp(\bm{\theta}_g).
\end{equation*}
Based on the form of this prior distribution, the complete-data posterior density for a finite mixture model, $p(\bm{\vartheta}|\bm{y},\bm{Z})$, can be factored conveniently in the same way, yielding $G+1$ factors:
	\begin{equation*}
	p(\bm{\vartheta}|\bm{y},\bm{Z}) = p(\bm{\alpha}|\bm{y},\bm{Z}) \prod_{g=1}^G p(\bm{\theta}_g|\bm{y},\bm{Z}), 
\end{equation*}
where,
\begin{equation}
	p(\bm{\theta}_g|\bm{y},\bm{Z}) \propto \left[\prod_{i:Z_i=g}p(\bm{y}_i|\bm{\theta}_g)\right]p(\bm{\theta}_g),
	\label{posterior_parametros}
\end{equation}
and,
\begin{equation}
	p(\bm{\alpha}|\bm{y},\bm{Z}) \propto \left[\prod_{g=1}^G \alpha_g^{n_g(\bm{Z})}\right]p(\bm{\alpha}).
		\label{posterior_pesos}
\end{equation}

In the context of the multivariate normal distribution, the parameter vector is denoted as $\bm{\theta}_g=(\bm{\mu}_g,\Sigma_g)$. Thus, equation (\ref{posterior_parametros}) can be expressed as:
\begin{equation*}
	\begin{split}
		p(\bm{\mu}_g,\Sigma_g|\bm{y},\bm{Z}) &\propto \left[\prod_{i:Z_i=g}p(\bm{y}_i|\bm{\mu}_g,\Sigma_g)\right]p(\bm{\mu}_g,\Sigma_g)\\
		&=\left[\prod_{i:Z_i=g}p(\bm{y}_i|\bm{\mu}_g,\Sigma_g)\right]p(\bm{\mu}_g|\Sigma_g)p(\Sigma_g).
	\end{split}
\end{equation*}
In this way, priori distributions are proposed as follows:
\begin{equation*}	
	\bm{\mu}_g|\Sigma_g \sim \mathcal{N}_p(\bm{\mu}_0,h^{-1}\Sigma_g) \quad \quad \text{and} \quad \quad \Sigma_g \sim \text{Inverse Wishart}(f,\bm{\Delta}), 	
\end{equation*}
where $\bm{\Delta}=\text{diag}(\delta_1,...,\delta_p)$ with $\delta_j \sim \text{Gamma}(a_{\delta},b_{\delta})$.

By utilizing the explicit formulation of the posterior distribution, we can express it as follows:
\begin{equation}	
	p(\bm{\mu}_g,\Sigma_g|\bm{y},\bm{Z}) \propto \left[\prod_{i:Z_i=g}p(\bm{y}_i|\bm{\mu}_g,\Sigma_g)\right] \\ \times p(\bm{\mu}_g|\Sigma_g,\bm{\mu}_0,h)
	p(\Sigma_g|f,\bm{\Delta}) p(\bm{\Delta}|a_{\delta},b_{\delta}),
	\label{eq:posterior}
\end{equation}
here,
\begin{equation}
	\begin{split}
\prod_{i:Z_i=g}p(\bm{y}_i|\bm{\mu}_g,\Sigma_g) &= \prod_{i:Z_i=g}\frac{1}{(2\pi)^{p/2} |\Sigma_g|^{1/2}} \exp \left[-\frac{1}{2}(\bm{y}_i-\bm{\mu}_g)^\top \Sigma_g^{-1}(\bm{y}_i- \bm{\mu}_g) \right]\\
&\propto |\Sigma_g|^{-n_g(\bm{Z})/2}\exp \left[-\frac{1}{2}\sum_{i:Z_i=g}(\bm{y}_i-\bm{\mu}_g)^\top \Sigma_g^{-1}(\bm{y}_i- \bm{\mu}_g) \right].
	\end{split}
	\label{eq:likelihood}
\end{equation}

From this formulation, it becomes possible to derive various marginal posterior distributions. By considering the expression within the exponential function in Equation (\ref{eq:likelihood}), it can be demonstrated that,
\begin{multline}
\sum_{i:Z_i=g}(\bm{y}_i-\bm{\mu}_g)^\top \Sigma_g^{-1}(\bm{y}_i- \bm{\mu}_g)=  n_g(\bm{Z}) (\bar{\bm{y}}_g-\bm{\mu}_g)^\top \Sigma_g^{-1}(\bar{\bm{y}}_g-\bm{\mu}_g) \\ + \text{tr} \left[\sum_{i:Z_i=g}(\bm{y}_i-\bar{\bm{y}}_g)^\top(\bm{y}_i-\bar{\bm{y}}_g) \times \Sigma_g^{-1}\right].
\label{rel_1}
\end{multline}

Based on the works of \citep{degroot2005optimal} and \citep{gelman2013bayesian}, we can manipulate the first three terms on the right-hand side of Equation (\ref{eq:posterior}) and consider the expression in (\ref{eq:likelihood}) as the likelihood function of $\bm{\mu}_g$ and $\Sigma_g$. Under this assumption, the joint probability density function (pdf) mentioned earlier satisfies the following relation:
The conditional distribution of $\bm{\mu}_g|\Sigma_g$ is a multivariate normal distribution with a mean vector $\bm{\mu}_0$ and a variance-covariance matrix that is proportional to $\Sigma_g$ with scaling factor $h^{-1}$. Additionally, the marginal distribution of $\Sigma_g$ follows an Inverse-Wishart distribution with $f$ degrees of freedom and a variance-covariance matrix $\Delta$, or equivalently, the precision matrix $\Sigma_g^{-1}$ follows a Wishart distribution with $f$ degrees of freedom and a precision matrix $\Delta^{-1}$:
\begin{multline*}
p(\bm{\mu}_g,\Sigma_g|\bm{\mu}_0, h, f, \bm{\Delta}) \propto |\Sigma_g^{-1}|^{1/2} \exp\left[-\frac{1}{2}(\bm{\mu}_g-\bm{\mu}_0)^\top h\Sigma_g^{-1}(\bm{\mu}_g - \bm{\mu}_0) \right] \\ \times |\Sigma_g^{-1}|^{(f-p-1)/2}  \exp\left[-\frac{1}{2} \text{tr} (\bm{\Delta}^{-1} \Sigma_g^{-1})\right].
\end{multline*}
It can be verified that,
\begin{multline*}
h(\bm{\mu}_g-\bm{\mu}_0)^\top \Sigma_g^{-1}(\bm{\mu}_g- \bm{\mu}_0) + n_g(\bm{Z})(\bar{\bm{y}}_g - \bm{\mu}_g)^\top \Sigma_g^{-1}(\bar{\bm{y}}_g - \bm{\mu}_g) = \\ [h + n_g(\bm{Z})] (\bm{\mu}_g-\bar{\bm{\mu}}_g)^\top \Sigma_g^{-1}(\bm{\mu}_g- \bar{\bm{\mu}}_g)  + \frac{hn_g(\bm{Z})}{h + n_g(\bm{Z})}(\bm{\mu}_0 -\bar{\bm{y}}_g)^\top \Sigma_g^{-1}(\bm{\mu}_0 - \bar{\bm{y}}_g),
\end{multline*}
where,
\begin{equation*}
\bar{\bm{\mu}}_g = \frac{ h\bm{\mu}_0 + n_g(\bm{Z}) \bar{\bm{y}}_g}{h+n_g(\bm{Z})} \quad \quad \text{and} \quad \quad \bar{\bm{y}}_g=\sum_{i:Z_i=g} \frac{\bm{y}_i}{n_g(\bm{Z})}.
\end{equation*}
Furthermore, the final term in above  equation  can be rewritten as follows:
\begin{equation}
\frac{hn_g(\bm{Z})}{h + n_g(\bm{Z})}(\bm{\mu}_0 -\bar{\bm{y}}_g)^\top \Sigma_g^{-1}(\bm{\mu}_0 - \bar{\bm{y}}_g) = \text{tr} \left[\frac{hn_g(\bm{Z})}{h + n_g(\bm{Z})}(\bm{\mu}_0 -\bar{\bm{y}}_g) (\bm{\mu}_0 - \bar{\bm{y}}_g)^\top \Sigma_g^{-1}\right].
\label{rel_2}
\end{equation}
From relation (\ref{eq:likelihood}) and from (\ref{rel_1}) to (\ref{rel_2}), we can obtain the following relation for the posterior joint p.d.f. of $\bm{\mu}_g$ and $\Sigma_g$:
\begin{equation}
\begin{split}
\left[\prod_{i:Z_i=g}p(\bm{y}_i|\bm{\mu}_g,\Sigma_g)\right] p(\bm{\mu}_g,\Sigma_g|\bm{\mu}_0, h, f, \bm{\Delta}) &= \left[\prod_{i:Z_i=g}p(\bm{y}_i|\bm{\mu}_g,\Sigma_g)\right] \\
&\qquad \times p(\bm{\mu}_g|\Sigma_g,\bm{\mu}_0,h)
	p(\Sigma_g|f,\bm{\Delta})\\\\
&\propto |\Sigma_g|^{-1/2} \exp\left[-\frac{h+n_g(\bm{Z})}{2}(\bm{\mu}_g-\bar{\bm{\mu}}_g)^\top \Sigma_g^{-1}(\bm{\mu}_g- \bar{\bm{\mu}}_g) \right]\\
&\qquad \times |\Sigma^{-1}_g|^{[f+n_g(\bm{Z})-p-1]/2}\exp\left[-\frac{1}{2} \text{tr} (\bm{\Delta}^{-1}_g \Sigma_g^{-1})\right].
\end{split}
\label{posterior_kernel}
\end{equation}

The function on the second line of Equation (\ref{posterior_kernel}), when considered as a function of $\bm{\mu}_g$, should be proportional to the conditional probability density function of $\bm{\mu}_g|\Sigma_g$, as the variable $\bm{\mu}_g$ does not appear in the third line. This function is proportional to the probability density function of a multivariate normal distribution, where the mean vector and variance-covariance matrix are specified as follows:
\begin{equation*}	
\bar{\bm{\mu}}_g = \frac{ h\bm{\mu}_0 + n_g(\bm{Z}) \bar{\bm{y}}_g}{h+n_g(\bm{Z})} \quad \quad \text{and} \quad \quad 	\bar{\Sigma}_g = \frac{1}{h+n_g(\bm{Z})}\Sigma_g.
\end{equation*}
As a result, the function on the third line of Equation (\ref{posterior_kernel}) is proportional to the marginal probability density function of $\Sigma_g^{-1}$, and it is also proportional to the pdf of the Wishart distribution with parameters $f_g$ and $\bm{\Delta}^{-1}_g$. Alternatively, we can express that $\Sigma_g$ is proportional to the pdf of an Inverse-Wishart distribution with the following parameters:
\begin{equation*}	
f_g=f+n_g \quad \quad \text{and} \quad \quad 	\bm{\Delta}_g=\bm{\Delta}+S_g+\frac{(\bar{\bm{y}}_g-\bm{\mu}_0)(\bar{\bm{y}}_g-\bm{\mu}_0)^\top}{1/h+1/n_g},
\end{equation*}
with $\bar{\bm{y}}_g=\sum_{\{i:Z_i=g\}} \bm{y}_i/n_g$ and $S_g=\sum_{\{i:z_i=g\}}(\bm{y}_i-\bar{\bm{y}}_g)(\bm{y}_i-\bar{\bm{y}}_g)^\top$.

Now, let's proceed to manipulate the last two terms in Equation~(\ref{eq:posterior}) as follows:
\begin{equation*}
	\begin{split}
		p(\delta_j|\bm{\Sigma}) 
		&\propto p(\bm{\Sigma}|f,\bm{\Delta})p(\delta_j|a_{\delta},b_{\delta})=\left[\prod_{g=1}^{G}p(\Sigma_g|f,\bm{\Delta})\right]p(\delta_j|a_{\delta},b_{\delta})\\
		&\propto  \prod_{g=1}^{G} \left\{|\bm{\Delta}|^{f/2} \exp\left[ -\frac{1}{2} \text{tr}\left(\bm{\Delta}\Sigma_g^{-1}\right) \right]\right\} \delta_j^{a_{\delta}-1}\exp\left(-b_{\delta}\delta_j \right)\\
		&= \left(\prod_{i=1}^{p}\delta_i^{Gf/2}\right) \exp\left[\sum_{g=1}^{G}\sum_{i=1}^{p}-\frac{1}{2}\delta_i(\Sigma_g)^{-1}_{i,i} \right] \delta_j^{a_{\delta}-1}\exp\left(-b_{\delta}\delta_j \right)\\
		&\propto \delta_j^{Gf/2} \exp\left[-\frac{\delta_j}{2}\sum_{g=1}^{G}(\Sigma_g)^{-1}_{j,j}\right] \delta_j^{a_{\delta}-1}\exp\left(-b_{\delta}\delta_j \right)\\
		&= \delta_j^{a_{\delta}+Gf/2-1}\exp\left\{ -\left[b_{\delta}+\frac{1}{2}\sum_{g=1}^{G}(\Sigma_g)^{-1}_{j,j} \right]\delta_j\right\}.
	\end{split}
\end{equation*}
From this analysis, we can conclude that
\begin{equation*}
	\delta_j|\bm{\Sigma} \sim \text{Gamma}(a^j_g,b^j_g),
\end{equation*}
where,
\begin{equation*}
	\begin{split}	
		a^j_g&=a_{\delta}+\frac{G(p+1)}{2},\\
		b^j_g&=b_{\delta}+\frac{1}{2}\sum_{g=1}^{G}(\Sigma_g)^{-1}_{j,j},
	\end{split}
\end{equation*}
for $j=1,...,p.$

When estimating the mixing probabilities, there are two approaches to consider. One option is to utilize the expression provided in (\ref{posterior_pesos}). Alternatively, one can employ the construction of Stick-breaking. Let's begin by considering the following:
	\begin{equation*}
	\begin{split}
		\alpha_g &= \nu_g \prod_{j<g} (1-\nu_j) \text{ para } g=1,...,G\\
		\nu_g &\sim \text{Beta}(1,\eta) \text{ para }g=1,...,G-1; \quad \nu_G=1\\
		\eta &\sim \text{Gamma}(a_{\eta},b_{\eta}).
	\end{split}
\end{equation*}
In order to derive the posterior distribution of $\nu_g|\bm{Z},\eta$, it is necessary to express $p(\bm{Z}|\bm{\alpha})$ as a function of $(\nu_1,.. .,\nu_G)$:
\begin{equation*}
	p(\bm{Z}|\bm{\alpha}) = \prod_{g=1}^G \alpha_g^{n_g(\bm{Z})}=\prod_{g=1}^G \left [\nu_g \prod_{j<g} (1-\nu_j)\right]^{n_g(\bm{Z})}.
\end{equation*}
Next,
\begin{equation*}
	p(\bm{Z}|\nu_1,...,\nu_G) =\prod_{g=1}^G \left[\nu_g^{n_g(\bm{Z})} \prod_{j<g } (1-\nu_j)^{n_g(\bm{Z})}\right],
\end{equation*}
now,
\begin{equation*}
	\begin{split}
		p(\nu_g|\bm{z},\eta) &\propto p(\bm{z}|\nu_1,...,\nu_G)p(\nu_g|\eta)\\
		&\propto\prod_{k=g}^G \left[\nu_k^{n_k(\bm{z})} \prod_{j<k} (1-\nu_j)^{n_j(\bm{z} )}\right] \times (1-\nu_g)^{\eta-1}\\
		&\propto \nu_g^{n_g(\bm{z})} \left[\prod_{j=g+1}^G (1-\nu_g)^{n_j(\bm{z})}\right] \times (1-\nu_g)^{\eta-1}\\
		&=\nu_g^{1+n_g(\bm{z})-1} (1-\nu_g)^{\eta+\sum_{j=g+1}^G{n_j(\bm{z})} -1}.
	\end{split}
\end{equation*}
So, for $g=1,...,G-1$, $\nu_g$ is updated from the full conditional distribution:
	\begin{equation*}
	\nu_g|\bm{z},\eta \sim \text{Beta}\left(1+n_g(\bm{z}),\eta + \sum_{j>g}n_j(\bm{z}) \right),
\end{equation*}
and $\nu_G=1$.
To update $\eta$, we need the likelihood:
\begin{equation*}
	\begin{split}
		p(\nu_1,...,\nu_G|\eta)&=\prod_{g=1}^{G-1} \frac{\Gamma(1+\eta)}{\Gamma(1)\Gamma(\eta)} (1-\nu_g)^{\eta-1}\\
		&= \eta^{G-1} \prod_{g=1}^{G-1}(1-\nu_g)^{\eta-1}.
	\end{split}
\end{equation*}
So, writing $\bm{\alpha}(\nu_1,...,\nu_G)$ noting that $\bm{\alpha}$ depends on $\nu_1,...,\nu_G$, we have that,
\begin{equation*}
	\begin{split}
		p[\eta|\bm{\alpha}(\nu_1,...,\nu_G)] &\propto p(\nu_1,...,\nu_G|\eta)p(\eta|a_{\eta},b_{\eta})\\
		&\propto \eta^{G-1} \prod_{g=1}^{G-1}(1-\nu_g)^{\eta-1} \eta^{a_{\eta}-1} e^{-b_{\eta}\eta} \\
		& =  \eta^{(a_{\eta}+G-1)-1} e^{ -b_{\eta}\eta + (\eta-1) \sum_{g=1}^{G-1}\log(1-\nu_g)}\\
		& \propto \eta^{(a_{\eta}+G-1)-1} e^{ -\eta \left[b_{\eta}-\sum_{g=1}^{G-1}\log(1-\nu_g)\right]}\\
		& = \eta^{(a_{\eta}+G-1)-1} e^{ -\eta \left(b_{\eta}- \log\alpha_G \right)}.
	\end{split}
\end{equation*}
From where, we update $\eta$ from the complete conditional:
\begin{equation*}
	\eta|\bm{\alpha} \sim \text{Gamma}\left(a_{\eta}+G-1,b_{\eta} - \log \alpha_G\right).
\end{equation*} 
Finally, for $i=1,...,n$, the complete conditional for $\bm{z}_i$ is:
	\begin{equation*}
	\bm{z}_i|\bm{\vartheta},\bm{y}_i \sim \text{Multinomial}(\alpha^*_{i,1},...,\alpha^*_{i,G})
\end{equation*} 
where,
\begin{equation*}
	\begin{split}
		\alpha^*_{i,g}&=P(Z_i=g|\bm{\vartheta},\bm{y}_i)\\
		&=\frac{p(\bm{y}_i|Z_i=g,\bm{\vartheta})P(Z_i=g|\bm{\vartheta})}{\sum_{k=1}^{G}p(\bm{y}_i|Z_i=k,\bm{\vartheta})P(Z_i=k|\bm{\vartheta})}\\
		&=\frac{p(\bm{y}_i|\bm{\theta}_g)\alpha_g}{\sum_{k=1}^{G}p(\bm{y}_i|\bm{\theta}_k)\alpha_k}\\
		&=\frac{\alpha_g \mathcal{N}(\bm{y}_i|\bm{\mu}_g,\Sigma_g)}{\sum_{k=1}^{G}\alpha_k \mathcal{N}(\bm{y}_i|\bm{\mu}_k,\Sigma_k)},
	\end{split}
\end{equation*}
for $g=1,...,G.$

\bibliography{bibliografia}